\documentclass[aps,prb,notitlepage,showkeys,showpacs,superscriptaddress,nofootinbib,longbibliography,preprintnumbers]{revtex4-2}

\usepackage[utf8]{inputenc}

\usepackage{here}
\usepackage{ulem}

\usepackage{url}

\usepackage[top=30truemm,bottom=30truemm,left=25truemm,right=25truemm]{geometry}

\renewcommand{\topfraction}{1.0}
\renewcommand{\bottomfraction}{1.0}
\renewcommand{\textfraction}{0.0}
\renewcommand{\floatpagefraction}{0.0}

\usepackage{spectralsequences}
\usepackage{tikz}

\usepackage{xcolor}
\newcommand{\delete}[1]{{\color{blue}\sout{#1}}}
\newcommand{\add}[1]{{\color{red} #1}}
\newcommand{\change}[2]{{\color{blue} \sout{#1}} {\color{red} #2}}
\newcommand{\new}[1]{{\color{olive}#1}}
\newcommand{\shu}[1]{\textcolor{blue}{(Shu:#1)}}
\newcommand{\ken}[1]{\textcolor{red}{(ken:#1)}}

\newcommand{\R}{\mathbb{R}}
\newcommand{\Z}{\mathbb{Z}}
\newcommand{\Zmod}[1]{\mathbb{Z}/#1\mathbb{Z}}
\newcommand{\zmod}[1]{\mathbb{Z}/#1\mathbb{Z}}
\newcommand{\rp}[1]{\mathbb{R}P^{#1}}
\newcommand{\tr}[1]{\mathrm{tr}\left(#1\right)}
\newcommand{\U}[1]{\mathrm{U}\left(#1\right)}
\newcommand{\RoverZ}{\mathbb{R}/\mathbb{Z}}
\newcommand{\PU}[1]{\mathrm{PU}\left(#1\right)}
\newcommand{\coho}[3]{\mathrm{H}^{#1}(#2;#3)}
\newcommand{\cohoZ}[2]{\mathrm{H}^{#1}(#2;\mathbb{Z})}
\newcommand{\cohoR}[2]{\mathrm{H}^{#1}(#2;\mathbb{R})}
\newcommand{\cohoU}[2]{\mathrm{H}^{#1}(#2;\U{1})}

\newcommand{\emb}[2]{\mathrm{Emb}(#1,#2)}

\newcommand{\abs}[1]{\left|#1\right|}

\usepackage[all]{xy}
\usetikzlibrary{cd}

\usepackage[colorlinks=true,linktoc=page,citecolor=red,linkcolor=blue]{hyperref}	
\usepackage[whole]{bxcjkjatype} 

\usepackage{slashed}

\usepackage{amsthm}
\usepackage{thmtools}
\usepackage{blindtext}
\usepackage{braket}

\declaretheoremstyle[
       shaded={bgcolor=\color{rgb}{0.9,0.9,0.9}}  
]{theorem}
\declaretheorem[style=theorem]{theorem}

\declaretheoremstyle[
       shaded={bgcolor=\color{rgb}{0.9,0.9,0.9}}
]{question}

\declaretheoremstyle[
       shaded={bgcolor=\color{rgb}{0.9,0.9,0.9}}  
]{remark}

\declaretheoremstyle[
       shaded={bgcolor=\color{rgb}{0.9,0.9,0.9}}  
]{proposition}
\declaretheorem[style=theorem]{proposition}

\declaretheoremstyle[
       shaded={bgcolor=\color{rgb}{0.9,0.9,0.9}}  
]{definition}
\declaretheorem[style=theorem]{definition}

\declaretheoremstyle[
       shaded={bgcolor=\color{rgb}{0.9,0.9,0.9}}  
]{assumption}

\declaretheoremstyle[
       shaded={bgcolor=\color{rgb}{0.9,0.9,0.9}}  
]{conjecture}

\declaretheoremstyle[
       shaded={bgcolor=\color{rgb}{0.9,0.9,0.9}}  
]{corrorary}

\declaretheoremstyle[
       shaded={bgcolor=\color{rgb}{0.9,0.9,0.9}}  
]{axiom}

\declaretheoremstyle[
       shaded={bgcolor=\color{rgb}{0.9,0.9,0.9}}  
]{lemma}
\declaretheorem[style=theorem]{lemma}

\def\N{{\mathbb{N}}}
\def\Z{{\mathbb{Z}}}
\def\Q{{\mathbb{Q}}}
\def\R{{\mathbb{R}}}
\def\C{{\mathbb{C}}}

\def\fp#1{{\frac{\partial}{\partial{#1}} }}
\def\fpp#1#2{{\frac{\partial{#1}}{\partial{#2}} }}

\def\g{{\mathfrak{g}}}
\def\oc{{\mathcal{U}}}
\def\D{{\mathcal{D}}}
\def\S{{\mathcal{S}}}
\def\F{{\mathcal{F}}}

\def\a{{\alpha}}
\def\b{{\beta}}
\def\e{{\epsilon}}
\def\o{{\omega}}
\def\r{{\rho}}
\def\ga{{\gamma}}
\def\th{{\theta}}
\def\si{{\sigma}}
\def\Ga{{\Gamma}}
\def\na{{\nabla}}
\def\io{{\iota}}
\def\id{{\bf 1}}
\def\t{{\tau}}
\def\De{{\Delta}}
\def\de{{\delta}}
\def\bu{{\bullet}}
\def\pa{{\partial}}

\def\>{{\geq }}
\def\<{{\leq }}

\def\un{\underline}
\def\os{\overset}
\def\vp{\varphi}
\def\ot{\otimes}

\def\ha{\frac{1}{2}}
\def\fd{{\mathcal D}}
\def\fo{{\mathcal F}}
\def\ph{{\phi}}
\def\of{{\omega_\fd}}
\def\pe{{\Lambda_\fd}}
\def\ua{{\underline{A}}}
\def\dc{{C_\fd}}
\def\M{{\mathcal M}}

\def\hol{{\text{Hol}}}
\def\tt{{T}}
\def\pg{{\text PGL}}

\usepackage{mathrsfs}
\usepackage{amsmath,amssymb}

\usepackage{esvect}

\usepackage{tikz}

\begin{document}

\title{Higher Berry Phase from Projected Entangled Pair States in (2+1) dimensions}

\author{Shuhei Ohyama}
\email{shuhei.ohyama@riken.jp}
\affiliation{
RIKEN Center for Emergent Matter Science, Wako, Saitama, 351-0198, Japan}
\author{Shinsei Ryu}
\email{shinseir@princeton.edu}
\affiliation{Department of Physics, Princeton University, Princeton, New Jersey 08544, USA}

\date{\today} 

\begin{abstract}
    We consider families of invertible many-body quantum states in $d$ spatial dimensions
    that are parameterized over some parameter space $X$. 
    The space of such families is expected to have topologically distinct sectors classified by the cohomology group $\cohoZ{d+2}{X}$.
    These topological sectors are distinguished
    by a topological invariant built from 
    a generalization of the Berry phase, called the higher Berry phase. 
    In the previous work \cite{OR23}, 
    we introduced a generalized inner product 
    for three one-dimensional many-body quantum states,
    (``triple inner product'').
    The higher Berry phase for one-dimensional invertible states can be introduced 
    through the triple inner product 
    and furthermore the topological invariant, 
    which takes its value in $\cohoZ{3}{X}$, can be extracted. 
    In this paper, we introduce an inner product of four two-dimensional invertible quantum many-body states.
    We use it to measure the topological nontriviality of parameterized families of 
    2d invertible states. In particular, we define a topological invariant 
    of such families that takes values in $\cohoZ{4}{X}$. 
    Our formalism uses projected entangled pair states (PEPS).
    We also construct a specific example of non-trivial parameterized families of 2d invertible states parameterized over $\mathbb{R}P^4$
    and demonstrate the use of our formula.
    Applications for symmetry-protected topological phases are also discussed.
\end{abstract}

\maketitle

\setcounter{tocdepth}{3}
\tableofcontents

\section{Introduction}

Invertible states are states realized as the unique ground states of gapped Hamiltonians. 
A notable subclass of invertible states is symmetry-protected topological (SPT) phases.
In the absence of symmetry conditions, these states are trivial and are continuously deformable to, e.g., a product state. 
On the other hand, when subjected to a set of symmetry conditions, 
they acquire a distinct topological nature and are separated from trivial states by a quantum phase transition 
\cite{Gu_2009,Pollmann_2012,
Chen_2011,Schuch_2011,
Chen_2013,Senthil_2015,
Chiu_2016,zeng2018quantum,CP-GSV21}.

More generally, one can consider a family of invertible states that depend continuously on some parameter(s) $x$,
$\{\ket{\psi(x)}: \text{invertible state}\left.\right|x\in X\}$,
where $X$ represents the parameter space. 
We refer to such a family as invertible states over $X$ 
or invertible states parameterized by $X$.
In simple cases, a parameterized family of invertible states 
can be viewed as an adiabatic process, that can be topologically non-trivial.
A classic example is the Thouless pump 
\cite{Thouless83},
which is a topologically non-trivial adiabatic process that transports quantized charges.

Considering families of invertible states is also useful in developing a systematic 
classification of invertible states.
According to Kitaev's conjecture, the classification of such states 
is provided by a generalized cohomology theory 
\cite{Kitaev13}.
From general and mathematical considerations using relativistic field theories, 
the classification of interacting invertible states is expected to be given 
by a generalized cohomology known as the Anderson dual of bordism theory.
Partial proofs supporting this idea have been provided \cite{FH21}.
However, it is worth noting that the classification of invertible states 
realized as lattice systems does not necessarily have to coincide 
with the classification expected from quantum field theory. 
A recent work 
\cite{Ogata22}
has pointed out the existence of classes that deviate from the classification 
by the Anderson dual of bordism theory in $2+1$-dimensional fermion systems.

This paper focuses on topologically non-trivial families of invertible states in dimensions higher than two.
In particular, we aim to construct a systematic method 
for computing the topological invariants of such families.
The space of families of $d$-dimensional invertible states parametrized by $X$ 
has topologically distinct sectors classified by $\cohoZ{d+2}{X}$. 
This class is distinguished by the higher Berry phase,
or more precisely the topological invariants built from the higher Berry phase,
which has been actively studied recently as a generalization of the regular
Berry phase -- see, e.g.,
\cite{
  KS20-1,KS20-2, Hsin_2020,
  Shiozaki21,
  Choi_2022,
  Cordova_2020a, Cordova_2020b,
  OTS23,
  beaudry2023homotopical,
  Wen_2023,
  qi2023charting,
  shiozaki2023higher,
artymowicz2023quantization}.
However, the explicit construction of the higher Berry phase,
in terms of many-body wave functions in particular,
is not known completely. 
In our previous paper
\cite{OR23}, 
we introduced a generalized inner product (or multi-wave function overlap) of three $1$d many-body states 
to measure the nontriviality and construct the invariant 
which takes its value in $\cohoZ{3}{X}$. 
In this paper,  we introduce an inner product of four $2$d states and 
try to measure the nontriviality in $\cohoZ{4}{X}$.

We do so by using tensor networks,
more specifically, projected entangled pair states (PEPS).
Tensor networks 
are a versatile framework
to study various problems in many-body quantum physics
\cite{CP-GSV21}.
Recent works also demonstrate that
tensor networks can capture higher geometric structures, such as the higher Berry phase, 
encoded in many-body quantum 
wavefunction
\cite{OSS22, OTS23, OR23, SHO23, QSWSPBH23, Spiegel23}.
In our previous work \cite{OR23}, 
we used the matrix product state (MPS) representation of
(1+1)-dimensional invertible states 
and introduced the triple inner product.
The triple inner product assigns 
a complex number for three invertible states (MPS) in (1+1) dimensions
and allows us to extract the higher Berry phase.
Furthermore, as a mathematical structure that underlies 
the higher Berry phase, 
we constructed and identified a gerbe on $X$ 
for a parameterized family of 
(1+1)-dimensional invertible states.
Here, a gerbe can be thought of as a  
generalization of a complex line bundle,
which underlies the description of the regular Berry phase. 
See Ref.\ \cite{QSWSPBH23} that also associates a gerbe to a parameterized family of MPS over $X$.
The subsequent work \cite{SHO23} 
provides a formulation for the higher Berry curvature using MPS,
and also demonstrated 
numerical evaluations of the Berry curvature using 
the density matrix renormalization group (DMRG).

%

Specifically, 
following Ref.\ \cite{OR23},
we introduce a generalized inner product 
for four PEPS, which we refer to as "quadruple inner product", 
and demonstrate that we can extract the higher Berry phase.
Additionally, we will argue (with some caveats) that the underlying mathematical structure is 
a 2-gerbe -- similarly to how gerbes are higher generalizations of complex line bundles, 
2-gerbes are generalizations of gerbes. 
While our primary focus will be on (2+1)-dimensional PEPS, we anticipate that generalizations to higher dimensions should be possible.

The paper is organized as follows: 
the remaining part of this section 
(Sec.\ \ref{Physical Intuition of the Invariant})  
provides an overview of the construction
of the triple inner product of MPS 
and outline its higher-dimensional generalization, 
serving as an instructive guide for 
the subsequent detailed construction. 
Section \ref{The Higher Berry Phase for (2+1)d Systems}
delves into a specific class of PEPS, 
the semi-injective PEPS, to implement the ideas sketched 
in Sec.\ \ref{Physical Intuition of the Invariant}.
Our formula for the quadruple inner product and the higher Berry phase, 
written in terms 
of various fixed point tensors, 
is summarized in Eq.\
\eqref{eq:quad_inner_product}
in Sec.\ \ref{preliminaries}.
Also,
at the end of Sec.\ \ref{preliminaries},
we outline two applications;
the topological invariants of SPT phases and 
the higher Berry phase
associated with parameterized invertible states.
Specifically, 
when applied to 
(2+1)d SPT phases protected by symmetry $G$, 
we can use the quadrupole inner product 
to extract the topological invariants of SPT phases 
valued in $\mathrm{H}^3(G, \mathrm{U}(1))$.
In the SPT context, our formulation is closely related to \cite{MGSC18}
and other prior works 
\cite{CLW11,
Else_2014,
Kawagoe_2021, kapustin2024anomalous, rubio2024classifying,Garre_Rubio_2023,rubio2024fractional}.
In the subsequent section
Sec.\ \ref{sec:proof},
we provide 
the rationale behind our formula
and prove various claims 
that are necessary to establish the formula. 
In Sec.\ 
\ref{Relation to the Group Cohomological Classification},
we discuss concrete models 
that realize non-trivial SPT phases.
For the example, we show that the higher Berry phase is given by
the group cohomology phase (3-cocycle) of the SPT phases.
In Sec.\ \ref{sec:Model parameterized over rp4},
we present a concrete model parametrized over 
$\mathbb{R}P^4$ and show,
using the quadruple inner product,
that it is characterized by a non-trivial topological invariant.
Finally in Sec.\ \ref{sec:2gerbe},
we identify a proper mathematical structure 
that describes parameterized (2+1)d invertible states
and the higher Berry phase.
We conclude in Sec.\ \ref{sec:summary}.

\subsection{Physical Intuition
}
\label{Physical Intuition of the Invariant}

\begin{figure}[t]
\includegraphics[scale=1.4]{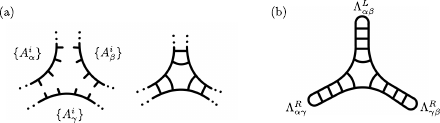}
\caption{
\label{fig:tri_inner}
(a) 
Taking the triple inner product of three matrix product states 
$\{A^i_{\alpha}\}$,
$\{A^i_{\beta}\}$
and
$\{A^i_{\beta}\}$.
(b) The triple inner product in which the virtual legs 
at spatial infinities are contracted 
with the left and right eigenvectors 
(fixed points)
$\Lambda^L_{\alpha\beta}$,
$\Lambda^R_{\alpha\gamma}$
and
$\Lambda^R_{\gamma\beta}$
of the mixed transfer matrices.
}
\end{figure}

Previous works studied the higher Berry phase for (1+1)d many-body states
using MPS. 
In this paper, we consider (2+1)d many-body states and their higher Berry phase using PEPS.
The motivation and intuition behind our construction 
of the quadruple inner product and the higher Berry phase
can be best described by reviewing their (1+1)d counterparts.

To set the stage, we consider a family of MPS 
with MPS matrices $\{A^i(x)\}_{i=1,\ldots,\mathsf{D}}$ parameterized by $x\in X$,
where $X$ is the parameter space.
(Here, $i=1, \ldots, \mathsf{D}$ represents physical degrees of freedom and 
indices for the auxiliary (internal) are implicit here 
-- see Ref.\ \cite{OR23} for our notations.)
We assume that our MPS tensors are normal, 
$\sum_i A^i A^{i\dagger} =1$, and work with the right canonical gauge.
As shown in Ref.\ \cite{OR23},
the higher Berry phase for the family, and the topological invariant (the Dixmier-Douady class)
can be extracted from a certain overlap of three MPS as follows. 
We consider an open covering of $X$, $\{U_{\alpha}\}$.
On each patch, we have MPS, $\{A^i_{\alpha}(x)\}$. 
When two patches overlap, two MPS 
$\{A^i_{\alpha}(x)\}$ 
and
$\{A^i_{\beta}(x)\}$
represent the same physical state.
The fundamental theorem of injective MPS 
then asserts that these two MPS are related
to each other by a gauge transformation,
$A^i_{\alpha}(x) = g_{\alpha\beta}(x)A^i_{\beta}(x)g^{\dag}_{\alpha\beta}(x)$. 
Now, at the triple intersection,
$U_{\alpha\beta\gamma}:= 
U_{\alpha}\cap U_{\beta} \cap U_{\gamma}$, we consider three MPS, 
$\{A^i_{\alpha}(x)\}$, $\{A^i_{\beta}(x)\}$ and $\{A^i_{\gamma}(x)\}$,
representing the same physical state,
and consider the triple inner product,
defined schematically in Fig.\ \ref{fig:tri_inner}(a). 
To make the definition precise, however, we need to address two issues.
First, we need to specify how virtual bonds at the ends of the chain are contracted, 
i.e., there are remaining legs at the interface between the two states. Otherwise, 
Fig.\ \ref{fig:tri_inner} represents an operator rather than a number. 
Second,  
we must address a conceptual concern: as we are interested in intrinsic quantities to many-body quantum states, the higher generalization of the inner product and the Berry phase must 
make sense in the thermodynamic limit. 
These can be addressed simultaneously by contracting the virtual legs 
using the left or right fixed point 
of the mixed transfer matrix [Fig.\ \ref{fig:tri_inner}(b)]. 
Using the eigenvalue equation, 
we can regard this diagram as the overlap of three infinite systems.
This quantity is called the triple inner product, and using it, we can compute the higher Berry phase for (1+1)d systems. 
Recalling that the regular Berry phase for (0+1)d quantum systems 
can be computed from the inner product of two states, 
the above generalization for (1+1)d systems can be considered extremely natural.

\begin{figure}[t]
  \includegraphics[scale=1.7]{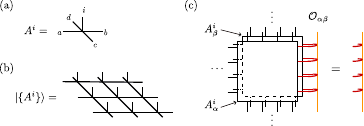}
  \caption{
    \label{4leg_tensor}
    (a) A PEPS tensor with one physical and four auxiliary bonds.
    (b) The PEPS representation of a (2+1)d invertible state.
    (c) A transition MPO is an eigen MPO for the mixed transfer matrix made of two PEPS
    $A^{i}_{\alpha}$ and $A^i_{\beta}$.
    Here, the two PEPS are contracted by their physical bonds (not shown explicitly).
  }
\end{figure}

Following this line of thought,
we expect that 
the higher Berry phase in general dimensions
can be computed by considering $d$-fold inner products of ($d+1$)d quantum states. 
Let us now turn to the (2+1)d case and consider 
PEPS $\{A^i(x)\}$ parameterized over $X$.
Throughout the paper, we consider PEPS tensors with one physical bond and four auxiliary bonds (Fig.\ \ref{4leg_tensor}). 
Paralleling the (1+1)d case, 
it is natural to consider the overlap 
of four physically equivalent states in different gauges 
at an intersection 
$U_{\alpha\beta\gamma\delta}:=
U_{\alpha}\cap U_{\beta}\cap U_{\gamma}
\cap U_{\delta}$
as in Fig.\ \ref{fig:quad_inner}.
There, as in the (1+1)d case, the physical indices of the PEPS 
are contracted "partially" with others.
I.e., the physical indices of $\{A^i_{\alpha}\}$ in the first quadrant  
are contracted with the 
physical indices of $\{A^i_{\beta}\}$ in the third quadrant, 
etc.
As in the (1+1)d case,
there are remaining legs (bonds) at the boundaries of the PEPS networks.  
They need to be properly contracted such that the inner product 
is well-defined for the many-body states (PEPS) 
in the thermodynamic limit.

\begin{figure}[t]
  \centering
  \includegraphics[scale=0.09]{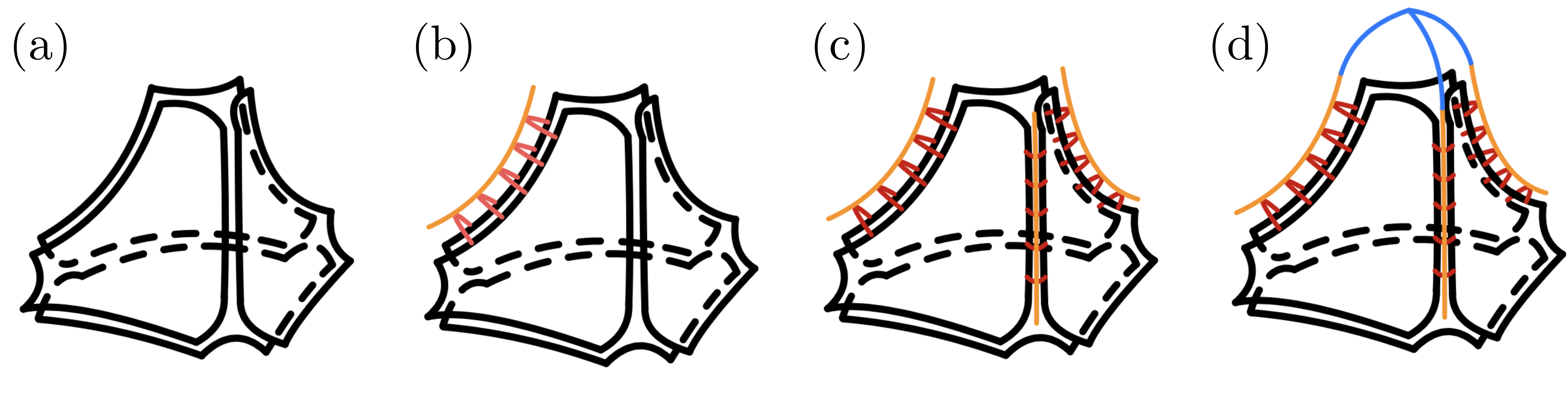}
  \caption{ \label{fig:quad_inner}
    The quadruple inner product for four PEPSs 
    (here represented as 2d "sheets")
    defined by gluing four PEPSs as illustrated above.
    (a) Four PEPSs from different patches at the quad intersection 
    $U_{\alpha\beta\gamma\delta}$.
    The physical indices (not shown explicitly) are contracted such that 
    the four PEPSs are glued with each other. 
    (b) For a pair of PEPS, we can define an eigen MPO ${\cal O}_{\alpha\beta}$,
    that can be used to contract indices of PEPS as in (c).
    (d) The three-leg tensor can contract indices at the intersection
    where three MPO meet.
  }
\end{figure}


First, let us consider the dangling bonds  
at the interface where two PEPS networks meet.
Much the same way with the (1+1)d case, 
we consider the mixed transfer matrix at the intersection 
$U_{\alpha}\cap U_{\beta}=U_{\alpha\beta}$
as depicted in Fig.\ \ref{4leg_tensor}(c).
Here, we assume the transfer matrix has a unique maximal eigenvalue $1$ and 
the corresponding eigenvector is given as a matrix product operator (MPO). 
Moreover, we assume that MPO tensors do not depend on the number 
of sites in the interface direction
\footnote{These are naive guesses intended to 
give an explanation of the idea for constructing invariants, 
and we should mention that they are too strong assumptions. 
When we actually construct the invariants in Sec.\  
\ref{The Higher Berry Phase for $2+1$-d Systems}, 
we will relax the eigenvalue equations related 
to the transfer matrix and avoid these difficulties.}.
This MPO is a higher-dimensional analog of the transition functions, denoted as $g_{\alpha\beta}(x)$, 
in (1+1)d case. 
We will call it a transition MPO and denote it as ${\cal O}_{\alpha\beta}$. 
By using the property of the fixed point, 
we can implement the thermodynamic limit in the bulk direction, 
i.e., the direction orthogonal to the interface.

In addition to the interface (the orange part) between the two PEPS networks, 
it is evident that there is "an interface of the interface" 
(the blue part
in Fig.\ \ref{fig:quad_inner}). 
At this point, the three virtual legs of the MPO,
${\cal O}_{\alpha\beta}$, ${\cal O}_{\beta\gamma}$ and ${\cal O}_{\alpha\gamma}$,
say, 
meet with each other. 
We need an appropriate tensor so that we can contract the legs. 
We thus postulate the existence of a three-leg tensor 
$\Lambda_{\alpha\beta\gamma}$ for triple intersection 
$U_{\alpha\beta\gamma}$.
Furthermore, similarly to the transition MPO ${\cal O}_{\alpha\beta}$ 
that satisfies the fixed point condition, 
it is natural to require the three-leg tensor to satisfy 
a fixed point condition such that we can take the thermodynamic
limit in the directions along the double interfaces between 
two PEPS networks.
Specifically, we require 
the "higher eigen equation" as in 
Fig.\  \ref{fig:highereigeneq}(a).
We will adopt this as the defining condition of the three-leg tensor.
%
%
%
\begin{figure}[t]
\includegraphics[scale=0.25]{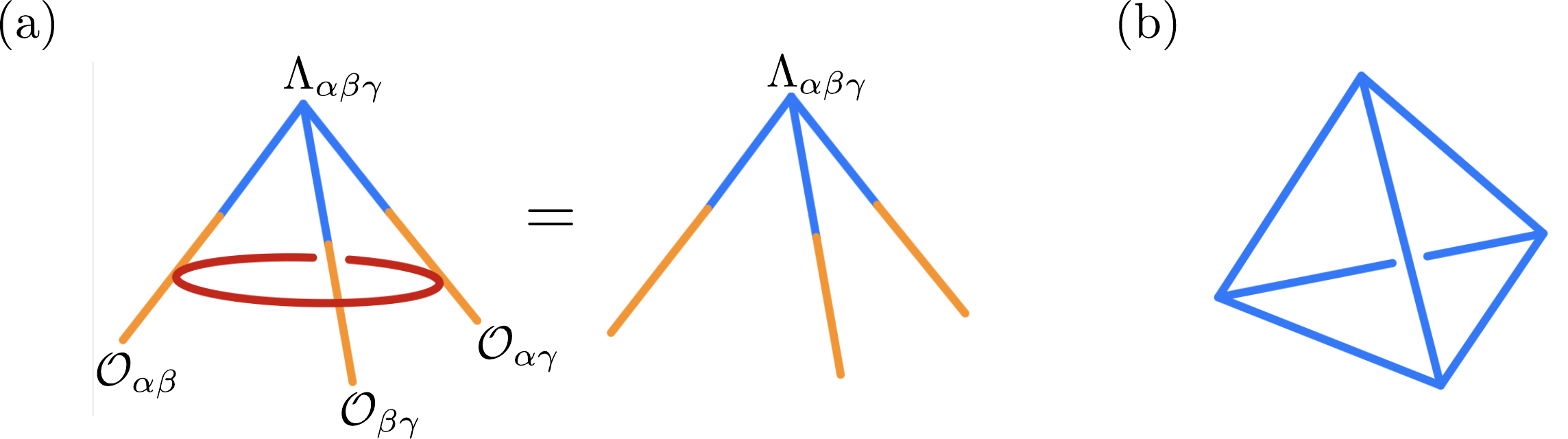}
  \caption{(a) The fixed point condition for the three-leg tensor 
    and (b) the higher Berry phase.
    \label{fig:highereigeneq}}
\end{figure}
(More generically, the eigenvalue 
can be a complex number of modulus one. 
The phase basically is the Dixmier-Douady class that 
appeared above when we considered three MPS. For now,
we assume such "lower invariant" vanishes for simplicity.)

By using these tensors, we can contract all the remaining legs 
in Fig.\ \ref{fig:quad_inner}. 
Furthermore, we can use the eigenequations (fixed point conditions)
to annihilate physical sites, 
and then reduce the diagram Fig.\ \ref{fig:quad_inner}(d) 
to the tetrahedral diagram made of four $3$-leg tensors as  
Fig.\ \ref{fig:highereigeneq}(b).
This completes the sketch of 
the construction of the higher Berry phase 
(and the topological invariant built from the higher Berry phase)
from PEPS.
Note however that in the above construction,
we assume the existence of various tensors. 
It is necessary to clarify these points. 
In the next section, we will describe our construction 
using semi-injective PEPS \cite{MGSC18} 
to concretely implement this idea.
\footnote{
In this section and 
in this figure, 
we denote the three-leg tensor 
simply by $\Lambda_{\alpha\beta\gamma}$.
However, as we will 
see in the later sections, 
we need various types of three-leg tensors
for the precise construction. 
}

Finally, to see why these are natural objects to consider,
and how these objects need to be combined,  
it is good to compare with the (1+1)d case.
In (1+1)d, we assign an MPS $|\{A^i_{\alpha}\}\rangle$ for $U_{\alpha}$, 
$|\{\Lambda_{\alpha\beta}\}\rangle$ for $U_{\alpha\beta}$,
and 
a $\mathrm{U}(1)$ phase $c_{\alpha\beta\gamma}$
for $U_{\alpha\beta\gamma}$, respectively. 
Here, $|\{\Lambda_{\alpha\beta}\}\rangle$ or simply 
$\Lambda_{\alpha\beta}$ represents 
"mixed gauge" MPS. 
In mathematical terms,
$|\{\Lambda_{\alpha\beta}\}\rangle$ is a 
line bundle (one-dimensional Hilbert space) assigned 
for the intersection $U_{\alpha\beta}$.
The $\mathrm{U}(1)$ phase $c_{\alpha\beta\gamma}$ 
is the higher Berry phase.
In (2+1)d, we need to consider a "higher" version.
We assign a 
PEPS $|\{A^i_{\alpha}\}\rangle$ for $U_{\alpha}$, 
an MPO ${\cal O}_{\alpha\beta}$ for $U_{\alpha\beta}$,
a three-leg tensor $\Lambda_{\alpha\beta\gamma}$ for $U_{\alpha\beta\gamma}$,
and 
a $\mathrm{U}(1)$ phase $c_{\alpha\beta\gamma\delta}$
for $U_{\alpha\beta\gamma\delta}$.
Roughly, the three-leg tensor 
$\Lambda_{\alpha\beta\gamma}$ 
is an analog of $\Lambda_{\alpha\beta}$
in (2+1)d, and the $\mathrm{U}(1)$ phase $c_{\alpha\beta\gamma\delta}$ 
is the higher Berry phase we are looking for.
Note that at the double intersection $U_{\alpha\beta}$
we assign an MPO, which can be thought of as an MPS over $U_{\alpha\beta}$, 
and if so, we are essentially putting a gerbe at the intersection  
$U_{\alpha\beta}$. This is a generalization of a gerbe, 
just like a gerbe is a generalization of a line bundle. 
We will discuss the mathematical structure 
behind the construction of our invariant in Sec.\ \ref{sec:2gerbe}.

\section{The Higher Berry Phase for (2+1)d Systems}
\label{The Higher Berry Phase for (2+1)d Systems}

\subsection{Preliminaries}
\label{preliminaries}

In this section, we will materialize the idea 
sketched above by using 
a class of PEPS, the so-called semi-injective PEPS
\cite{MGSC18}.
We begin by reviewing necessary backgrounds and notations;
MPO, semi-injective PEPS,
and transfer matrices.

\subsubsection{MPO}

Matrix product operators (MPO)
are central objects in our discussions.
They are defined by a set of matrices (tensors)
$\{B^{ij}\}_{i,j=1}^\mathsf{D}$.
Here, $i,j=1,\cdots, \mathsf{D}$ are ``physical'' indices   
representing, e.g., physical spin degrees of freedom. 
In addition, the tensors $\{B^{ij}\}_{i,j=1}^\mathsf{D}$
carry ``virtual'' indices (not shown explicitly).
The dimension of the virtual space (``bond dimension'')
is denoted by $\chi$. 
We use
$\mathcal{O}_L[\{B^{ij}\}]$ to denote
the MPO of length $L$ that is generated by
the MPO tensors $\{B^{ij}\}_{i,j=1}^{\mathsf{D}}$. 
To be concrete, we consider periodic boundary conditions,
although it is possible to discuss other boundary conditions.
When there is no confusion, we shall call both
$\{B^{ij}\}_{i,j=1}^\mathsf{D}$
and ${\cal O}_L[\{B^{ij}\}]$
MPO.
We recall that MPS (MPO) matrices are said to be normal when
the corresponding transfer matrix has a unique largest eigenvalue. 
We also note that any injective matrices are proportional 
to normal matrices.
Conversely, any normal matrices become injective after blocking.

\subsubsection{Semi-injective PEPS}\label{sec:semi}

Semi-injective PEPS,
introduced in Ref.\ \cite{MGSC18},
is a type of PEPS that can represent a wide class of (2+1)d invertible states
including ground states of symmetry-protected topological (SPT) phases. 
To introduce semi-injective PEPS,
let's consider a square lattice, where each unit cell consists of four sites
[Fig.\ \ref{fig:semiinj}(a)].
At each site within a unit cell, 
we assign a physical quantum mechanical degree of freedom ("spin") with local Hilbert space.
The dimensions of the local Hilbert space 
can be different from site to site. 
The total Hilbert space is constructed 
from the local Hilbert spaces, 
by arranging the local degrees of freedom in a translationally symmetric manner.
With this setup,
a semi-injective PEPS is defined by two data 
$(\ket{\phi},O)$. Here, $\ket{\phi}$ is a state defined on each 
$4$-site system as in Fig.\ \ref{fig:semiinj}(c) with full rank reduced density matrix. 
$O$ is an invertible operator that acts on the four corners of four adjacent
squares,
as in Fig.\ \ref{fig:semiinj}(b).
By using these data, a semi-injective PEPS generated by ($|\phi\rangle$, $O$) 
is defined as in
Fig.\ \ref{fig:semiinj}(d).
\begin{figure}[t]
\includegraphics[scale=1.5]{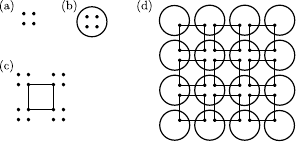}
\caption{
\label{fig:semiinj}
The setup for semi-injective PEPS.
(a) The unit cell;
(b) the invertible operator 
${O}$;
(c) the full-rank state $|\phi\rangle$;
(c) the semi-injective PEPS
$(|\phi\rangle, {O})$.
}
\end{figure}

What is important for us is 
the fundamental theorem 
of semi-injective PEPS.
Let us now consider two semi-injective PEPS,
$(|\phi_{\alpha}\rangle,O_{\alpha})$
and 
$(|\phi_{\beta}\rangle,O_{\beta})$, 
that generate
the same physical state.
Then, the fundamental theorem of semi-injective PEPS
states that
they are related by an invertible MPO acting on their auxiliary indices
\cite{MGSC18}.
We shall call the MPO a transition MPO. 
(Recall also Fig.\ \ref{4leg_tensor} (c).)
In the following, we denote the MPO matrices (tensors) that
generate the transition MPO as
$\{B^{ij}_{\alpha\beta}\}_{i,j=1}^{\mathsf{D}}$.
Also, 
$\mathcal{O}_L[\{B^{ij}_{\alpha\beta}\}]$ denotes
the corresponding MPO of length $L$ that is generated by
the MPO tensors $\{B^{ij}_{\alpha\beta}\}_{i,j=1}^{\mathsf{D}}$. 
Furthermore, without losing generality,
the fundamental theorem asserts that we can take 
$\{B^{ij}_{\alpha\beta}\}_{i,j=1}^{\mathsf{D}}$ as normal matrices (tensors).

\subsubsection{Key Propositions and the Unitarity Assumption}\label{sec:prop}\label{sec:assumption}
In the construction of the quadruple inner product, various quantities are realized as fixed points of the transfer matrix. The following claim proves useful in this context\footnote{We would like to thank Norbert Schuch for teaching us the proposition and its proof.}.

\begin{proposition}\label{thm:fp}\;\\
\\
Let $\{B_{1}^{ij}\}_{i,j=1}^{\mathsf{D}},...,\{B_{k}^{ij}\}_{i,j=1}^{\mathsf{D}}$ be normal MPO matrices and assume that the composition of all the MPO is proportional to the identity MPO. That is, there exists an invertible complex number $c\in\mathbb{C}^{\times}$ such that 
\begin{eqnarray}\label{eq:cycle}
    \mathcal{O}_{L}[\{\sum_{j_2,...,j_k}B_{1}^{j_1 j_2}\otimes \cdots\otimes
  B_{k}^{j_{k} j_{k+1}}\}]
  =c^L\mathcal{O}_{L}[\{\delta^{j_1 j_{k+1}}\}]
\end{eqnarray}
for all $L\in\mathbb{N}$.
Then, the spectrum of the $k$-transfer matrix
\begin{eqnarray}
    T^{(k)}[B_{1},...,B_{k}]:=\frac{1}{c\mathsf{D}}\sum_{j_1,...,j_k}B_{1}^{j_1 j_2}\otimes \cdots\otimes B_{k}^{j_{k} j_{1}},
\end{eqnarray}
acting on $M\in\mathbb{C}^{\chi_1\cdots \chi_k}$ as
\begin{eqnarray}
\begin{tikzpicture}[line cap=round,line join=round,x=1.0cm,y=1.0cm, scale=0.3, baseline={([yshift=-.6ex]current bounding box.center)}, thick, shift={(0,0)}, scale=0.8]
  \def\tate{3} 
  \def\yoko{0.5} 
  \def\sen{7} 
  \def\gaisen{2} 
  \pgfmathsetmacro\hankei{\tate*1.7} 
  \def\hankeinosa{2} 
  \def\nobishiro{1}
  \def\kankaku{6}
  \draw (0,\tate/2+\tate/2+\tate/2+\nobishiro) -- (0,0-\nobishiro);
  \fill (0,0-\nobishiro) circle (2mm); 
  \fill (0,\tate/2+\tate/2+\tate/2+\nobishiro) circle (2mm); 
  \draw (-\nobishiro,\tate+\tate/2) -- (+\nobishiro,\tate+\tate/2);
  \draw (+\nobishiro,\tate+\tate/2) node [right] {$B_1$};
  \draw (-\nobishiro,\tate) -- (+\nobishiro,\tate);
  \draw (+\nobishiro,\tate) node [right] {$B_2$};
  \fill (0-\nobishiro/2,\tate/2) circle (1mm);
  \fill (0+\nobishiro/2,\tate/2) circle (1mm);
  \fill (0-\nobishiro/2,\tate*0.75) circle (1mm);
  \fill (0+\nobishiro/2,\tate*0.75) circle (1mm);
  \fill (0-\nobishiro/2,\tate*0.25) circle (1mm);
  \fill (0+\nobishiro/2,\tate*0.25) circle (1mm);
  \draw (-\nobishiro,0) -- (+\nobishiro,0);
  \draw (+\nobishiro,0) node [right] {$B_{k}$};
\end{tikzpicture}
\;\;\;\cdot\;\;\;
\begin{tikzpicture}[line cap=round,line join=round,x=1.0cm,y=1.0cm, scale=0.3, baseline={([yshift=-.6ex]current bounding box.center)}, thick, shift={(0,0)}, scale=0.7]
  \def\tate{3} 
  \def\yoko{0.5} 
  \def\sen{7} 
  \def\gaisen{2} 
  \pgfmathsetmacro\hankei{\tate*1.7} 
  \def\hankeinosa{2} 
  \def\nobishiro{1}
  \def\kankaku{6}
  \draw (0+\nobishiro,\tate/2+\tate/2+\tate/2) -- (0+\nobishiro,0);
  \draw (-\nobishiro,\tate+\tate/2) -- (+\nobishiro,\tate+\tate/2);
  \draw (-\nobishiro,\tate) -- (+\nobishiro,\tate);
  \fill (0,\tate/2) circle (1mm);
  \fill (0,\tate*0.75) circle (1mm);
  \fill (0,\tate*0.25) circle (1mm);
  \draw (-\nobishiro,0) -- (+\nobishiro,0);
  \draw (\nobishiro,\tate*0.75) node [right] {$M$};
\end{tikzpicture}
\;\;\;:=\;\;\;
\begin{tikzpicture}[line cap=round,line join=round,x=1.0cm,y=1.0cm, scale=0.3, baseline={([yshift=-.6ex]current bounding box.center)}, thick, shift={(0,0)}, scale=0.7]
  \def\tate{3} 
  \def\yoko{0.5} 
  \def\sen{7} 
  \def\gaisen{2} 
  \pgfmathsetmacro\hankei{\tate*1.7} 
  \def\hankeinosa{2} 
  \def\nobishiro{1}
  \def\kankaku{6}
  \draw (0,\tate/2+\tate/2+\tate/2+\nobishiro) -- (0,0-\nobishiro);
  \fill (0,0-\nobishiro) circle (2mm); 
  \fill (0,\tate/2+\tate/2+\tate/2+\nobishiro) circle (2mm); 
  \draw (-\nobishiro,\tate+\tate/2) -- (+\kankaku/2,\tate+\tate/2);
  \draw (-\nobishiro,\tate) -- (+\kankaku/2,\tate);
  \fill (0+\nobishiro,\tate/2) circle (1mm);
  \fill (0+\nobishiro,\tate*0.75) circle (1mm);
  \fill (0+\nobishiro,\tate*0.25) circle (1mm);
  \draw (-\nobishiro,0) -- (+\kankaku/2,0);
  \draw (0+\kankaku/2,\tate/2+\tate/2+\tate/2) -- (0+\kankaku/2,0);
\end{tikzpicture}\, ,
\end{eqnarray}
is $\{0,1\}$ and the eigenvalue $1$ is nondegenerate.
In this diagram, 
black dots at the ends 
of vertical bonds  
mean that the physical indices 
at the top and the bottom are contracted together.
\end{proposition}

\begin{proof}
  By taking the vertical trace of Eq.\ \eqref{eq:cycle},
  we obtain $\mathrm{Tr}[(T^{(k)})^L]=1$. Here, $\mathrm{Tr}$ is the horizontal trace
  originating from PBC. Let $\{\xi_{i}\}$ be the spectrum of $T^{(k)}$. Then, the above equation turns out to be $\sum_{i}\xi_{i}^{L}=1$ for any $L$. This equation implies that one of the $\xi_{i}$, say $\xi_1$, is $1$ and the others are $0$.
\end{proof}
In the following, we often use this proposition for $k=2$ and $3$.

As mentioned in Sec.\ \ref{sec:semi},
while a transition MPO is guaranteed to be normal (and hence invertible),
it is not necessarily an MPU.
In our construction of the higher Berry phase below,
however, this invertibility, which holds 
when an MPO is an MPU, plays a crucial role. 
Therefore, we will further assume 
transition MPO are MPU.
\footnote{
When an MPO is an MPU,
it is guaranteed that 
its transfer matrix 
has 
an invertible fixed point.
} 
This can be rationalized by recalling the (1+1)d case. 
There, 
the transition functions 
($g_{\alpha\beta}$)
defined for physically equivalent normal MPS matrices are invertible
but not necessarily unitary.
However, 
it is possible to make the transition functions
unitary by performing a partial gauge fixing known as the canonical form
\cite{P-GVWC07}.
While the existence of such a canonical form for PEPS is not known at present,
it is expected that a similar gauge fixing exists,
allowing us to transform the transition MPO into an MPU.
We also note that in the context of SPT phases,
their classification does not change 
upon restricting ourselves to the case where the transition MPO is an MPU.
\footnote{In fact, when considering the classification of SPT phases, 
the invariants using MPO take values in $\coho{3}{G}{\mathbb{C}^{\times}}$, and when they have an MPU, they take values in $\cohoU{3}{G}$, but the two cohomologies are known to be mathematically isomorphic.}
For these reasons, we will assume henceforth that transition MPO are MPU.

For an MPO that is an MPU, the following proposition holds true: 
\begin{proposition}\label{thm:fp_unitary}\;\\
\\
Let
$\{B_{1}^{ij}\}_{i,j=1}^{\mathsf{D}},...,\{B_{k}^{ij}\}_{i,j=1}^{\mathsf{D}}$ be
normal MPU matrices
and assume that the composition of 
$\{B_{1}^{ij}\}_{i,j=1}^{\mathsf{D}},...,\{B_{k-1}^{ij}\}_{i,j=1}^{\mathsf{D}}$ 
is proportional to the MPO generated by $\{B_{k}^{ij}\}_{i,j=1}^{\mathsf{D}}$. That is, there exists an invertible complex number $c\in\mathbb{C}^{\times}$ such that 
\begin{eqnarray}
    \mathcal{O}_{L}[\{\sum_{j_2,...,j_{k-1}}B_{1}^{j_1 j_2}\otimes \cdots\otimes
  B_{k-1}^{j_{k-1} j_{k}}\}]
  =c^L\mathcal{O}_{L}[\{B_{k}^{j_1 j_{k}}\}]
\end{eqnarray}
for all $L\in\mathbb{N}$.
Then,\\
(i) the spectrum of the $(k-1,1)$-transfer matrix
\begin{eqnarray}
    T^{(k-1,1)}[B_{1},...,B_{k-1};B_{k}]:=\frac{1}{c \mathsf{D}}\sum_{j_1,...,j_{k}}B_{1}^{j_1 j_2}\otimes \cdots\otimes
  B_{k-1}^{j_{k-1} j_{k}}\otimes B_{k}^{j_{1}, j_{k}\dagger}
\end{eqnarray}
with respect to the right action on $M\in\mathbb{C}^{\chi_1\cdots \chi_{k-1}}\otimes\mathbb{C}^{\chi_k\ast}$ 
defined as
\begin{eqnarray}
\begin{tikzpicture}[line cap=round,line join=round,x=1.0cm,y=1.0cm, scale=0.3, baseline={([yshift=-.6ex]current bounding box.center)}, thick, shift={(0,0)}, scale=0.8]
  \def\tate{3} 
  \def\yoko{0.5} 
  \def\sen{7} 
  \def\gaisen{2} 
  \pgfmathsetmacro\hankei{\tate*1.7} 
  \def\hankeinosa{2} 
  \def\nobishiro{1}
  \def\kankaku{6}
  \draw (0,\tate/2+\tate/2+\tate/2+\nobishiro) -- (0,0-\nobishiro);
  \draw (-\nobishiro,\tate+\tate/2) -- (+\nobishiro,\tate+\tate/2);
  \draw (+\nobishiro,\tate+\tate/2) node [right] {$B_1$};
  \draw (-\nobishiro,\tate) -- (+\nobishiro,\tate);
  \draw (+\nobishiro,\tate) node [right] {$B_2$};
  \fill (0-\nobishiro/2,\tate/2) circle (1mm);
  \fill (0+\nobishiro/2,\tate/2) circle (1mm);
  \fill (0-\nobishiro/2,\tate*0.75) circle (1mm);
  \fill (0+\nobishiro/2,\tate*0.75) circle (1mm);
  \fill (0-\nobishiro/2,\tate*0.25) circle (1mm);
  \fill (0+\nobishiro/2,\tate*0.25) circle (1mm);
  \draw (-\nobishiro,0) -- (+\nobishiro,0);
  \draw (+\nobishiro,0) node [right] {$B_{k-1}$};
  \draw (0+\kankaku,\tate/2+\tate/2+\tate/2+\nobishiro) -- (0+\kankaku,0-\nobishiro);
  \draw (-\nobishiro+\kankaku,\tate*0.75) -- (+\nobishiro+\kankaku,\tate*0.75);
  \draw (+\nobishiro+\kankaku,\tate*0.75) node [right] {$B_{k}^{\dagger}$};
  \draw (0,\tate/2+\tate/2+\tate/2+\nobishiro) to [out=90,in=90] (0+\kankaku,\tate/2+\tate/2+\tate/2+\nobishiro);
  \draw (0,0-\nobishiro) to [out=270,in=270] (0+\kankaku,0-\nobishiro);
\end{tikzpicture}
\;\;\;\cdot\;\;\;
\begin{tikzpicture}[line cap=round,line join=round,x=1.0cm,y=1.0cm, scale=0.3, baseline={([yshift=-.6ex]current bounding box.center)}, thick, shift={(0,0)}, scale=0.7]
  \def\tate{3} 
  \def\yoko{0.5} 
  \def\sen{7} 
  \def\gaisen{2} 
  \pgfmathsetmacro\hankei{\tate*1.7} 
  \def\hankeinosa{2} 
  \def\nobishiro{1}
  \def\kankaku{6}
  \draw (0+\nobishiro,\tate/2+\tate/2+\tate/2) -- (0+\nobishiro,0);
  \draw (-\nobishiro,\tate+\tate/2) -- (+\nobishiro,\tate+\tate/2);
  \draw (-\nobishiro,\tate) -- (+\nobishiro,\tate);
  \fill (0,\tate/2) circle (1mm);
  \fill (0,\tate*0.75) circle (1mm);
  \fill (0,\tate*0.25) circle (1mm);
  \draw (-\nobishiro,0) -- (+\nobishiro,0);
  \draw (\nobishiro,\tate*0.75) -- (+\nobishiro+\nobishiro+\nobishiro,\tate*0.75);
  \draw (\nobishiro,\tate*0.75) node [below right] {$M$};
\end{tikzpicture}
\;\;\;:=\;\;\;
\begin{tikzpicture}[line cap=round,line join=round,x=1.0cm,y=1.0cm, scale=0.3, baseline={([yshift=-.6ex]current bounding box.center)}, thick, shift={(0,0)}, scale=0.7]
  \def\tate{3} 
  \def\yoko{0.5} 
  \def\sen{7} 
  \def\gaisen{2} 
  \pgfmathsetmacro\hankei{\tate*1.7} 
  \def\hankeinosa{2} 
  \def\nobishiro{1}
  \def\kankaku{6}
  \draw (0,\tate/2+\tate/2+\tate/2+\nobishiro) -- (0,0-\nobishiro);
  \draw (-\nobishiro,\tate+\tate/2) -- (+\kankaku/2,\tate+\tate/2);
  \draw (-\nobishiro,\tate) -- (+\kankaku/2,\tate);
  \fill (0+\nobishiro,\tate/2) circle (1mm);
  \fill (0+\nobishiro,\tate*0.75) circle (1mm);
  \fill (0+\nobishiro,\tate*0.25) circle (1mm);
  \draw (-\nobishiro,0) -- (+\kankaku/2,0);
  \draw (0+\kankaku,\tate/2+\tate/2+\tate/2+\nobishiro) -- (0+\kankaku,0-\nobishiro);
  \draw (+\kankaku/2,\tate*0.75) -- (+\nobishiro+\kankaku,\tate*0.75);
  \draw (0,\tate/2+\tate/2+\tate/2+\nobishiro) to [out=90,in=90] (0+\kankaku,\tate/2+\tate/2+\tate/2+\nobishiro);
  \draw (0,0-\nobishiro) to [out=270,in=270] (0+\kankaku,0-\nobishiro);
  \draw (0+\kankaku/2,\tate/2+\tate/2+\tate/2) -- (0+\kankaku/2,0);
\end{tikzpicture}
\end{eqnarray}
is $\{0,1\}$ and the eigenvalue $1$ is nondegenerate.
Here, $B^{ij\dagger}$ represents  
the hermitian conjugation of
$(B^{ij})_{ab}$
with respect to the virtual legs $a,b$,
$(B^{ij\dagger})_{ab}=
(B^{ij})^*_{ba}$. 
\\
(ii) the spectrum of the $(k-1,1)$-transfer matrix
\begin{eqnarray}
    T^{(k-1,1)}[B_{1},...,B_{k-1};B_{k}]:=\frac{1}{c\mathsf{D}}\sum_{j_1,...,j_k} B_{k}^{j_{1}, j_{k}\dagger}\otimes B_{1}^{j_1 j_2}\otimes \cdots\otimes
  B_{k-1}^{j_{k-1} j_{k}} 
\end{eqnarray}
with respect to the left action on $M\in\mathbb{C}^{\chi_1\cdots \chi_{k-1}}\otimes\mathbb{C}^{\chi_k\ast}$ defined as
\begin{eqnarray}
\begin{tikzpicture}[line cap=round,line join=round,x=1.0cm,y=1.0cm, scale=0.3, baseline={([yshift=-.6ex]current bounding box.center)}, thick, shift={(0,0)}, scale=0.8]
  \def\tate{3} 
  \def\yoko{0.5} 
  \def\sen{7} 
  \def\gaisen{2} 
  \pgfmathsetmacro\hankei{\tate*1.7} 
  \def\hankeinosa{2} 
  \def\nobishiro{1}
  \def\kankaku{-6}
  \draw (0,\tate/2+\tate/2+\tate/2+\nobishiro) -- (0,0-\nobishiro);
  \draw (-\nobishiro,\tate+\tate/2) -- (+\nobishiro,\tate+\tate/2);
  \draw (+\nobishiro,\tate+\tate/2) node [right] {$B_{1}$};
  \draw (-\nobishiro,\tate) -- (+\nobishiro,\tate);
  \draw (+\nobishiro,\tate) node [right] {$B_{2}$};
  \fill (0-\nobishiro/2,\tate/2) circle (1mm);
  \fill (0+\nobishiro/2,\tate/2) circle (1mm);
  \fill (0-\nobishiro/2,\tate*0.75) circle (1mm);
  \fill (0+\nobishiro/2,\tate*0.75) circle (1mm);
  \fill (0-\nobishiro/2,\tate*0.25) circle (1mm);
  \fill (0+\nobishiro/2,\tate*0.25) circle (1mm);
  \draw (-\nobishiro,0) -- (+\nobishiro,0);
  \draw (+\nobishiro,0) node [right] {$B_{k-1}$};
  \draw (0+\kankaku,\tate/2+\tate/2+\tate/2+\nobishiro) -- (0+\kankaku,0-\nobishiro);
  \draw (-\nobishiro+\kankaku,\tate*0.75) -- (+\nobishiro+\kankaku,\tate*0.75);
  \draw (+\nobishiro+\kankaku,\tate*0.75) node [right] {$B_{k}^{\dagger}$};
  \draw (0,\tate/2+\tate/2+\tate/2+\nobishiro) to [out=90,in=90] (0+\kankaku,\tate/2+\tate/2+\tate/2+\nobishiro);
  \draw (0,0-\nobishiro) to [out=270,in=270] (0+\kankaku,0-\nobishiro);
\end{tikzpicture}
\;\;\;\cdot\;\;\;
\begin{tikzpicture}[line cap=round,line join=round,x=1.0cm,y=1.0cm, scale=0.3, baseline={([yshift=-.6ex]current bounding box.center)}, thick, shift={(0,0)}, scale=0.7]
  \def\tate{3} 
  \def\yoko{0.5} 
  \def\sen{7} 
  \def\gaisen{2} 
  \pgfmathsetmacro\hankei{\tate*1.7} 
  \def\hankeinosa{2} 
  \def\nobishiro{1}
  \def\kankaku{6}
  \draw (0-\nobishiro,\tate/2+\tate/2+\tate/2) -- (0-\nobishiro,0);
  \draw (-\nobishiro,\tate+\tate/2) -- (+\nobishiro,\tate+\tate/2);
  \draw (-\nobishiro,\tate) -- (+\nobishiro,\tate);
  \fill (0,\tate/2) circle (1mm);
  \fill (0,\tate*0.75) circle (1mm);
  \fill (0,\tate*0.25) circle (1mm);
  \draw (-\nobishiro,0) -- (+\nobishiro,0);
  \draw (-\nobishiro-\nobishiro-\nobishiro,\tate*0.75) -- (-\nobishiro,\tate*0.75);
  \draw (0+\nobishiro/2,\tate/2) node [right] {$M$};
\end{tikzpicture}
\;\;\;:=\;\;\;
\begin{tikzpicture}[line cap=round,line join=round,x=1.0cm,y=1.0cm, scale=0.3, baseline={([yshift=-.6ex]current bounding box.center)}, thick, shift={(0,0)}, scale=0.7]
  \def\tate{3} 
  \def\yoko{0.5} 
  \def\sen{7} 
  \def\gaisen{2} 
  \pgfmathsetmacro\hankei{\tate*1.7} 
  \def\hankeinosa{2} 
  \def\nobishiro{1}
  \def\kankaku{-6}
  \draw (0,\tate/2+\tate/2+\tate/2+\nobishiro) -- (0,0-\nobishiro);
  \draw (\nobishiro,\tate+\tate/2) -- (+\kankaku/2,\tate+\tate/2);
  \draw (\nobishiro,\tate) -- (+\kankaku/2,\tate);
  \fill (0-\nobishiro,\tate/2) circle (1mm);
  \fill (0-\nobishiro,\tate*0.75) circle (1mm);
  \fill (0-\nobishiro,\tate*0.25) circle (1mm);
  \draw (\nobishiro,0) -- (+\kankaku/2,0);
  \draw (0+\kankaku,\tate/2+\tate/2+\tate/2+\nobishiro) -- (0+\kankaku,0-\nobishiro);
  \draw (+\kankaku/2,\tate*0.75) -- (-\nobishiro+\kankaku,\tate*0.75);
  \draw (0,\tate/2+\tate/2+\tate/2+\nobishiro) to [out=90,in=90] (0+\kankaku,\tate/2+\tate/2+\tate/2+\nobishiro);
  \draw (0,0-\nobishiro) to [out=270,in=270] (0+\kankaku,0-\nobishiro);
  \draw (0+\kankaku/2,\tate/2+\tate/2+\tate/2) -- (0+\kankaku/2,0);
\end{tikzpicture}
\end{eqnarray}
is $\{0,1\}$ and the eigenvalue $1$ is nondegenerate.

\end{proposition}
Since 
$T^{(1,1)}[B_1;B_2]$ 
is nothing but the usual mixed transfer matrix \cite{P-GWSV08} for the normal MPU
matrices $\{B_1^{ij}\}$ and $\{B_2^{ij}\}$, the fixed point of $T^{(1,1)}[B_1;B_2]$ is invertible. 



\subsubsection{The Definition of the Quadruple Inner Product}


With these preliminaries, 
our definition of the quadruple inner 
product of four quantum states
can be summarized as follows 
in terms of the fixed points
of 2- and 3-transfer matrices.
Let $\ket{\psi_{\alpha}}$, $\ket{\psi_{\beta}}$, $\ket{\psi_{\gamma}}$ and $\ket{\psi_{\delta}}$ be physically identical states and assume that they admit the semi-injective PEPS representation. 
For each pair of these states, 
we obtain the set of transition MPO tensors 
$\{B^{ij}_{\alpha\beta}\}$. 
As we mentioned in Sec.\ \ref{sec:assumption}, 
we assume that the MPO $\{\mathcal{O}_L[\{B^{ij}_{\alpha\beta}\}]\}$
are MPU.
By using this assumption, 
$\mathcal{O}_{L}[\{B^{ij}_{\alpha\beta}\}]^{-1}
=\mathcal{O}_{L}[\{B^{ij}_{\alpha\beta}\}]^{\dagger}
=\mathcal{O}_{L}[\{B^{ji\ast}_{\alpha\beta}\}]$. 

For each pair of these states, 
$\ket{\psi_\alpha}$ and $\ket{\psi_\beta}$, say, 
we have two transition MPU generated by $\{B^{ij}_{\alpha\beta}\}$ and $\{B^{ij}_{\beta\alpha}\}$. By the fundamental theorem for semi-injective PEPS, we can show that 
\begin{eqnarray}\label{eq:2_transfer_condition}
    \mathcal{O}_{L}[\{\sum_{k}B_{\alpha\beta}^{ik}\otimes B_{\beta\alpha}^{kj}\}]
  =(c_{2})^L\mathcal{O}_{L}[\{\delta^{ij}\}],
\end{eqnarray}
for some $c_{2}\in\mathbb{C}^{\times}$. This equation implies that
\begin{eqnarray}
    \mathcal{O}_{L}[\{B^{ij}_{\beta\alpha}\}]=(c_{\alpha\beta})^L\mathcal{O}_{L}[\{B^{ij}_{\alpha\beta}\}]^{-1}=(c_{\alpha\beta})^L\mathcal{O}_{L}[\{B^{ji\ast}_{\alpha\beta}\}],
\end{eqnarray}
that is, $\{B^{ij}_{\beta\alpha}\}$ and $\{B^{ji\ast}_{\alpha\beta}\}$ are gauge equivalent. Therefore, we can take the gauge 
in which the MPO matrices satisfy 
\begin{eqnarray}
    B^{ij}_{\beta\alpha}=B^{ji\ast}_{\alpha\beta}.
\end{eqnarray}
Since Eq.\ \eqref{eq:2_transfer_condition} holds, we can define the $2$-transfer matrix
\begin{eqnarray}
    T^{(2)}_{\alpha\beta}:=T^{(2)}[B_{\alpha\beta},B_{\beta\alpha}],
\end{eqnarray}
and apply Prop.\ \ref{thm:fp}. Let $\Lambda^{R(2)}_{\alpha\beta}$ (resp. $\Lambda^{L(2)}_{\alpha\beta}$) denote the right (resp. left) fixed point of $T^{(2)}_{\alpha\beta}$.

For each triplet of these states, 
say $\ket{\psi_\alpha}$, $\ket{\psi_\beta}$ and $\ket{\psi_\gamma}$,
we have three transition MPU
generated by $\{B^{ij}_{\alpha\beta}\}$, $\{B^{ij}_{\beta\gamma}\}$ and $\{B^{ij}_{\gamma\alpha}\}$. By the fundamental theorem for semi-injective PEPS, we can show that 
\begin{eqnarray}
    \mathcal{O}_{L}[\{\sum_{k}B_{\alpha\beta}^{ik}\otimes B_{\beta\gamma}^{kj}\}]
  =c_{3}^L\mathcal{O}_{L}[\{B_{\alpha\gamma}^{ij}\}],
\end{eqnarray}
for some $c_{3}\in\mathbb{C}^{\times}$. Thus we can define the $(2,1)$-transfer matrix
\begin{eqnarray}
    T^{R(2,1)}_{\alpha\beta\gamma}:=T^{(2,1)}[B_{\alpha\beta},B_{\beta\gamma};\Lambda_{\alpha\gamma}^{R}B_{\alpha\gamma}(\Lambda_{\alpha\gamma}^{R})^{-1}],
\end{eqnarray}
and apply Prop.\ \ref{thm:fp_unitary}. Here, $\Lambda_{\alpha\gamma}^{R}$ is 
the unique right fixed point of the transfer matrix $
T^{(1,1)}[B_{\alpha\gamma};B_{\alpha\gamma}]
$.
Note that although the phase of $\Lambda_{\alpha\gamma}^{R}$ is indeterminate from the fix point equation, $T^{R(2,1)}_{\alpha\beta\gamma}$ does not depend on the choice of the phase factor. 
Let $\Lambda^{R(3)}_{\alpha\beta\gamma}$ denote the right fixed point of $T^{R(2,1)}_{\alpha\beta\gamma}$. Similarly,
we can define the left $(2,1)$-transfer matrix
\begin{eqnarray}
    T^{L(2,1)}_{\alpha\beta\gamma}:=T^{(2,1)}[B_{\alpha\beta},B_{\beta\gamma};(\Lambda_{\alpha\gamma}^{L})^{-1}B_{\alpha\gamma}\Lambda_{\alpha\gamma}^{L}],
\end{eqnarray}
and apply the Prop.\ \ref{thm:fp_unitary}. Here, $\Lambda_{\alpha\gamma}^{L}$ is a unique left fixed point of the transfer matrix 
$
T^{(1,1)}[B_{\alpha\gamma};B_{\alpha\gamma}]
$.
Note that although the phase of $\Lambda_{\alpha\gamma}^{L}$ is indeterminate from the fix point equation, $T^{L(2,1)}_{\alpha\beta\gamma}$ does not depend on the choice of the phase factor. 
Let $\Lambda^{L(3)}_{\alpha\beta\gamma}$ denote the left fixed point of $T^{L(2,1)}_{\alpha\beta\gamma}$.




Now, we have the fixed point tensors $\{\Lambda^{R(2)}_{\alpha\beta},\Lambda^{L(2)}_{\alpha\beta}\}_{\alpha\beta}$ 
for each pair, and
$\{\Lambda^{R(3)}_{\alpha\beta\gamma},\Lambda^{L(3)}_{\alpha\beta\gamma}\}_{\alpha\beta\gamma}$
for each triplet. Since these tensors are constructed as 
the fixed points of the transfer matrices,
there is an ambiguity as to how the phases are assigned.
In particular, the phases of $\Lambda_{\alpha\beta}^{L(2)}$
and $\Lambda_{\alpha\beta}^{R(2)}$
(as well as those of $\Lambda_{\alpha\beta\gamma}^{L(3)}$ and $\Lambda_{\alpha\beta\gamma}^{R(3)}$) can be determined independently. 
However, we note the following proposition:
  \footnote{In the following, we will denote the fixed point tensor as a white box.} 
\begin{proposition}\label{prop3}\;\\ 
\\
We define the inner product of $\Lambda_{\alpha\beta}^{L(2)}$ and $\Lambda_{\alpha\beta}^{R(2)}$ as
\begin{eqnarray}
    (\Lambda_{\alpha\beta}^{L(2)},\Lambda_{\alpha\beta}^{R(2)}):=
    \begin{tikzpicture}[line cap=round,line join=round,x=1.0cm,y=1.0cm, scale=0.3, baseline={([yshift=-.6ex]current bounding box.center)}, thick, shift={(0,0)}, scale=0.7]
  \def\tate{3} 
  \def\yoko{0.5} 
  \def\sen{7} 
  \draw (0,0) rectangle ++(\yoko,\tate);
  \node[anchor=east] at (0,\tate/2) {$\Lambda_{\alpha\beta}^{L(2)}$};
  \draw (\yoko+\sen,0) rectangle ++(\yoko,\tate);
  \node[anchor=west] at (\yoko+\sen+\yoko,\tate/2) 
  {$\Lambda_{\alpha\beta}^{R(2)}$};
  \draw (\yoko,\tate) -- (\yoko+\sen,\tate);
  \draw (\yoko,0) -- (\yoko+\sen,0);
\end{tikzpicture}
\end{eqnarray}
 and the inner product of $\Lambda_{\alpha\beta\gamma}^{L(3)}$ and $\Lambda_{\alpha\beta\gamma}^{R(3)}$ as
  \begin{eqnarray}\label{eq:normalization_3leg}
    (\Lambda_{\alpha\beta\gamma}^{L(3)},\Lambda_{\alpha\beta\gamma}^{R(3)})
    :=
    \begin{tikzpicture}[line cap=round,line join=round,x=1.0cm,y=1.0cm, scale=0.3, baseline={([yshift=-.4ex]current bounding box.center)}, thick, shift={(0,0)}, scale=0.7]
  \def\tate{3} 
  \def\yoko{0.5} 
  \def\sen{7} 
  \def\gaisen{2} 
  \draw (0,0) rectangle ++(\yoko,\tate);
  \node[anchor=east] at (0,\tate) {$\Lambda_{\alpha\beta\gamma}^{L(2,1)}$};
  \draw (\yoko+\sen,0) rectangle ++(\yoko,\tate);
  \node[anchor=west] at (\yoko+\sen+\yoko,\tate)  {$\Lambda_{\alpha\beta\gamma}^{R(2,1)}$};
  \draw (\yoko,\tate) -- (\yoko+\sen,\tate);
  \draw (\yoko,0) -- (\yoko+\sen,0);
  \draw (\yoko+\sen+\yoko,\tate/2) -- (\yoko+\sen+\yoko+\gaisen,\tate/2);
  \draw (-\gaisen,\tate/2) -- (0,\tate/2);
  \draw (-\gaisen,-\tate/2) -- (\yoko+\sen+\yoko+\gaisen,-\tate/2);
  \draw (-\gaisen,\tate/2) arc (90:270:\tate/2);
  \draw (\yoko+\sen+\yoko+\gaisen,\tate/2) arc (90:-90:\tate/2);
\end{tikzpicture}\, .
\end{eqnarray}
  Then, they are non-zero complex numbers.
\end{proposition}

Using Prop.\ \ref{prop3},
once we chose
the phase of the right fixed points, 
we 
can determine the phase of the left fixed point such that 
the inner product is $1$.
We will show this proposition in Sec.\ \ref{sec:proof}

For 
the quartet of
states $\ket{\psi_\alpha}$, $\ket{\psi_\beta}$, $\ket{\psi_\gamma}$ and $\ket{\psi_\delta}$, we can also show that 
\begin{eqnarray}\label{eq:4_transfer_condition}
    \mathcal{O}_{L}[\{\sum_{k,l,m}B_{\alpha\beta}^{ik}\otimes B_{\beta\gamma}^{kl}\otimes B_{\gamma\delta}^{lm}\otimes B_{\delta\alpha}^{mj}\}]
  =c_{4}^L\mathcal{O}_{L}[\{\delta^{ij}\}]
\end{eqnarray}
for some complex number $c_4\in\mathbb{C}^{\times}$. Thus we can define a $4$-transfer matrix 
\begin{eqnarray}
    T^{(4)}_{\alpha\beta\gamma\delta}:=T^{(4)}[B_{\alpha\beta}, B_{\beta\gamma}, B_{\gamma\delta},B_{\delta\alpha}].
\end{eqnarray}
Finally, by using these tensors, we define the quadruple inner product as follows:
\begin{definition}\;\\
\\
For physically identical four states $\ket{\psi_{\alpha}}$, $\ket{\psi_{\beta}}$, 
$\ket{\psi_{\gamma}}$ and $\ket{\psi_{\delta}}$,
(i) we define $N_{\alpha\beta\gamma\delta}^{1}\in\mathbb{N}$ as a minimal natural number such that   
\begin{eqnarray}
    \left(\;\;\;\begin{tikzpicture}[line cap=round,line join=round,x=1.0cm,y=1.0cm, scale=0.3, baseline={([yshift=-.6ex]current bounding box.center)}, thick, shift={(0,0)}, scale=0.7]
  \def\tate{3} 
  \def\yoko{0.5} 
  \def\sen{7} 
  \def\gaisen{2} 
  \pgfmathsetmacro\hankei{\tate*1.7} 
  \def\hankeinosa{2} 
  \def\nobishiro{1}
  \draw (0,\tate+\tate/2+\tate/2+\nobishiro) -- (0,0-\nobishiro);
  \draw (-\nobishiro,\tate+\tate) -- (+\nobishiro,\tate+\tate);
  \draw (-\nobishiro,\tate) -- (+\nobishiro,\tate);
  \draw (-\nobishiro,\tate/2) -- (+\nobishiro,\tate/2);
  \draw (-\nobishiro,0) -- (+\nobishiro,0);
\end{tikzpicture}
\;\;\;-\;\;\;
    \begin{tikzpicture}[line cap=round,line join=round,x=1.0cm,y=1.0cm, scale=0.3, baseline={([yshift=-.6ex]current bounding box.center)}, thick, shift={(0,0)}, scale=0.7]
  \def\tate{3} 
  \def\yoko{0.5} 
  \def\sen{0} 
  \def\gaisen{2} 
  \pgfmathsetmacro\hankei{\tate*1.7} 
  \def\hankeinosa{2} 
  \def\nobishiro{1}
  \def\kankaku{3}
  \draw (0,0) rectangle ++(\yoko,\tate);
  \draw (-\gaisen,\tate) -- (0,\tate);
  \draw (-\gaisen-\gaisen-\yoko-\gaisen-\yoko,0) -- (0,0);
  \draw (-\gaisen-\yoko,\tate/2) rectangle ++(\yoko,\tate);
  \draw (-\gaisen-\gaisen-\yoko,\tate+\tate/2) -- (-\gaisen-\yoko,\tate+\tate/2);
  \draw (-\gaisen-\gaisen-\yoko-\gaisen-\yoko,\tate/2) -- (-\gaisen-\yoko,\tate/2);
  \draw (-\gaisen-\yoko-\gaisen-\yoko,\tate/2+\tate/2) rectangle ++(\yoko,\tate);
  \draw (-\gaisen-\gaisen-\yoko-\gaisen-\yoko,\tate+\tate/2+\tate/2) -- (-\gaisen-\yoko-\gaisen-\yoko,\tate+\tate/2+\tate/2);
  \draw (-\gaisen-\gaisen-\yoko-\gaisen-\yoko,\tate/2+\tate/2) -- (-\gaisen-\yoko-\gaisen-\yoko,\tate/2+\tate/2);
  \draw (0+\kankaku,0) rectangle ++(\yoko,\tate);
  \draw (\yoko+\kankaku,\tate) -- (\yoko+\gaisen+\kankaku,\tate);
  \draw (\yoko+\kankaku,0) -- (\yoko+\gaisen+\gaisen+\yoko+\gaisen+\yoko+\kankaku,0);
  \draw (\yoko+\gaisen+\kankaku,\tate/2) rectangle ++(\yoko,\tate);
  \draw (\yoko+\gaisen+\yoko+\kankaku,\tate+\tate/2) -- (\yoko+\gaisen+\yoko+\gaisen+\kankaku,\tate+\tate/2);
  \draw (\yoko+\gaisen+\yoko+\kankaku,\tate/2) -- (\yoko+\gaisen+\gaisen+\yoko+\gaisen+\yoko+\kankaku,\tate/2);
  \draw (\yoko+\gaisen+\yoko+\gaisen+\kankaku,\tate/2+\tate/2) rectangle ++(\yoko,\tate);
  \draw (\yoko+\gaisen+\yoko+\gaisen+\yoko+\kankaku,\tate+\tate/2+\tate/2) -- (\yoko+\gaisen+\gaisen+\yoko+\gaisen+\yoko+\kankaku,\tate+\tate/2+\tate/2);
  \draw (\yoko+\gaisen+\yoko+\gaisen+\yoko+\kankaku,\tate/2+\tate/2) -- (\yoko+\gaisen+\gaisen+\yoko+\gaisen+\yoko+\kankaku,\tate/2+\tate/2);
  \draw [dotted] (0+\yoko/2+\kankaku/2,\tate+\tate+\nobishiro) -- (0+\yoko/2+\kankaku/2,-\nobishiro);
  \draw (0+\yoko/2+\kankaku/2,\tate+\tate) node [right]{$\delta^{ij}$};
\end{tikzpicture}\;\;\;\right)^{N_{\alpha\beta\gamma\delta}^{1}}=0, 
\end{eqnarray}
and define $N_{\alpha\beta\gamma\delta}^{2}\in\mathbb{N}$ as a minimal natural number such that   
\begin{eqnarray}
    \left(\;\;\;\begin{tikzpicture}[line cap=round,line join=round,x=1.0cm,y=1.0cm, scale=0.3, baseline={([yshift=-.6ex]current bounding box.center)}, thick, shift={(0,0)}]
  \def\tate{3} 
  \def\yoko{0.5} 
  \def\sen{7} 
  \def\gaisen{2} 
  \pgfmathsetmacro\hankei{\tate*1.7} 
  \def\hankeinosa{2} 
  \def\nobishiro{1}
  \draw (0,\tate+\tate/2+\tate/2+\nobishiro) -- (0,0-\nobishiro);
  \draw (-\nobishiro,\tate+\tate) -- (+\nobishiro,\tate+\tate);
  \draw (-\nobishiro,\tate) -- (+\nobishiro,\tate);
  \draw (-\nobishiro,\tate/2) -- (+\nobishiro,\tate/2);
  \draw (-\nobishiro,0) -- (+\nobishiro,0);
\end{tikzpicture}
\;\;\;-\;\;\;
    \begin{tikzpicture}[line cap=round,line join=round,x=1.0cm,y=1.0cm, scale=0.3, baseline={([yshift=-.6ex]current bounding box.center)}, thick, shift={(0,0)}]
  \def\tate{3} 
  \def\yoko{0.5} 
  \def\sen{0} 
  \def\gaisen{2} 
  \pgfmathsetmacro\hankei{\tate*1.7} 
  \def\hankeinosa{2} 
  \def\nobishiro{1}
  \def\kankaku{3}
  \draw (0,0) rectangle ++(\yoko,\tate+\tate/2);
  \draw (-\gaisen,\tate+\tate/2) -- (0,\tate+\tate/2);
  \draw (-\gaisen-\gaisen-\yoko-\gaisen-\yoko,0) -- (0,0);
  \draw (-\gaisen-\yoko,\tate/2+\tate/2) rectangle ++(\yoko,\tate);
  \draw (-\gaisen-\gaisen-\yoko-\yoko-\gaisen,\tate+\tate/2+\tate/2) -- (-\gaisen-\yoko,\tate+\tate/2+\tate/2);
  \draw (-\gaisen-\gaisen-\yoko-\gaisen-\yoko+\yoko+\gaisen,\tate/2+\tate/2) -- (-\gaisen-\yoko,\tate/2+\tate/2);
  \draw (-\gaisen-\yoko-\gaisen-\yoko,\tate/2) rectangle ++(\yoko,\tate);
  \draw (-\gaisen-\gaisen-\yoko-\gaisen-\yoko,\tate+\tate/2) -- (-\gaisen-\yoko-\gaisen-\yoko,\tate+\tate/2);
  \draw (-\gaisen-\gaisen-\yoko-\gaisen-\yoko,\tate/2) -- (-\gaisen-\yoko-\gaisen-\yoko,\tate/2);
  \draw (0+\kankaku,0) rectangle ++(\yoko,\tate+\tate/2);
  \draw (\yoko+\kankaku,\tate+\tate/2) -- (\yoko+\gaisen+\kankaku,\tate+\tate/2);
  \draw (\yoko+\kankaku,0) -- (\yoko+\gaisen+\gaisen+\yoko+\gaisen+\yoko+\kankaku,0);
  \draw (\yoko+\gaisen+\kankaku,\tate/2+\tate/2) rectangle ++(\yoko,\tate);
  \draw (\yoko+\gaisen+\yoko+\kankaku,\tate+\tate/2+\tate/2) -- (\yoko+\gaisen+\yoko+\gaisen+\kankaku+\gaisen+\yoko,\tate+\tate/2+\tate/2);
  \draw (\yoko+\gaisen+\yoko+\kankaku,\tate/2+\tate/2) -- (\yoko+\gaisen+\gaisen+\yoko+\kankaku,\tate/2+\tate/2);
  \draw (\yoko+\gaisen+\yoko+\gaisen+\kankaku,\tate/2) rectangle ++(\yoko,\tate);
  \draw (\yoko+\gaisen+\yoko+\gaisen+\yoko+\kankaku,\tate+\tate/2) -- (\yoko+\gaisen+\gaisen+\yoko+\gaisen+\yoko+\kankaku,\tate+\tate/2);
  \draw (\yoko+\gaisen+\yoko+\gaisen+\yoko+\kankaku,\tate/2) -- (\yoko+\gaisen+\gaisen+\yoko+\gaisen+\yoko+\kankaku,\tate/2);
  \draw [dotted] (0+\yoko/2+\kankaku/2,\tate+\tate+\nobishiro) -- (0+\yoko/2+\kankaku/2,-\nobishiro);
  \draw (0+\yoko/2+\kankaku/2,\tate+\tate) node [right]{$\delta^{ij}$};
\end{tikzpicture}\;\;\;\right)^{N_{\alpha\beta\gamma\delta}^{2}}=0.
\end{eqnarray}
(ii) we define the quadruple inner product of them as
\begin{eqnarray}
\label{eq:quad_inner_product}
  c_{\alpha\beta\gamma\delta}=\begin{tikzpicture}[line cap=round,line join=round,x=1.0cm,y=1.0cm, scale=0.3, baseline={([yshift=-.6ex]current bounding box.center)}, thick, shift={(0,0)}]
  \def\tate{3} 
  \def\yoko{0.5} 
  \def\sen{7} 
  \def\gaisen{2} 
  \pgfmathsetmacro\hankei{\tate*1.7} 
  \def\hankeinosa{2} 
  \def\nobishiro{1}
  \draw (0,0) rectangle ++(\yoko,\tate+\tate/2);
  \draw (\yoko,\tate+\tate/2) -- (\yoko+\gaisen,\tate+\tate/2);
  \draw (\yoko,0) -- (\yoko+\gaisen+\gaisen+\yoko+\gaisen+\yoko+\gaisen*2,0);
  \draw (0,0) node [left]{$\Lambda_{\alpha\delta}^{L(2)}$};
  \draw (\yoko+\gaisen,\tate/2+\tate/2) rectangle ++(\yoko,\tate);
  \draw (\yoko+\gaisen+\yoko,\tate/2+\tate/2+\tate) -- (\yoko+\gaisen+\gaisen+\yoko+\gaisen+\yoko+\gaisen*2,\tate/2+\tate/2+\tate);
  \draw (\yoko+\gaisen+\yoko,\tate/2+\tate/2) -- (\yoko+\gaisen+\yoko+\gaisen,\tate/2+\tate/2);
  \draw (\yoko+\gaisen+\yoko,\tate/2+\tate/2+\tate) node [left]{$\Lambda_{\alpha\beta\delta}^{L(3)}$};
  \draw (\yoko+\gaisen+\yoko+\gaisen,\tate/2) rectangle ++(\yoko,\tate);
  \draw (\yoko+\gaisen+\yoko+\gaisen+\yoko,\tate+\tate/2) -- (\yoko+\gaisen+\yoko+\gaisen+\yoko+\yoko+\yoko,\tate+\tate/2);
  \draw (\yoko+\gaisen+\yoko+\gaisen+\yoko,\tate/2) -- (\yoko+\gaisen+\gaisen+\yoko+\gaisen+\yoko+\gaisen*2,\tate/2);
  \draw (\gaisen+\yoko+\gaisen+\yoko,\tate/2) node [left]{$\Lambda_{\beta\gamma\delta}^{L(3)}$};
  \draw (\yoko+\gaisen+\yoko+\gaisen+\yoko+3*\gaisen*0.25,\tate+\tate/2+\tate/2+\nobishiro) -- (\yoko+\gaisen+\yoko+\gaisen+\yoko+3*\gaisen*0.25,0-\nobishiro);
  \fill (\yoko+\gaisen+\yoko+\gaisen+\yoko+3*\gaisen*0.25,\tate+\tate/2+\tate/2+\nobishiro) circle (2mm); 
  \fill (\yoko+\gaisen+\yoko+\gaisen+\yoko+3*\gaisen*0.25,0-\nobishiro) circle (2mm);
  \draw (\yoko+\gaisen+\yoko+\gaisen+\yoko+3*\gaisen*0.75,\tate+\tate/2+\tate/2+\nobishiro) -- (\yoko+\gaisen+\yoko+\gaisen+\yoko+3*\gaisen*0.75,0-\nobishiro);
  \fill (\yoko+\gaisen+\yoko+\gaisen+\yoko+3*\gaisen*0.75,\tate+\tate/2+\tate/2+\nobishiro) circle (2mm); 
  \fill (\yoko+\gaisen+\yoko+\gaisen+\yoko+3*\gaisen*0.75,0-\nobishiro) circle (2mm);
  \fill (\yoko+\gaisen+\yoko+\gaisen+\yoko+3*\gaisen*0.375,\tate+\tate/2+\tate/2+\nobishiro/2) circle (1mm); 
  \fill (\yoko+\gaisen+\yoko+\gaisen+\yoko+3*\gaisen*0.5,\tate+\tate/2+\tate/2+\nobishiro/2) circle (1mm); 
  \fill (\yoko+\gaisen+\yoko+\gaisen+\yoko+3*\gaisen*0.625,\tate+\tate/2+\tate/2+\nobishiro/2) circle (1mm); 
  \fill (\yoko+\gaisen+\yoko+\gaisen+\yoko+3*\gaisen*0.375,-\nobishiro/2) circle (1mm); 
  \fill (\yoko+\gaisen+\yoko+\gaisen+\yoko+3*\gaisen*0.5,-\nobishiro/2) circle (1mm); 
  \fill (\yoko+\gaisen+\yoko+\gaisen+\yoko+3*\gaisen*0.625,-\nobishiro/2) circle (1mm); 
  \draw (\yoko+\gaisen+\gaisen+\yoko+\gaisen+\yoko+\gaisen*2+\gaisen+\yoko+\gaisen+\yoko,0) rectangle ++(\yoko,\tate);
  \draw (\yoko+\gaisen+\gaisen+\yoko+\gaisen+\yoko+\gaisen*2+\gaisen+\yoko+\yoko,\tate) -- (\yoko+\gaisen+\gaisen+\yoko+\gaisen+\yoko+\gaisen*2+\gaisen+\yoko+\gaisen+\yoko,\tate);
  \draw (\yoko+\gaisen+\gaisen+\yoko+\gaisen+\yoko+\gaisen*2,0) -- (\yoko+\gaisen+\gaisen+\yoko+\gaisen+\yoko+\gaisen*2+\gaisen+\yoko+\gaisen+\yoko,0);
  \draw (\yoko+\gaisen+\gaisen+\yoko+\gaisen+\yoko+\gaisen*2+\gaisen+\yoko+\gaisen+\yoko+\yoko/2,0) node [right] {$\Lambda_{\alpha\gamma}^{R(2)}$};
  \draw (\yoko+\gaisen+\gaisen+\yoko+\gaisen+\yoko+\gaisen*2,\tate/2+\tate/2) rectangle ++(\yoko,\tate);
  \draw (\yoko+\gaisen+\gaisen+\yoko+\gaisen+\yoko+\gaisen*2-\yoko-\yoko,\tate/2+\tate/2) -- (\yoko+\gaisen+\gaisen+\yoko+\gaisen+\yoko+\gaisen*2,\tate/2+\tate/2);
  \draw (\yoko+\gaisen+\gaisen+\yoko+\gaisen+\yoko+\gaisen*2+\yoko,\tate/2+\tate/2+\tate) node [right] {$\Lambda_{\alpha\beta\gamma}^{R(3)}$};
  \draw (\yoko+\gaisen+\gaisen+\yoko+\gaisen+\yoko+\gaisen*2+\yoko+\gaisen,\tate/2) rectangle ++(\yoko,\tate);
  \draw (\yoko+\gaisen+\gaisen+\yoko+\gaisen+\yoko+\gaisen*2+\yoko,\tate+\tate/2) -- (\yoko+\gaisen+\gaisen+\yoko+\gaisen+\yoko+\gaisen*2+\yoko+\gaisen,\tate+\tate/2);
  \draw (\yoko+\gaisen+\gaisen+\yoko+\gaisen+\yoko+\gaisen*2,\tate/2) -- (\yoko+\gaisen+\gaisen+\yoko+\gaisen+\yoko+\gaisen*2+\yoko+\gaisen,\tate/2);
  \draw (\yoko+\gaisen+\gaisen+\yoko+\gaisen+\yoko+\gaisen*2+\yoko+\gaisen+\yoko,\tate/2+\tate) node [right] {$\Lambda_{\alpha\gamma\delta}^{R(3)}$};
  \draw (\yoko+\gaisen+\yoko+\gaisen+\yoko+\yoko+\yoko,\tate+\tate/2) to [out=0,in=180] (\yoko+\gaisen+\gaisen+\yoko+\gaisen+\yoko+\gaisen*2-\yoko-\yoko,\tate/2+\tate/2);
\end{tikzpicture},
\end{eqnarray}
where the number of vertical lines is $\min(N_{\alpha\beta\gamma\delta}^{1},N_{\alpha\beta\gamma\delta}^{2})$. 
\end{definition}
Viewing the three-leg tensors as vertices, 
Eq.\ \eqref{eq:quad_inner_product} can be regarded as a tetrahedron as in 
Fig.\ \ref{fig:highereigeneq}(b) with some additional tensors.
For Eq.\ \eqref{eq:quad_inner_product} 
to serve as a valid definition of 
the quadruple inner produce and the higher Berry phase,
it must have certain properties.
In fact, we can show:
\begin{proposition}\;\label{thm:quad_inner}\\
\\
Let $\ket{\psi_{\alpha}}$, $\ket{\psi_{\beta}}$, 
$\ket{\psi_{\gamma}}$ and $\ket{\psi_{\delta}}$ be physically identical four states, and let $c_{\alpha\beta\gamma\delta}$ be the quadruple inner product of them. Then,\\
(i) $N_{\alpha\beta\gamma\delta}^{1}$ and $N_{\alpha\beta\gamma\delta}^{2}$ are finite natural number.\\
(ii) $c_{\alpha\beta\gamma\delta}$ is a non-zero complex number.\\
(iii) For physically identical five states $\ket{\psi_{\alpha}}$, $\ket{\psi_{\beta}}$, 
$\ket{\psi_{\gamma}}$, $\ket{\psi_{\delta}}$ and $\ket{\psi_{\eta}}$, $\{c_{\alpha\beta\gamma\delta}\}$ satisfies the cocycle condition
\begin{eqnarray}\label{eq:cocycle}
    c_{\alpha\beta\gamma\delta}c_{\alpha\beta\gamma\eta}^{-1}c_{\alpha\beta\delta\eta}c_{\alpha\gamma\delta\eta}^{-1}c_{\beta\gamma\delta\eta}=1.
\end{eqnarray}\\
(iv) Under the redefinition of the phase
\begin{eqnarray}
    \Lambda_{\alpha\beta\gamma}^{R(3)}\mapsto \lambda_{\alpha\beta\gamma}\Lambda_{\alpha\beta\gamma}^{R(3)},
    \quad 
    \Lambda_{\alpha\beta\gamma}^{L(3)}\mapsto \lambda_{\alpha\beta\gamma}^{-1}\Lambda_{\alpha\beta\gamma}^{L(3)},
\end{eqnarray}
the quadruple inner product $c_{\alpha\beta\gamma\delta}$ changes as 
\begin{eqnarray}
    c_{\alpha\beta\gamma\delta}\mapsto c_{\alpha\beta\gamma\delta}\lambda_{\alpha\beta\gamma}\lambda_{\alpha\beta\delta}^{-1}\lambda_{\alpha\gamma\delta}\lambda_{\beta\gamma\delta}^{-1}.
\end{eqnarray}
\end{proposition}
We will show this proposition in Sec.\ \ref{sec:proof}. 

%
%

This concludes our construction of the quadruple inner product and the higher Berry phase.
Furthermore, by using this proposition, we can construct a topological invariant of (families of) invertible states. 
We provide an outline of this construction and conclude this section.
\begin{description}
    \item[Application 1] {\bf SPT invariant}\\
    Let's consider the case where a semi-injective PEPS $\ket{A}=(\ket{\phi}_A,{O}_A)$ is 
    invariant under an onsite unitary $G$-symmetry ${U(g)}_{g\in G}$ for some finite group $G$. 
    If the state is defined on a closed surface (e.g. $2$-dimensional torus), it is completely identical before and after the action of the symmetry. However, $G$ action has a non-trivial effect on the virtual Hilbert space in general. 
    Let $\ket{A^{g}}$ be a semi-injective PEPS. 
    Let's denote the state obtained by acting the symmetry $U(g)$ on $\ket{A}$ by $\ket{A^{g}}$, that is, $\ket{A^{g}}:=(\ket{\phi}_A,U(g){O}_A)$. 
    Because of the symmetry, $\{\ket{A^{g}}\}_{g\in G}$ 
    represent identical states and we can evaluate the quadruple inner product for any quadruple states taken from $\{\ket{A^{g}}\}_{g\in G}$:
    \begin{eqnarray}
        c_{ghkl}:=\text{quadruple inner product of }\ket{A^g},\ket{A^h},\ket{A^k}\text{ and }\ket{A^l}.
    \end{eqnarray}
    From Prop.\ \ref{thm:quad_inner}, we see that $c_{ghkl}$ defines an element $[c_{ghkl}]$ of the group cohomology $\coho{3}{G}{\mathbb{C}^{\times}}$. It is known that the classification of $2$-dimensional SPT phases with $G$ symmetry is given by $\coho{3}{G}{\mathbb{C}^{\times}}$
    \cite{Chen_2011, Chen_2013} 
    and the quadruple inner product is an invariant for SPT phases. 

    \item[Application 2] {\bf higher Berry phase}\\
    Let's consider a family of gapped Hamiltonians over a parameter space $X$ and assume that the ground states admit semi-injective PEPS representations. Though the parent Hamiltonian is parametrized over $X$, in general, we cannot have 
    global parameterization of ground states. 
    For example, if the ordinary Berry phase is nontrivial, we cannot take a global phase fixing of the state. Similarly, the quadruple inner product measures the higher obstruction to take a global PEPS expression of the family. 
    In the next section, we will explain this point using  a concrete model.

    Let $\{U_{\alpha}\}$ be an open covering of the parameter space $X$, and we take a semi-injective PEPS representation
    $\ket{A_{\alpha}(x)}$
    $(x\in U_\alpha)$
    for each patch. Then, for each point of quadruple intersection $x\in U_{\alpha\beta\gamma\delta}:=U_{\alpha}\cap U_{\beta}\cap U_{\gamma}\cap U_{\delta}$, 
    we have physically identical semi-injective PEPS, $\ket{A_{\alpha}}$, $\ket{A_{\beta}}$, $\ket{A_{\gamma}}$ and $\ket{A_{\delta}}$. Thus, we can evaluate the quadruple inner product of them:
    \begin{eqnarray}
        c_{\alpha\beta\gamma\delta}:=\text{quadruple inner product of $\ket{A_{\alpha}}$, $\ket{A_{\beta}}$, $\ket{A_{\gamma}}$ and $\ket{A_{\delta}}$}.
    \end{eqnarray}
    From Prop.\ \ref{thm:quad_inner}, 
    we see that
    $c_{\alpha\beta\gamma\delta}$ defines an element $[c_{\alpha\beta\gamma\delta}]$ of the \v{C}eck cohomology $\coho{3}{X}{\mathrm{sh}(\mathbb{C}^{\times})}\simeq\coho{4}{X}{\mathbb{Z}}$. It is known that the higher Berry phase of $2$-dimensional invertible states takes its value in $\coho{3}{X}{\mathrm{sh}(\mathbb{C}^{\times})}\simeq\coho{4}{X}{\mathbb{Z}}$, and this implies that the quadruple inner product detects the higher Berry phase. We will carry out an explicit computation of the invariant in Sec.\ \ref{sec:Model parameterized over rp4} for the model parametrized by $\rp{4}$.
\end{description}

These two applications of the quadruple inner product are related as follows.
The basic idea is that SPT ground states (Hamiltonians)
may be realized at particular, special locations of some parameter space $X$. 
While generic invertible states in a given family may not 
respect any symmetry, 
symmetry enhancement can happen at these special points.
In the field theory context,
it is known that various SPT phases can be realized 
as a gapped phase of a quantum non-linear sigma model
\cite{Xu_2013,Bi_2015,Else_2014}.

In the SPT context, 
the gauge degrees of freedom of PEPS tensors we discussed 
are closely related to the symmetry protecting SPT phases.
While the symmetry 
leaves SPT ground states invariant, 
PEPS tensors undergo a gauge transformation. 
Thus, for the SPT models, 
invertible states in different gauge choices can be obtained by acting with the symmetry.
More precisely,
while the symmetry action does not change the state in the bulk, 
it can be implemented non-trivially 
on virtual bonds and 
on a boundary 
as non-onsite symmetry;
on the boundary, 
the symmetry operator 
is given by an MPU. 
Furthermore, 
this means that the eigen MPO of the mixed transfer 
matrix discussed above
are given by 
boundary symmetry unitary 
of the SPT phase. 

Ref.\ \cite{MGSC18} 
extracted a $3$-cocycle from $G$-symmetric semi-injective PEPS, and showed that it is an SPT invariant. 
In terms of the detection of SPT phases, the quadruple inner product and the invariants in Ref.\ \cite{MGSC18} 
are exactly the same quantity. Therefore, the quadruple inner product can be used as an invariant for SPT phases with general $G$ symmetry.

\subsection{The Construction of the Higher Berry Phase}
\label{sec:proof}

In this section, we will show 
Props.\ \ref{prop3} and \ref{thm:quad_inner}.
To this end, we first introduce 
and review the 
concept of reduction, 
which proves useful to construct 
the fixed points of various 
transfer matrices.
We will then read off 
the higher Berry phase as 
the phase that arises 
in the different constructions
of fixed point tensors.

\subsubsection{Reduction}

The concept of reduction, 
introduced in \cite{MGSC18}, is defined as an isometry connecting MPS or MPO 
of 
different virtual dimensions.
For example, 
let us consider
three different MPO,
$\{B^{ij}\}_{i,j=1}^{\mathsf{D}}$,
$\{C^{ij}\}_{i,j=1}^{\mathsf{D}}$
and
$\{D^{ij}\}_{i,j=1}^{\mathsf{D}}$.
Then, the set of isometries
$V, W$ and phase $c\in \mathbb{C}^\times$ 
satisfying 
\begin{eqnarray}\label{eq:reduce}
  V
  (\sum_{j}B^{i_{1}j}\otimes C^{jk_{1}})\cdots
  (\sum_{j}B^{i_{n}j}\otimes C^{jk_{n}})
  W
  =(c)^{n}
  D^{i_{1}j_{1}}\cdots D^{i_{n}j_{n}},
\end{eqnarray}
for any $n\in\mathbb{N}$ 
is a reduction from 
$\{\sum_{j}B^{ij}\otimes C^{jk}\}_{i,k=1}^{\mathsf{D}}$
to $\{D^{ij}\}_{i,j=1}^{\mathsf{D}}$.
We call $W$ and $V$  the right and left reductions,
respectively.
Let us also recall the following theorem regarding the redundancy of reductions:
\begin{theorem}(Them.\ 22 of \cite{MGSC18})\label{thm:uniquenessofreduction}\\
\\
Let $V,W$ and $\tilde{V},\tilde{W}$ be two reductions from $X^{i}$ to a normal tensor $Y^{i}$. Let $N_{0}$ be the maximal 
    nilpotency length of both reductions. Then there exists
    $\lambda\in\mathbb{C}$ such that 
\begin{align}
    VX^{i_1}\cdots X^{i_n}
    &=\lambda\tilde{V}X^{i_1}\cdots X^{i_n},\\
    X^{i_1}\cdots X^{i_n}W
    &=\lambda^{-1}X^{i_1}\cdots X^{i_n}\tilde{W},
\end{align}
for any $n >2N_0$. 
\end{theorem}
Here, the nilpotency length for a reduction $V,W$ from $X^i$ to $Y^i$ is defined as the minimal number $k$ so that 
\begin{eqnarray}
    (X^{i_1}-WY^{i_1}V)\cdots(X^{i_k}-WY^{i_k}V)=0,
\end{eqnarray}
for any configuration $(i_1,\ldots,i_k)$.
It is shown that this number is finite in Prop.\ 21 of \cite{MGSC18}. 
This theorem implies that, when multiplied by a sufficiently large number of $X^i$'s,
the reduction tensors become indistinguishable.
%
We define $N^i:=X^i-WY^iV$ and call $N^i$ the nilpotent part of $X^i$. Then $N^i$ satisfies
\begin{eqnarray}\label{eq:lemma27}
    VN^{i_1}\cdots N^{i_m}W=0
\end{eqnarray}
for any $m\in\mathbb{N}$.
This is also shown in Lem.\ 27 of \cite{MGSC18}.
By using these properties, we can show that the ``broken zipper condition'':
\begin{proposition}
\;\label{thm:prop5}\\
\\
Let $V,W$ be a reduction from $X^{i}$ to a normal tensor $Y^{i}$ and let $N_{0}$ be the nilpotency length of the reduction. Then, they satisfy
\begin{eqnarray}\label{eq:broken_zipper}
    X^{i_1}\cdots X^{i_m}X^{i_{m+1}}W=c X^{i_1}\cdots X^{i_m}WY^{i_{m+1}}
    \,
    \Longleftrightarrow 
    \,
\begin{tikzpicture}[line cap=round,line join=round,x=1.0cm,y=1.0cm, scale=0.3, baseline={([yshift=-.4ex]current bounding box.center)}, thick, shift={(0,0)},scale=0.7]
  \def\tate{3} 
  \def\yoko{0.5} 
  \def\sen{7} 
  \def\gaisen{2} 
  \def\nobishiro{1} 
  \filldraw[fill=black!30] (\yoko+\sen,0) rectangle ++(\yoko,\tate);
  \draw (\yoko,\tate) -- (\yoko+\sen,\tate);
  \draw (\yoko,0) -- (\yoko+\sen,0);
  \draw (\yoko+\sen+\yoko,\tate/2) -- (\yoko+\sen+\yoko+\gaisen,\tate/2);
  \draw (\yoko+\sen*0.25,\tate+\nobishiro) -- (\yoko+\sen*0.25,0-\nobishiro);
  \draw (\yoko+\sen*0.75,\tate+\nobishiro) -- (\yoko+\sen*0.75,0-\nobishiro);
  \draw (\yoko+\sen*0.875,\tate+\nobishiro) -- (\yoko+\sen*0.875,0-\nobishiro);
  \fill (\yoko+\sen*0.375,\tate+\nobishiro/2) circle (1mm); 
  \fill (\yoko+\sen*0.5,\tate+\nobishiro/2) circle (1mm); 
  \fill (\yoko+\sen*0.625,\tate+\nobishiro/2) circle (1mm); 
  \fill (\yoko+\sen*0.375,-\nobishiro/2) circle (1mm); 
  \fill (\yoko+\sen*0.5,-\nobishiro/2) circle (1mm); 
  \fill (\yoko+\sen*0.625,-\nobishiro/2) circle (1mm); 
\end{tikzpicture}
=
c\,\,
\begin{tikzpicture}[line cap=round,line join=round,x=1.0cm,y=1.0cm, scale=0.3, baseline={([yshift=-.4ex]current bounding box.center)}, thick, shift={(0,0)}, scale=0.7]
  \def\tate{3} 
  \def\yoko{0.5} 
  \def\sen{7} 
  \def\gaisen{2} 
  \def\nobishiro{1} 
  \filldraw[fill=black!30] (\yoko+\sen,0) rectangle ++(\yoko,\tate);
  \draw (\yoko,\tate) -- (\yoko+\sen,\tate);
  \draw (\yoko,0) -- (\yoko+\sen,0);
  \draw (\yoko+\sen+\yoko,\tate/2) -- (\yoko+\sen+\yoko+\gaisen,\tate/2);
  \draw (\yoko+\sen*0.25,\tate+\nobishiro) -- (\yoko+\sen*0.25,0-\nobishiro);
  \draw (\yoko+\sen*0.75,\tate+\nobishiro) -- (\yoko+\sen*0.75,0-\nobishiro);
  \draw (\yoko+\sen+\yoko+\gaisen/2,\tate+\nobishiro) -- (\yoko+\sen+\yoko+\gaisen/2,0-\nobishiro);
  \fill (\yoko+\sen*0.375,\tate+\nobishiro/2) circle (1mm); 
  \fill (\yoko+\sen*0.5,\tate+\nobishiro/2) circle (1mm); 
  \fill (\yoko+\sen*0.625,\tate+\nobishiro/2) circle (1mm); 
  \fill (\yoko+\sen*0.375,-\nobishiro/2) circle (1mm); 
  \fill (\yoko+\sen*0.5,-\nobishiro/2) circle (1mm); 
  \fill (\yoko+\sen*0.625,-\nobishiro/2) circle (1mm); 
\end{tikzpicture}
\end{eqnarray}
for any $m >N_0$. Here, the constant
$c$ is the same as that of
Eq.\ \eqref{eq:reduce}.
\end{proposition}
\begin{proof}
By using $X^i=WY^iV+N^i$,
    \begin{eqnarray}\label{eq:pre_broken_zipper}
        X^{i_1}\cdots X^{i_{m+1}}W=c X^{i_1}\cdots X^{i_m}WY^{i_{m+1}}+X^{i_1}\cdots X^{i_m}N^{i_{m+1}}W.
    \end{eqnarray}
    Using the equation $X^i=WY^iV+N^i$ once again,
    expand the second term of the above equation.
      Each time we pick $WY^iV$ in the expansion (even once),
      the corresponding contribution is zero because of
      Eq.\ \eqref{eq:lemma27}.
    If we choose the term $N^i$ from
    all factors, then that term becomes zero
    also since $m$ is greater than the
    nilpotency length $N_0$. Therefore, the second term of
    Eq.\ \eqref{eq:pre_broken_zipper} is zero, and
    Eq.\ \eqref{eq:broken_zipper} holds.\footnote{In the following, we will denote a reduction tensor as a gray box to distinguish it from the fixed point tensor.} 
\end{proof}
Therefore, if there is a sufficient number of $X^i$s, we can partially "zip"
the tensors by the reduction.
In particular, when the nilpotency length $N_0$ is zero, Eq.\ \eqref{eq:broken_zipper} is called the zipper condition, 
e.g. \cite{WBV17,G-RLM23}.


\subsubsection{Fixed Points of the $2$- and $3$-Transfer Matrices}
\label{sec:fp_2_3_transfer}

We can construct the fixed-point tensors of the $2$- and $3$-transfer 
matrices
as concatenations of some reduction
tensors as follows.
\begin{lemma}\;\label{thm:fp_as_reduction}\\
\\
(i) Let
$(V_{\alpha\beta},W_{\alpha\beta})$ be a reduction from
$\{\sum_j B^{ij}_{\alpha\beta}\otimes B^{jk}_{\beta\alpha}\}$
to $\{\delta^{ik}$\}.
Then, the right fixed point $\Lambda^{R(2)}_{\alpha\beta}$ of the $2$-transfer matrix can be written as
\begin{eqnarray}\label{eq:2_fp_from_reduction_right}
    \Lambda^{R(2)}_{\alpha\beta}=[T_{\alpha\beta}^{(2)}]^{N_{\alpha\beta}}\cdot W_{\alpha\beta}=
    \begin{tikzpicture}[line cap=round,line join=round,x=1.0cm,y=1.0cm, scale=0.3, baseline={([yshift=-.4ex]current bounding box.center)}, thick, shift={(0,0)},scale=0.7]
  \def\tate{3} 
  \def\yoko{0.5} 
  \def\sen{7} 
  \def\gaisen{2} 
  \def\nobishiro{1} 
  \filldraw[fill=black!30] (\yoko+\sen,0) rectangle ++(\yoko,\tate);
  \node[anchor=west] at (\yoko+\sen+\yoko,\tate/2)  {$W_{\alpha\beta}$};
  \draw (\yoko,\tate) -- (\yoko+\sen,\tate);
  \draw (\yoko,0) -- (\yoko+\sen,0);
  \draw (\yoko+\sen*0.25,\tate+\nobishiro) -- (\yoko+\sen*0.25,0-\nobishiro);
  \fill (\yoko+\sen*0.25,\tate+\nobishiro) circle (2mm); 
  \fill (\yoko+\sen*0.25,0-\nobishiro) circle (2mm);
  \draw (\yoko+\sen*0.75,\tate+\nobishiro) -- (\yoko+\sen*0.75,0-\nobishiro);
  \fill (\yoko+\sen*0.75,\tate+\nobishiro) circle (2mm); 
  \fill (\yoko+\sen*0.75,0-\nobishiro) circle (2mm);
  \fill (\yoko+\sen*0.375,\tate+\nobishiro/2) circle (1mm); 
  \fill (\yoko+\sen*0.5,\tate+\nobishiro/2) circle (1mm); 
  \fill (\yoko+\sen*0.625,\tate+\nobishiro/2) circle (1mm); 
  \fill (\yoko+\sen*0.375,-\nobishiro/2) circle (1mm); 
  \fill (\yoko+\sen*0.5,-\nobishiro/2) circle (1mm); 
  \fill (\yoko+\sen*0.625,-\nobishiro/2) circle (1mm); 
\end{tikzpicture},
\end{eqnarray}
where $N_{\alpha\beta}$ is the nilpotency
length of the reduction $(V_{\alpha\beta},W_{\alpha\beta})$. Similarly, the left fixed point $\Lambda^{L(2)}_{\alpha\beta}$ of the $2$-transfer matrix $T^{(2)}_{\alpha\beta}$ can be written as
\begin{eqnarray}\label{eq:2_fp_from_reduction_left}    \Lambda^{L(2)}_{\alpha\beta}=V_{\alpha\beta}\cdot [T_{\alpha\beta}^{(2)}]^{N_{\alpha\beta}}
    =
    \begin{tikzpicture}[line cap=round,line join=round,x=1.0cm,y=1.0cm, scale=0.3, baseline={([yshift=-.4ex]current bounding box.center)}, thick, shift={(0,0)}, scale=0.7]
  \def\tate{3} 
  \def\yoko{0.5} 
  \def\sen{7} 
  \def\gaisen{2} 
  \def\nobishiro{1} 
  \filldraw[fill=black!30] (0,0) rectangle ++(\yoko,\tate);
  \node[anchor=east] at (0,\tate/2)
  {$V_{\alpha\beta}$};
  \draw (\yoko,\tate) -- (\yoko+\sen,\tate);
  \draw (\yoko,0) -- (\yoko+\sen,0);
  \draw (\yoko+\sen*0.25,\tate+\nobishiro) -- (\yoko+\sen*0.25,0-\nobishiro);
  \fill (\yoko+\sen*0.25,\tate+\nobishiro) circle (2mm); 
  \fill (\yoko+\sen*0.25,0-\nobishiro) circle (2mm);
  \draw (\yoko+\sen*0.75,\tate+\nobishiro) -- (\yoko+\sen*0.75,0-\nobishiro);
  \fill (\yoko+\sen*0.75,\tate+\nobishiro) circle (2mm); 
  \fill (\yoko+\sen*0.75,0-\nobishiro) circle (2mm);
  \fill (\yoko+\sen*0.375,\tate+\nobishiro/2) circle (1mm); 
  \fill (\yoko+\sen*0.5,\tate+\nobishiro/2) circle (1mm); 
  \fill (\yoko+\sen*0.625,\tate+\nobishiro/2) circle (1mm); 
  \fill (\yoko+\sen*0.375,-\nobishiro/2) circle (1mm); 
  \fill (\yoko+\sen*0.5,-\nobishiro/2) circle (1mm); 
  \fill (\yoko+\sen*0.625,-\nobishiro/2) circle (1mm); 
\end{tikzpicture}\, .
\end{eqnarray}
(ii) Let $(V_{\alpha\beta\gamma},W_{\alpha\beta\gamma})$
be a reduction from
$\{\sum_{j}B^{ij}_{\alpha\beta}\otimes B^{jk}_{\beta\gamma}\}$
to
$\{B^{ik}_{\alpha\gamma}\}$. 
Then, the right fixed point
$\Lambda^{R(2,1)}_{\alpha\beta\gamma}$ of the right
$(2,1)$-transfer matrix $T^{R(2,1)}_{\alpha\beta\gamma}$ can be written as 
\begin{eqnarray}\label{eq:3_fp_from_reduction_right}
    \Lambda^{R(2,1)}_{\alpha\beta\gamma}&=&[T_{\alpha\beta\gamma}^{R(2,1)}]^{N_{\alpha\beta\gamma}}\cdot W_{\alpha\beta\gamma}=
    \begin{tikzpicture}[line cap=round,line join=round,x=1.0cm,y=1.0cm, scale=0.25, baseline={([yshift=-.6ex]current bounding box.center)}, thick, shift={(0,0)}, scale=0.9]
  \def\tate{3} 
  \def\yoko{0.5} 
  \def\sen{7} 
  \def\gaisen{8} 
  \pgfmathsetmacro\hankei{\tate*1.7} 
  \def\hankeinosa{2} 
  \filldraw[fill=black!30] (0,0) rectangle ++(\yoko,\tate);
  \node[anchor=west] at (\yoko,\tate) {$W_{\alpha\beta\gamma}$};
  \draw (-\gaisen,\tate) -- (0,\tate);
  \draw (-\gaisen,0) -- (0,0);
  \draw (\yoko,\tate/2) -- (\yoko+\gaisen,\tate/2);
  \draw (\yoko/2,\tate/2) circle (\hankei);
  \draw (\yoko/2,\tate/2) circle (\hankei+\hankeinosa);
  \node[anchor=west] at (\yoko/2+\hankei+\hankeinosa/8,\tate/2+\hankei+\hankeinosa/8) {$(T^{R(2,1)}_{\alpha\beta\gamma})^{N_{\alpha\beta\gamma}}$};
  \fill (\yoko/2,\tate/2+\hankei+\hankeinosa*0.25) circle (1mm); 
  \fill (\yoko/2,\tate/2+\hankei+\hankeinosa*0.5) circle (1mm); 
  \fill (\yoko/2,\tate/2+\hankei+\hankeinosa*0.75) circle (1mm); 
  \fill (\yoko/2,\tate/2-\hankei-\hankeinosa*0.25) circle (1mm); 
  \fill (\yoko/2,\tate/2-\hankei-\hankeinosa*0.5) circle (1mm); 
  \fill (\yoko/2,\tate/2-\hankei-\hankeinosa*0.75) circle (1mm); 
\end{tikzpicture},
\end{eqnarray}
where $N_{\alpha\beta\gamma}$ is the nilpotency length of $T_{\alpha\beta\gamma}^{R(2,1)}$.
Similarly, the left fixed point $\Lambda^{L(2,1)}_{\alpha\beta\gamma}$ of the left $(2,1)$-transfer matrix $T^{L(2,1)}_{\alpha\beta\gamma}$ can be written as 
\begin{eqnarray}\label{eq:3_fp_from_reduction_left}
    \Lambda^{L(2,1)}_{\alpha\beta\gamma}&=&V_{\alpha\beta\gamma}\cdot [T_{\alpha\beta\gamma}^{L(2,1)}]^{N_{\alpha\beta\gamma}}=
    \begin{tikzpicture}[line cap=round,line join=round,x=1.0cm,y=1.0cm, scale=0.25, baseline={([yshift=-.6ex]current bounding box.center)}, thick, shift={(0,0)}, scale=0.9]
  \def\tate{3} 
  \def\yoko{0.5} 
  \def\sen{7} 
  \def\gaisen{8} 
  \pgfmathsetmacro\hankei{\tate*1.7} 
  \def\hankeinosa{2} 
  \filldraw[fill=black!30] (0,0) rectangle ++(\yoko,\tate);
  \node[anchor=east] at (0,\tate) {$V_{\alpha\beta\gamma}$};
  \draw (\yoko,\tate) -- (\yoko+\gaisen,\tate);
  \draw (\yoko,0) -- (\yoko+\gaisen,0);
  \draw (-\gaisen,\tate/2) -- (0,\tate/2);
  \draw (\yoko/2,\tate/2) circle (\hankei);
  \draw (\yoko/2,\tate/2) circle (\hankei+\hankeinosa);
  \node[anchor=west] at (\yoko/2+\hankei+\hankeinosa/8,\tate/2+\hankei+\hankeinosa/8) {$(T^{L(2,1)}_{\alpha\beta\gamma})^{N_{\alpha\beta\gamma}}$};
  \fill (\yoko/2,\tate/2+\hankei+\hankeinosa*0.25) circle (1mm); 
  \fill (\yoko/2,\tate/2+\hankei+\hankeinosa*0.5) circle (1mm); 
  \fill (\yoko/2,\tate/2+\hankei+\hankeinosa*0.75) circle (1mm); 
  \fill (\yoko/2,\tate/2-\hankei-\hankeinosa*0.25) circle (1mm); 
  \fill (\yoko/2,\tate/2-\hankei-\hankeinosa*0.5) circle (1mm); 
  \fill (\yoko/2,\tate/2-\hankei-\hankeinosa*0.75) circle (1mm); 
\end{tikzpicture}.
\end{eqnarray}
\end{lemma}
\begin{proof} 
    (i) By multiplying the transfer matrix $T_{\alpha\beta}^{(2)}$ to the right-hand side of Eq.\ \eqref{eq:2_fp_from_reduction_right},
    \begin{eqnarray}
    T_{\alpha\beta}^{(2)}\cdot [T_{\alpha\beta}^{(2)}]^{N_{\alpha\beta}}\cdot W_{\alpha\beta}=
        \begin{tikzpicture}[line cap=round,line join=round,x=1.0cm,y=1.0cm, scale=0.3, baseline={([yshift=-.4ex]current bounding box.center)}, thick, shift={(0,0)},scale=0.7]
  \def\tate{3} 
  \def\yoko{0.5} 
  \def\sen{7} 
  \def\gaisen{2} 
  \def\nobishiro{1} 
  \filldraw[fill=black!30] (\yoko+\sen,0) rectangle ++(\yoko,\tate);
  \draw (\yoko,\tate) -- (\yoko+\sen,\tate);
  \draw (\yoko,0) -- (\yoko+\sen,0);
  \draw (\yoko+\sen*0.125,\tate+\nobishiro) -- (\yoko+\sen*0.125,0-\nobishiro);
  \fill (\yoko+\sen*0.125,\tate+\nobishiro) circle (2mm); 
  \fill (\yoko+\sen*0.125,0-\nobishiro) circle (2mm);
  \draw (\yoko+\sen*0.25,\tate+\nobishiro) -- (\yoko+\sen*0.25,0-\nobishiro);
  \fill (\yoko+\sen*0.25,\tate+\nobishiro) circle (2mm); 
  \fill (\yoko+\sen*0.25,0-\nobishiro) circle (2mm);
  \draw (\yoko+\sen*0.75,\tate+\nobishiro) -- (\yoko+\sen*0.75,0-\nobishiro);
  \fill (\yoko+\sen*0.75,\tate+\nobishiro) circle (2mm); 
  \fill (\yoko+\sen*0.75,0-\nobishiro) circle (2mm);
  \fill (\yoko+\sen*0.375,\tate+\nobishiro/2) circle (1mm); 
  \fill (\yoko+\sen*0.5,\tate+\nobishiro/2) circle (1mm); 
  \fill (\yoko+\sen*0.625,\tate+\nobishiro/2) circle (1mm); 
  \fill (\yoko+\sen*0.375,-\nobishiro/2) circle (1mm); 
  \fill (\yoko+\sen*0.5,-\nobishiro/2) circle (1mm); 
  \fill (\yoko+\sen*0.625,-\nobishiro/2) circle (1mm); 
\end{tikzpicture}
\,=\,
\begin{tikzpicture}[line cap=round,line join=round,x=1.0cm,y=1.0cm, scale=0.3, baseline={([yshift=-.4ex]current bounding box.center)}, thick, shift={(0,0)}, scale=0.7]
  \def\tate{3} 
  \def\yoko{0.5} 
  \def\sen{7} 
  \def\gaisen{2} 
  \def\nobishiro{1} 
  \filldraw[fill=black!30] (\yoko+\sen,0) rectangle ++(\yoko,\tate);
  \draw (\yoko,\tate) -- (\yoko+\sen,\tate);
  \draw (\yoko,0) -- (\yoko+\sen,0);
  \draw (\yoko+\sen*0.25,\tate+\nobishiro) -- (\yoko+\sen*0.25,0-\nobishiro);
  \fill (\yoko+\sen*0.25,\tate+\nobishiro) circle (2mm); 
  \fill (\yoko+\sen*0.25,0-\nobishiro) circle (2mm);
  \draw (\yoko+\sen*0.75,\tate+\nobishiro) -- (\yoko+\sen*0.75,0-\nobishiro);
  \fill (\yoko+\sen*0.75,\tate+\nobishiro) circle (2mm); 
  \fill (\yoko+\sen*0.75,0-\nobishiro) circle (2mm);
  \draw (\yoko+\sen+\yoko+\sen*0.25,\tate+\nobishiro) -- (\yoko+\sen+\yoko+\sen*0.25,0-\nobishiro);
  \fill (\yoko+\sen+\yoko+\sen*0.25,\tate+\nobishiro) circle (2mm); 
  \fill (\yoko+\sen+\yoko+\sen*0.25,0-\nobishiro) circle (2mm);
  \fill (\yoko+\sen*0.375,\tate+\nobishiro/2) circle (1mm); 
  \fill (\yoko+\sen*0.5,\tate+\nobishiro/2) circle (1mm); 
  \fill (\yoko+\sen*0.625,\tate+\nobishiro/2) circle (1mm); 
  \fill (\yoko+\sen*0.375,-\nobishiro/2) circle (1mm); 
  \fill (\yoko+\sen*0.5,-\nobishiro/2) circle (1mm); 
  \fill (\yoko+\sen*0.625,-\nobishiro/2) circle (1mm); 
\end{tikzpicture}
\, =\,[T_{\alpha\beta}^{(2)}]^{N_{\alpha\beta}}\cdot W_{\alpha\beta}.
\end{eqnarray}
Here, we used the broken zipper equation Eq.\ \eqref{eq:broken_zipper}. 
Thus, the tensor is the fixed point unless it is the zero tensor.
By multiplying left reduction $V_{\alpha\beta}$ to the right-hand side of
Eq.\ \eqref{eq:2_fp_from_reduction_right}, 
the value of the diagram turns out to be the size of the matrices $\{B^{ij}_{\alpha\beta}\}$:
\begin{eqnarray}
        V_{\alpha\beta}\Lambda^{R(2)}_{\alpha\beta}=
        \begin{tikzpicture}[line cap=round,line join=round,x=1.0cm,y=1.0cm, scale=0.3, baseline={([yshift=-.4ex]current bounding box.center)}, thick, shift={(0,0)}, scale=0.7]
  \def\tate{3} 
  \def\yoko{0.5} 
  \def\sen{7} 
  \def\gaisen{2} 
  \def\nobishiro{1} 
  \filldraw[fill=black!30] (0,0) rectangle ++(\yoko,\tate);
  \node[anchor=east] at (0,\tate/2)
  {$V_{\alpha\beta}$};
  \filldraw[fill=black!30] (\yoko+\sen,0) rectangle ++(\yoko,\tate);
  \node[anchor=west] at (\yoko+\sen+\yoko,\tate/2)  {$W_{\alpha\beta}$};
  \draw (\yoko,\tate) -- (\yoko+\sen,\tate);
  \draw (\yoko,0) -- (\yoko+\sen,0);
  \draw (\yoko+\sen*0.25,\tate+\nobishiro) -- (\yoko+\sen*0.25,0-\nobishiro);
  \fill (\yoko+\sen*0.25,\tate+\nobishiro) circle (2mm); 
  \fill (\yoko+\sen*0.25,0-\nobishiro) circle (2mm);
  \draw (\yoko+\sen*0.75,\tate+\nobishiro) -- (\yoko+\sen*0.75,0-\nobishiro);
  \fill (\yoko+\sen*0.75,\tate+\nobishiro) circle (2mm); 
  \fill (\yoko+\sen*0.75,0-\nobishiro) circle (2mm);
  \fill (\yoko+\sen*0.375,\tate+\nobishiro/2) circle (1mm); 
  \fill (\yoko+\sen*0.5,\tate+\nobishiro/2) circle (1mm); 
  \fill (\yoko+\sen*0.625,\tate+\nobishiro/2) circle (1mm); 
  \fill (\yoko+\sen*0.375,-\nobishiro/2) circle (1mm); 
  \fill (\yoko+\sen*0.5,-\nobishiro/2) circle (1mm); 
  \fill (\yoko+\sen*0.625,-\nobishiro/2) circle (1mm); 
\end{tikzpicture}
=
\begin{tikzpicture}[line cap=round,line join=round,x=1.0cm,y=1.0cm, scale=0.3, baseline={([yshift=-.4ex]current bounding box.center)}, thick, shift={(0,0)}, scale=0.7]
  \def\tate{3} 
  \def\yoko{0.5} 
  \def\sen{3} 
  \def\gaisen{2} 
  \def\nobishiro{1} 
  \def\shift{3} 
  \def\sensen{6} 
  \filldraw[fill=black!30] (0,0) rectangle ++(\yoko,\tate);
  \filldraw[fill=black!30] (\yoko+\sen,0) rectangle ++(\yoko,\tate);
  \draw (\yoko,\tate) -- (\yoko+\sen,\tate);
  \draw (\yoko,0) -- (\yoko+\sen,0);
  \draw (\yoko+\sensen*0.25+\shift,\tate+\nobishiro) -- (\yoko+\sensen*0.25+\shift,0-\nobishiro);
  \fill (\yoko+\sensen*0.25+\shift,\tate+\nobishiro) circle (2mm); 
  \fill (\yoko+\sensen*0.25+\shift,0-\nobishiro) circle (2mm);
  \draw (\yoko+\sensen*0.75+\shift,\tate+\nobishiro) -- (\yoko+\sensen*0.75+\shift,0-\nobishiro);
  \fill (\yoko+\sensen*0.75+\shift,\tate+\nobishiro) circle (2mm); 
  \fill (\yoko+\sensen*0.75+\shift,0-\nobishiro) circle (2mm);
  \fill (\yoko+\sensen*0.375+\shift,\tate/2) circle (1mm); 
  \fill (\yoko+\sensen*0.5+\shift,\tate/2) circle (1mm); 
  \fill (\yoko+\sensen*0.625+\shift,\tate/2) circle (1mm); 
\end{tikzpicture}
=1.
\end{eqnarray}
This implies that the tensor is nonzero. 
Therefore, the right-hand side of Eq.\ \eqref{eq:2_fp_from_reduction_right}
is the fixed point of the transfer matrix $T_{\alpha\beta}^{(2)}$. 
Similarly, we can show that the right-hand side of 
Eq.\ \eqref{eq:2_fp_from_reduction_left}
is the left fixed point of the transfer matrix $T_{\alpha\beta}^{(2)}$.\\
    (ii) By multiplying the transfer matrix $T_{\alpha\beta\gamma}^{R(2,1)}$ to
    the right-hand side of Eq.\ (\ref{eq:3_fp_from_reduction_right}),    
    \begin{eqnarray}
    T_{\alpha\beta\gamma}^{R(2,1)}\cdot [T_{\alpha\beta\gamma}^{R(2,1)}]^{N_{\alpha\beta\gamma}}\cdot W_{\alpha\beta\gamma}&=&
    \begin{tikzpicture}[line cap=round,line join=round,x=1.0cm,y=1.0cm, scale=0.25, baseline={([yshift=-.6ex]current bounding box.center)}, thick, shift={(0,0)}, scale=0.7]
  \def\tate{3} 
  \def\yoko{0.5} 
  \def\sen{7} 
  \def\gaisen{8} 
  \pgfmathsetmacro\hankei{\tate*1.7} 
  \def\hankeinosa{2} 
  \filldraw[fill=black!30] (0,0) rectangle ++(\yoko,\tate);
  \draw (-\gaisen,\tate) -- (0,\tate);
  \draw (-\gaisen,0) -- (0,0);
  \draw (\yoko,\tate/2) -- (\yoko+\gaisen,\tate/2);
  \draw (\yoko/2,\tate/2) circle (\hankei);
  \draw (\yoko/2,\tate/2) circle (\hankei+\hankeinosa);
  \draw (\yoko/2,\tate/2) circle (\hankei+\hankeinosa+0.5);
  \fill (\yoko/2,\tate/2+\hankei+\hankeinosa*0.25) circle (1mm); 
  \fill (\yoko/2,\tate/2+\hankei+\hankeinosa*0.5) circle (1mm); 
  \fill (\yoko/2,\tate/2+\hankei+\hankeinosa*0.75) circle (1mm); 
  \fill (\yoko/2,\tate/2-\hankei-\hankeinosa*0.25) circle (1mm); 
  \fill (\yoko/2,\tate/2-\hankei-\hankeinosa*0.5) circle (1mm); 
  \fill (\yoko/2,\tate/2-\hankei-\hankeinosa*0.75) circle (1mm); 
\end{tikzpicture}
=
\begin{tikzpicture}[line cap=round,line join=round,x=1.0cm,y=1.0cm, scale=0.25, baseline={([yshift=-.6ex]current bounding box.center)}, thick, shift={(0,0)}, scale=0.7]
  \def\tate{3} 
  \def\yoko{0.5} 
  \def\sen{7} 
  \def\gaisen{8} 
  \pgfmathsetmacro\hankei{\tate*1.7} 
  \def\hankeinosa{2} 
  \filldraw[fill=black!30] (0,0) rectangle ++(\yoko,\tate);
  \draw (-\gaisen,\tate) -- (0,\tate);
  \draw (-\gaisen,0) -- (0,0);
  \draw (\yoko,\tate/2) -- (\yoko+\gaisen,\tate/2);
  \draw (\yoko/2,\tate/2) circle (\hankei);
  \draw (\yoko/2,\tate/2) circle (\hankei+\hankeinosa);
  \draw (\yoko/2,\tate/2) circle (\hankei-0.5);
  \fill (\yoko/2,\tate/2+\hankei+\hankeinosa*0.25) circle (1mm); 
  \fill (\yoko/2,\tate/2+\hankei+\hankeinosa*0.5) circle (1mm); 
  \fill (\yoko/2,\tate/2+\hankei+\hankeinosa*0.75) circle (1mm); 
  \fill (\yoko/2,\tate/2-\hankei-\hankeinosa*0.25) circle (1mm); 
  \fill (\yoko/2,\tate/2-\hankei-\hankeinosa*0.5) circle (1mm); 
  \fill (\yoko/2,\tate/2-\hankei-\hankeinosa*0.75) circle (1mm); 
\end{tikzpicture}
\nonumber \\
&=&
\begin{tikzpicture}[line cap=round,line join=round,x=1.0cm,y=1.0cm, scale=0.25, baseline={([yshift=-.6ex]current bounding box.center)}, thick, shift={(0,0)}, scale=0.7]
  \def\tate{3} 
  \def\yoko{0.5} 
  \def\sen{7} 
  \def\gaisen{8} 
  \pgfmathsetmacro\hankei{\tate*1.7} 
  \def\hankeinosa{2} 
  \filldraw[fill=black!30] (0,0) rectangle ++(\yoko,\tate);
  \draw (-\gaisen,\tate) -- (0,\tate);
  \draw (-\gaisen,0) -- (0,0);
  \draw (\yoko,\tate/2) -- (\yoko+\gaisen,\tate/2);
  \draw (\yoko/2,\tate/2) circle (\hankei);
  \draw (\yoko/2,\tate/2) circle (\hankei+\hankeinosa);
  \draw (\yoko/2+\hankei/2,\tate/2) circle (\hankei/3);
  \fill (\yoko/2,\tate/2+\hankei+\hankeinosa*0.25) circle (1mm); 
  \fill (\yoko/2,\tate/2+\hankei+\hankeinosa*0.5) circle (1mm); 
  \fill (\yoko/2,\tate/2+\hankei+\hankeinosa*0.75) circle (1mm); 
  \fill (\yoko/2,\tate/2-\hankei-\hankeinosa*0.25) circle (1mm); 
  \fill (\yoko/2,\tate/2-\hankei-\hankeinosa*0.5) circle (1mm); 
  \fill (\yoko/2,\tate/2-\hankei-\hankeinosa*0.75) circle (1mm); 
\end{tikzpicture}
\nonumber 
=
\begin{tikzpicture}[line cap=round,line join=round,x=1.0cm,y=1.0cm, scale=0.25, baseline={([yshift=-.6ex]current bounding box.center)}, thick, shift={(0,0)}, scale=0.7]
  \def\tate{3} 
  \def\yoko{0.5} 
  \def\sen{7} 
  \def\gaisen{8} 
  \pgfmathsetmacro\hankei{\tate*1.7} 
  \def\hankeinosa{2} 
  \filldraw[fill=black!30] (0,0) rectangle ++(\yoko,\tate);
  \draw (-\gaisen,\tate) -- (0,\tate);
  \draw (-\gaisen,0) -- (0,0);
  \draw (\yoko,\tate/2) -- (\yoko+\gaisen,\tate/2);
  \draw (\yoko/2,\tate/2) circle (\hankei);
  \draw (\yoko/2,\tate/2) circle (\hankei+\hankeinosa);
  \fill (\yoko/2,\tate/2+\hankei+\hankeinosa*0.25) circle (1mm); 
  \fill (\yoko/2,\tate/2+\hankei+\hankeinosa*0.5) circle (1mm); 
  \fill (\yoko/2,\tate/2+\hankei+\hankeinosa*0.75) circle (1mm); 
  \fill (\yoko/2,\tate/2-\hankei-\hankeinosa*0.25) circle (1mm); 
  \fill (\yoko/2,\tate/2-\hankei-\hankeinosa*0.5) circle (1mm); 
  \fill (\yoko/2,\tate/2-\hankei-\hankeinosa*0.75) circle (1mm); 
\end{tikzpicture}\\
&=&[T_{\alpha\beta\gamma}^{R(2,1)}]^{N_{\alpha\beta\gamma}}\cdot W_{\alpha\beta\gamma}.
\end{eqnarray}
Here, we used the broken zipper equation Eq.\ \eqref{eq:broken_zipper}. 
Thus, the tensor is the fixed point unless it is the zero tensor. By multiplying left reduction $V_{\alpha\beta\gamma}$ to the right-hand side of Eq.\ \eqref{eq:3_fp_from_reduction_right}, 
the value of the diagram turns out to be the identity matrix: 
\begin{eqnarray}
        V_{\alpha\beta\gamma}\Lambda^{R(2,1)}_{\alpha\beta\gamma}=
        \begin{tikzpicture}[line cap=round,line join=round,x=1.0cm,y=1.0cm, scale=0.25, baseline={([yshift=-.6ex]current bounding box.center)}, thick, shift={(0,0)}, scale=0.9]
  \def\tate{3} 
  \def\yoko{0.5} 
  \def\sen{7} 
  \def\gaisen{8} 
  \pgfmathsetmacro\hankei{\tate*1.7} 
  \def\hankeinosa{2} 
  \filldraw[fill=black!30] (-\gaisen-\yoko,0) rectangle ++(\yoko,\tate);
  \node[anchor=east] at (-\yoko-\gaisen,\tate)
  {$V_{\alpha\beta\gamma}$};
  \draw (-\gaisen-\yoko-2,\tate/2) -- (-\gaisen-\yoko,\tate/2);
  \filldraw[fill=black!30] (0,0) rectangle ++(\yoko,\tate);
  \node[anchor=west] at (\yoko,\tate) {$W_{\alpha\beta\gamma}$};
  \draw (-\gaisen,\tate) -- (0,\tate);
  \draw (-\gaisen,0) -- (0,0);
  \draw (\yoko,\tate/2) -- (\yoko+\gaisen,\tate/2);
  \draw (\yoko/2,\tate/2) circle (\hankei);
  \draw (\yoko/2,\tate/2) circle (\hankei+\hankeinosa);
  \fill (\yoko/2,\tate/2+\hankei+\hankeinosa*0.25) circle (1mm); 
  \fill (\yoko/2,\tate/2+\hankei+\hankeinosa*0.5) circle (1mm); 
  \fill (\yoko/2,\tate/2+\hankei+\hankeinosa*0.75) circle (1mm); 
  \fill (\yoko/2,\tate/2-\hankei-\hankeinosa*0.25) circle (1mm); 
  \fill (\yoko/2,\tate/2-\hankei-\hankeinosa*0.5) circle (1mm); 
  \fill (\yoko/2,\tate/2-\hankei-\hankeinosa*0.75) circle (1mm); 
\end{tikzpicture}
=
\begin{tikzpicture}[line cap=round,line join=round,x=1.0cm,y=1.0cm, scale=0.25, baseline={([yshift=-.6ex]current bounding box.center)}, thick, shift={(0,0)}]
  \def\tate{1.5} 
  \def\yoko{0.5} 
  \def\sen{5} 
  \def\gaisen{8} 
  \pgfmathsetmacro\hankei{\tate*1.7} 
  \def\hankeinosa{2} 
  \draw (\yoko-\gaisen,\tate/2) -- (\yoko+\gaisen,\tate/2);
  \draw (\yoko/2,\tate/2) circle (\hankei);
  \draw (\yoko/2,\tate/2) circle (\hankei+\hankeinosa);
  \fill (\yoko/2,\tate/2+\hankei+\hankeinosa*0.25) circle (1mm); 
  \fill (\yoko/2,\tate/2+\hankei+\hankeinosa*0.5) circle (1mm); 
  \fill (\yoko/2,\tate/2+\hankei+\hankeinosa*0.75) circle (1mm); 
  \fill (\yoko/2,\tate/2-\hankei-\hankeinosa*0.25) circle (1mm); 
  \fill (\yoko/2,\tate/2-\hankei-\hankeinosa*0.5) circle (1mm); 
  \fill (\yoko/2,\tate/2-\hankei-\hankeinosa*0.75) circle (1mm); 
\end{tikzpicture}
=
\begin{tikzpicture}[line cap=round,line join=round,x=1.0cm,y=1.0cm, scale=0.25, baseline={([yshift=-.6ex]current bounding box.center)}, thick, shift={(0,0)}]
  \def\tate{1.5} 
  \def\yoko{0.5} 
  \def\sen{5} 
  \def\gaisen{3} 
  \pgfmathsetmacro\hankei{\tate*1.7} 
  \def\hankeinosa{2} 
  \draw (\yoko-\gaisen,\tate/2) -- (\yoko+\gaisen,\tate/2);
\end{tikzpicture}
.
\end{eqnarray}
This implies that the tensor is nonzero. 
Therefore, the right-hand side of 
Eq.\ \eqref{eq:3_fp_from_reduction_right} is the fixed point of the transfer matrix $T_{\alpha\beta\gamma}^{R(2,1)}$. 
Similarly, we can show that the right-hand side of 
Eq.\ \eqref{eq:3_fp_from_reduction_left} is the fixed point of the transfer matrix $T_{\alpha\beta\gamma}^{L(2,1)}$.
\end{proof}
Remark that, as a consequence of Lem.\ \ref{thm:fp_as_reduction}, the fixed point $(\Lambda^{L(2,1)}_{\alpha\beta\gamma},\Lambda^{R(2,1)}_{\alpha\beta\gamma})$ itself is a reduction pair.

\subsubsection{Fixed Points of the $4$-Transfer Matrix}

We can construct the fixed-point tensor of the $4$-transfer matrix as a concatenation of the fixed points of the $2$- and $3$-transfer matrices:
\begin{lemma}\label{thm:4_fp_as_reduction}\;\\
\\
(i) Let $\Lambda^{R(4)}_{\alpha\beta\gamma\delta}$ be a $4$-leg tensor defined by
    \begin{eqnarray}\label{eq:4_fp_right}
        \Lambda^{R(4)}_{\alpha\beta\gamma\delta}=[T^{(4)}_{\alpha\beta\gamma}]^{N_{\alpha\beta\gamma\delta}}\Lambda^{R(2,1)}_{\alpha\beta\gamma}\Lambda^{R(2,1)}_{\alpha\gamma\delta}\Lambda^{R}_{\alpha\delta}=
        \begin{tikzpicture}[line cap=round,line join=round,x=1.0cm,y=1.0cm, scale=0.3, baseline={([yshift=-.6ex]current bounding box.center)}, thick, shift={(0,0)}, scale=0.7]
  \def\tate{3} 
  \def\yoko{0.5} 
  \def\sen{7} 
  \def\gaisen{2} 
  \pgfmathsetmacro\hankei{\tate*1.7} 
  \def\hankeinosa{2} 
  \def\nobishiro{1}
  \draw (0,0) rectangle ++(\yoko,\tate);
  \draw (-\gaisen,\tate) -- (0,\tate);
  \draw (-\gaisen-\gaisen-\yoko-\gaisen-\yoko-\gaisen*2,0) -- (0,0);
  \draw (-\gaisen-\yoko,\tate/2) rectangle ++(\yoko,\tate);
  \draw (-\gaisen-\gaisen-\yoko,\tate+\tate/2) -- (-\gaisen-\yoko,\tate+\tate/2);
  \draw (-\gaisen-\gaisen-\yoko-\gaisen-\yoko-\gaisen*2,\tate/2) -- (-\gaisen-\yoko,\tate/2);
  \draw (-\gaisen-\yoko-\gaisen-\yoko,\tate/2+\tate/2) rectangle ++(\yoko,\tate);
  \draw (-\gaisen-\gaisen-\yoko-\gaisen-\yoko-\gaisen*2,\tate+\tate/2+\tate/2) -- (-\gaisen-\yoko-\gaisen-\yoko,\tate+\tate/2+\tate/2);
  \draw (-\gaisen-\gaisen-\yoko-\gaisen-\yoko-\gaisen*2,\tate/2+\tate/2) -- (-\gaisen-\yoko-\gaisen-\yoko,\tate/2+\tate/2);
  \draw (-\gaisen-\gaisen-\yoko-\gaisen-\yoko-\gaisen*2+3*\gaisen*0.25,\tate+\tate/2+\tate/2+\nobishiro) -- (-\gaisen-\gaisen-\yoko-\gaisen-\yoko-\gaisen*2+3*\gaisen*0.25,0-\nobishiro);
  \fill (-\gaisen-\gaisen-\yoko-\gaisen-\yoko-\gaisen*2+3*\gaisen*0.25,\tate+\tate/2+\tate/2+\nobishiro) circle (2mm); 
  \fill (-\gaisen-\gaisen-\yoko-\gaisen-\yoko-\gaisen*2+3*\gaisen*0.25,0-\nobishiro) circle (2mm);
  \draw (-\gaisen-\gaisen-\yoko-\gaisen-\yoko-\gaisen*2+3*\gaisen*0.75,\tate+\tate/2+\tate/2+\nobishiro) -- (-\gaisen-\gaisen-\yoko-\gaisen-\yoko-\gaisen*2+3*\gaisen*0.75,0-\nobishiro);
  \fill (-\gaisen-\gaisen-\yoko-\gaisen-\yoko-\gaisen*2+3*\gaisen*0.75,\tate+\tate/2+\tate/2+\nobishiro) circle (2mm); 
  \fill (-\gaisen-\gaisen-\yoko-\gaisen-\yoko-\gaisen*2+3*\gaisen*0.75,0-\nobishiro) circle (2mm);
  \fill (-\gaisen-\gaisen-\yoko-\gaisen-\yoko-\gaisen*2+3*\gaisen*0.375,\tate+\tate/2+\tate/2+\nobishiro/2) circle (1mm); 
  \fill (-\gaisen-\gaisen-\yoko-\gaisen-\yoko-\gaisen*2+3*\gaisen*0.5,\tate+\tate/2+\tate/2+\nobishiro/2) circle (1mm); 
  \fill (-\gaisen-\gaisen-\yoko-\gaisen-\yoko-\gaisen*2+3*\gaisen*0.625,\tate+\tate/2+\tate/2+\nobishiro/2) circle (1mm); 
  \fill (-\gaisen-\gaisen-\yoko-\gaisen-\yoko-\gaisen*2+3*\gaisen*0.375,-\nobishiro/2) circle (1mm); 
  \fill (-\gaisen-\gaisen-\yoko-\gaisen-\yoko-\gaisen*2+3*\gaisen*0.5,-\nobishiro/2) circle (1mm); 
  \fill (-\gaisen-\gaisen-\yoko-\gaisen-\yoko-\gaisen*2+3*\gaisen*0.625,-\nobishiro/2) circle (1mm); 
\end{tikzpicture} 
\, ,
\end{eqnarray}
where $N_{\alpha\beta\gamma\delta}$ is the nilpotency length of the reduction. Then, $\Lambda^{R(4)}_{\alpha\beta\gamma\delta}$ is the right fixed point of the $4$-transfer matrix $T^{(4)}_{\alpha\beta\gamma\delta}$. 
    \\
    (ii) Let $\Lambda^{L(4)}_{\alpha\beta\gamma\delta}$ be a $4$-leg tensor defined by
    \begin{eqnarray}\label{eq:4_fp_left}
        \Lambda^{L(4)}_{\alpha\beta\gamma\delta}=\Lambda^{L}_{\alpha\delta}\Lambda^{L(2,1)}_{\alpha\gamma\delta}\Lambda^{L(2,1)}_{\alpha\beta\gamma}[T^{(4)}_{\alpha\beta\gamma}]^{N_{\alpha\beta\gamma\delta}}=
.
    \end{eqnarray}
     Then, $\Lambda^{L(4)}_{\alpha\beta\gamma\delta}$ is the left fixed point of the $4$-transfer matrix $T^{(4)}_{\alpha\beta\gamma\delta}$. 
\end{lemma}
\begin{proof}
    (i) By multiplying the $4$-transfer matrix $T^{(4)}_{\alpha\beta\gamma\delta}$ to the right-hand side of Eq.\ \eqref{eq:4_fp_right},
    \begin{eqnarray}
        %
\, .
\end{eqnarray}
Here, we used the broken zipper equation Eq.\ \eqref{eq:broken_zipper}. 
    Therefore, we can zip the top fixed point tensor until the number of remaining transfer matrices becomes $N_{\alpha\beta\gamma}$:
\begin{eqnarray}
        %
\, .
    \end{eqnarray}
    Similarly, we can zip the middle fixed point tensor until the number of the remaining transfer matrices between the top and middle tensors becomes $N_{\alpha\gamma\delta}$:
\begin{eqnarray}
        %
\, .
    \end{eqnarray}
    Since there are $N_{\alpha\beta}$ transfer matrices remaining between the middle and the rightmost fixed point tensor, it is possible to zip the rightmost transfer matrix:
\begin{eqnarray}
        %
.
    \end{eqnarray}    
    Consequently, the tensor is the fixed point unless it is the zero tensor.
    By multiplying left reduction tensors to the right-hand side of 
    Eq.\ \eqref{eq:4_fp_right}, 
    the value of the diagram turns out to be the size of the matrices $\{B^{ij}_{\alpha\beta}\}$:
    \begin{eqnarray}
        %
\, .
\end{eqnarray}
Similarly, we can zip the second and the third reduction tensors. 
Consequently, the left-hand side turns out to be $1$:
    \begin{eqnarray}
        %
=1.
\end{eqnarray}
    This implies that the tensor is nonzero. 
    Therefore, the right-hand side of Eq.\ \eqref{eq:4_fp_right} 
    is the right fixed point of the transfer matrix $T^{(4)}_{\alpha\beta\gamma\delta}$.\\
    (ii) Similarly, we can show that the right-hand side of Eq.\ \eqref{eq:4_fp_left} 
    is the left fixed point of the transfer matrix $T^{(4)}_{\alpha\beta\gamma\delta}$.
\end{proof}
By using the fixed point tensors, we can construct the fixed point of the $4$-transfer matrix $T_{\alpha\beta\gamma\delta}^{(4)}$ in a different manner:
\begin{lemma}\;\\
\\
(i) Let $\tilde{\Lambda}^{R(4)}_{\alpha\beta\gamma\delta}$ be a $4$-leg tensor defined by
    \begin{eqnarray}\label{eq:4_fp_right}
        \tilde{\Lambda}^{R(4)}_{\alpha\beta\gamma\delta}:=T^{(4)}_{\alpha\beta\gamma}\Lambda^{R(2,1)}_{\beta\gamma\delta}\Lambda^{R(2,1)}_{\alpha\beta\delta}\Lambda^{R}_{\alpha\delta}
        =
        \begin{tikzpicture}[line cap=round,line join=round,x=1.0cm,y=1.0cm, scale=0.3, baseline={([yshift=-.6ex]current bounding box.center)}, thick, shift={(0,0)}, scale=0.7]
  \def\tate{3} 
  \def\yoko{0.5} 
  \def\sen{7} 
  \def\gaisen{2} 
  \pgfmathsetmacro\hankei{\tate*1.7} 
  \def\hankeinosa{2} 
  \def\nobishiro{1}
  \draw (0,0) rectangle ++(\yoko,\tate+\tate/2);
  \draw (-\gaisen,\tate+\tate/2) -- (0,\tate+\tate/2);
  \draw (-\gaisen-\gaisen-\yoko-\gaisen-\yoko-\gaisen*2,0) -- (0,0);
  \draw (-\gaisen-\yoko,\tate/2+\tate/2) rectangle ++(\yoko,\tate);
  \draw (-\gaisen-\gaisen-\yoko,\tate/2+\tate/2) -- (-\gaisen-\yoko,\tate/2+\tate/2);
  \draw (-\gaisen-\gaisen-\yoko-\gaisen-\yoko-\gaisen*2,\tate+\tate/2+\tate/2) -- (-\gaisen-\yoko,\tate+\tate/2+\tate/2);
  \draw (-\gaisen-\yoko-\gaisen-\yoko,\tate/2) rectangle ++(\yoko,\tate);
  \draw (-\gaisen-\gaisen-\yoko-\gaisen-\yoko-\gaisen*2,\tate+\tate/2) -- (-\gaisen-\yoko-\gaisen-\yoko,\tate+\tate/2);
  \draw (-\gaisen-\gaisen-\yoko-\gaisen-\yoko-\gaisen*2,\tate/2) -- (-\gaisen-\yoko-\gaisen-\yoko,\tate/2);
  \draw (-\gaisen-\gaisen-\yoko-\gaisen-\yoko-\gaisen*2+3*\gaisen*0.25,\tate+\tate/2+\tate/2+\nobishiro) -- (-\gaisen-\gaisen-\yoko-\gaisen-\yoko-\gaisen*2+3*\gaisen*0.25,0-\nobishiro);
  \fill (-\gaisen-\gaisen-\yoko-\gaisen-\yoko-\gaisen*2+3*\gaisen*0.25,\tate+\tate/2+\tate/2+\nobishiro) circle (2mm); 
  \fill (-\gaisen-\gaisen-\yoko-\gaisen-\yoko-\gaisen*2+3*\gaisen*0.25,0-\nobishiro) circle (2mm);
  \draw (-\gaisen-\gaisen-\yoko-\gaisen-\yoko-\gaisen*2+3*\gaisen*0.75,\tate+\tate/2+\tate/2+\nobishiro) -- (-\gaisen-\gaisen-\yoko-\gaisen-\yoko-\gaisen*2+3*\gaisen*0.75,0-\nobishiro);
  \fill (-\gaisen-\gaisen-\yoko-\gaisen-\yoko-\gaisen*2+3*\gaisen*0.75,\tate+\tate/2+\tate/2+\nobishiro) circle (2mm); 
  \fill (-\gaisen-\gaisen-\yoko-\gaisen-\yoko-\gaisen*2+3*\gaisen*0.75,0-\nobishiro) circle (2mm);
  \fill (-\gaisen-\gaisen-\yoko-\gaisen-\yoko-\gaisen*2+3*\gaisen*0.375,\tate+\tate/2+\tate/2+\nobishiro/2) circle (1mm); 
  \fill (-\gaisen-\gaisen-\yoko-\gaisen-\yoko-\gaisen*2+3*\gaisen*0.5,\tate+\tate/2+\tate/2+\nobishiro/2) circle (1mm); 
  \fill (-\gaisen-\gaisen-\yoko-\gaisen-\yoko-\gaisen*2+3*\gaisen*0.625,\tate+\tate/2+\tate/2+\nobishiro/2) circle (1mm); 
  \fill (-\gaisen-\gaisen-\yoko-\gaisen-\yoko-\gaisen*2+3*\gaisen*0.375,-\nobishiro/2) circle (1mm); 
  \fill (-\gaisen-\gaisen-\yoko-\gaisen-\yoko-\gaisen*2+3*\gaisen*0.5,-\nobishiro/2) circle (1mm); 
  \fill (-\gaisen-\gaisen-\yoko-\gaisen-\yoko-\gaisen*2+3*\gaisen*0.625,-\nobishiro/2) circle (1mm); 
\end{tikzpicture}\, .
    \end{eqnarray}
    Then, $\tilde{\Lambda}^{R(4)}_{\alpha\beta\gamma\delta}$ is the right fixed point of the $4$-transfer matrix $T^{(4)}_{\alpha\beta\gamma\delta}$. 
    \\
    (ii) Let $\tilde{\Lambda}^{L(4)}_{\alpha\beta\gamma\delta}$ be a $4$-leg tensor defined by
    \begin{eqnarray}\label{eq:4_fp_left}
        \tilde{\Lambda}^{L(4)}_{\alpha\beta\gamma\delta}=\Lambda^{L}_{\alpha\delta}\Lambda^{L(2,1)}_{\alpha\gamma\delta}\Lambda^{L(2,1)}_{\alpha\beta\gamma}T^{(4)}_{\alpha\beta\gamma}
        =
        \begin{tikzpicture}[line cap=round,line join=round,x=1.0cm,y=1.0cm, scale=0.3, baseline={([yshift=-.6ex]current bounding box.center)}, thick, shift={(0,0)}, scale=0.7]
  \def\tate{3} 
  \def\yoko{0.5} 
  \def\sen{7} 
  \def\gaisen{2} 
  \pgfmathsetmacro\hankei{\tate*1.7} 
  \def\hankeinosa{2} 
  \def\nobishiro{1}
  \draw (0,0) rectangle ++(\yoko,\tate+\tate/2);
  \draw (\yoko,\tate+\tate/2) -- (\yoko+\gaisen,\tate+\tate/2);
  \draw (\yoko,0) -- (\yoko+\gaisen+\gaisen+\yoko+\gaisen+\yoko+\gaisen*2,0);
  \draw (\yoko+\gaisen,\tate/2+\tate/2) rectangle ++(\yoko,\tate);
  \draw (\yoko+\gaisen+\yoko,\tate/2+\tate/2+\tate) -- (\yoko+\gaisen+\gaisen+\yoko+\gaisen+\yoko+\gaisen*2,\tate/2+\tate/2+\tate);
  \draw (\yoko+\gaisen+\yoko,\tate/2+\tate/2) -- (\yoko+\gaisen+\yoko+\gaisen,\tate/2+\tate/2);
  \draw (\yoko+\gaisen+\yoko+\gaisen,\tate/2) rectangle ++(\yoko,\tate);
  \draw (\yoko+\gaisen+\yoko+\gaisen+\yoko,\tate+\tate/2) -- (\yoko+\gaisen+\gaisen+\yoko+\gaisen+\yoko+\gaisen*2,\tate+\tate/2);
  \draw (\yoko+\gaisen+\yoko+\gaisen+\yoko,\tate/2) -- (\yoko+\gaisen+\gaisen+\yoko+\gaisen+\yoko+\gaisen*2,\tate/2);
  \draw (\yoko+\gaisen+\yoko+\gaisen+\yoko+3*\gaisen*0.25,\tate+\tate/2+\tate/2+\nobishiro) -- (\yoko+\gaisen+\yoko+\gaisen+\yoko+3*\gaisen*0.25,0-\nobishiro);
  \fill (\yoko+\gaisen+\yoko+\gaisen+\yoko+3*\gaisen*0.25,\tate+\tate/2+\tate/2+\nobishiro) circle (2mm); 
  \fill (\yoko+\gaisen+\yoko+\gaisen+\yoko+3*\gaisen*0.25,0-\nobishiro) circle (2mm);
  \draw (\yoko+\gaisen+\yoko+\gaisen+\yoko+3*\gaisen*0.75,\tate+\tate/2+\tate/2+\nobishiro) -- (\yoko+\gaisen+\yoko+\gaisen+\yoko+3*\gaisen*0.75,0-\nobishiro);
  \fill (\yoko+\gaisen+\yoko+\gaisen+\yoko+3*\gaisen*0.75,\tate+\tate/2+\tate/2+\nobishiro) circle (2mm); 
  \fill (\yoko+\gaisen+\yoko+\gaisen+\yoko+3*\gaisen*0.75,0-\nobishiro) circle (2mm);
  \fill (\yoko+\gaisen+\yoko+\gaisen+\yoko+3*\gaisen*0.375,\tate+\tate/2+\tate/2+\nobishiro/2) circle (1mm); 
  \fill (\yoko+\gaisen+\yoko+\gaisen+\yoko+3*\gaisen*0.5,\tate+\tate/2+\tate/2+\nobishiro/2) circle (1mm); 
  \fill (\yoko+\gaisen+\yoko+\gaisen+\yoko+3*\gaisen*0.625,\tate+\tate/2+\tate/2+\nobishiro/2) circle (1mm); 
  \fill (\yoko+\gaisen+\yoko+\gaisen+\yoko+3*\gaisen*0.375,-\nobishiro/2) circle (1mm); 
  \fill (\yoko+\gaisen+\yoko+\gaisen+\yoko+3*\gaisen*0.5,-\nobishiro/2) circle (1mm); 
  \fill (\yoko+\gaisen+\yoko+\gaisen+\yoko+3*\gaisen*0.625,-\nobishiro/2) circle (1mm); 
\end{tikzpicture}\, .
    \end{eqnarray}
     Then, $\tilde{\Lambda}^{L(4)}_{\alpha\beta\gamma\delta}$ is the left fixed point of the $4$-transfer matrix $T^{(4)}_{\alpha\beta\gamma\delta}$. 
\end{lemma}
The proof is the same as that of Lem.\ \ref{thm:4_fp_as_reduction}. 

Since the fixed point of the $4$-transfer matrix $T^{(4)}_{\alpha\beta\gamma\delta}$ is unique, $\Lambda^{R(4)}_{\alpha\beta\gamma\delta}$ and $\tilde{\Lambda}^{R(4)}_{\alpha\beta\gamma\delta}$ are proportional to each other. 
We denote the proportionality constant as $c_{\alpha\beta\gamma\delta}$:
\begin{eqnarray}\label{eq:prop_constant}
    \Lambda^{R(4)}_{\alpha\beta\gamma\delta}=c_{\alpha\beta\gamma\delta}\tilde{\Lambda}^{R(4)}_{\alpha\beta\gamma\delta}
    \quad
    \Longleftrightarrow
    \quad
\begin{tikzpicture}[line cap=round,line join=round,x=1.0cm,y=1.0cm, scale=0.3, baseline={([yshift=-.6ex]current bounding box.center)}, thick, shift={(0,0)}, scale=0.7]
  \def\tate{3} 
  \def\yoko{0.5} 
  \def\sen{7} 
  \def\gaisen{2} 
  \pgfmathsetmacro\hankei{\tate*1.7} 
  \def\hankeinosa{2} 
  \def\nobishiro{1}
  \draw (0,0) rectangle ++(\yoko,\tate);
  \draw (-\gaisen,\tate) -- (0,\tate);
  \draw (-\gaisen-\gaisen-\yoko-\gaisen-\yoko-\gaisen*2,0) -- (0,0);
  \draw (-\gaisen-\yoko,\tate/2) rectangle ++(\yoko,\tate);
  \draw (-\gaisen-\gaisen-\yoko,\tate+\tate/2) -- (-\gaisen-\yoko,\tate+\tate/2);
  \draw (-\gaisen-\gaisen-\yoko-\gaisen-\yoko-\gaisen*2,\tate/2) -- (-\gaisen-\yoko,\tate/2);
  \draw (-\gaisen-\yoko-\gaisen-\yoko,\tate/2+\tate/2) rectangle ++(\yoko,\tate);
  \draw (-\gaisen-\gaisen-\yoko-\gaisen-\yoko-\gaisen*2,\tate+\tate/2+\tate/2) -- (-\gaisen-\yoko-\gaisen-\yoko,\tate+\tate/2+\tate/2);
  \draw (-\gaisen-\gaisen-\yoko-\gaisen-\yoko-\gaisen*2,\tate/2+\tate/2) -- (-\gaisen-\yoko-\gaisen-\yoko,\tate/2+\tate/2);
  \draw (-\gaisen-\gaisen-\yoko-\gaisen-\yoko-\gaisen*2+3*\gaisen*0.25,\tate+\tate/2+\tate/2+\nobishiro) -- (-\gaisen-\gaisen-\yoko-\gaisen-\yoko-\gaisen*2+3*\gaisen*0.25,0-\nobishiro);
  \fill (-\gaisen-\gaisen-\yoko-\gaisen-\yoko-\gaisen*2+3*\gaisen*0.25,\tate+\tate/2+\tate/2+\nobishiro) circle (2mm); 
  \fill (-\gaisen-\gaisen-\yoko-\gaisen-\yoko-\gaisen*2+3*\gaisen*0.25,0-\nobishiro) circle (2mm);
  \draw (-\gaisen-\gaisen-\yoko-\gaisen-\yoko-\gaisen*2+3*\gaisen*0.75,\tate+\tate/2+\tate/2+\nobishiro) -- (-\gaisen-\gaisen-\yoko-\gaisen-\yoko-\gaisen*2+3*\gaisen*0.75,0-\nobishiro);
  \fill (-\gaisen-\gaisen-\yoko-\gaisen-\yoko-\gaisen*2+3*\gaisen*0.75,\tate+\tate/2+\tate/2+\nobishiro) circle (2mm); 
  \fill (-\gaisen-\gaisen-\yoko-\gaisen-\yoko-\gaisen*2+3*\gaisen*0.75,0-\nobishiro) circle (2mm);
  \fill (-\gaisen-\gaisen-\yoko-\gaisen-\yoko-\gaisen*2+3*\gaisen*0.375,\tate+\tate/2+\tate/2+\nobishiro/2) circle (1mm); 
  \fill (-\gaisen-\gaisen-\yoko-\gaisen-\yoko-\gaisen*2+3*\gaisen*0.5,\tate+\tate/2+\tate/2+\nobishiro/2) circle (1mm); 
  \fill (-\gaisen-\gaisen-\yoko-\gaisen-\yoko-\gaisen*2+3*\gaisen*0.625,\tate+\tate/2+\tate/2+\nobishiro/2) circle (1mm); 
  \fill (-\gaisen-\gaisen-\yoko-\gaisen-\yoko-\gaisen*2+3*\gaisen*0.375,-\nobishiro/2) circle (1mm); 
  \fill (-\gaisen-\gaisen-\yoko-\gaisen-\yoko-\gaisen*2+3*\gaisen*0.5,-\nobishiro/2) circle (1mm); 
  \fill (-\gaisen-\gaisen-\yoko-\gaisen-\yoko-\gaisen*2+3*\gaisen*0.625,-\nobishiro/2) circle (1mm); 
\end{tikzpicture}
    =
    c_{\alpha\beta\gamma\delta}\times
    \begin{tikzpicture}[line cap=round,line join=round,x=1.0cm,y=1.0cm, scale=0.3, baseline={([yshift=-.6ex]current bounding box.center)}, thick, shift={(0,0)}, scale=0.7]
  \def\tate{3} 
  \def\yoko{0.5} 
  \def\sen{7} 
  \def\gaisen{2} 
  \pgfmathsetmacro\hankei{\tate*1.7} 
  \def\hankeinosa{2} 
  \def\nobishiro{1}
  \draw (0,0) rectangle ++(\yoko,\tate+\tate/2);
  \draw (-\gaisen,\tate+\tate/2) -- (0,\tate+\tate/2);
  \draw (-\gaisen-\gaisen-\yoko-\gaisen-\yoko-\gaisen*2,0) -- (0,0);
  \draw (-\gaisen-\yoko,\tate/2+\tate/2) rectangle ++(\yoko,\tate);
  \draw (-\gaisen-\gaisen-\yoko,\tate/2+\tate/2) -- (-\gaisen-\yoko,\tate/2+\tate/2);
  \draw (-\gaisen-\gaisen-\yoko-\gaisen-\yoko-\gaisen*2,\tate+\tate/2+\tate/2) -- (-\gaisen-\yoko,\tate+\tate/2+\tate/2);
  \draw (-\gaisen-\yoko-\gaisen-\yoko,\tate/2) rectangle ++(\yoko,\tate);
  \draw (-\gaisen-\gaisen-\yoko-\gaisen-\yoko-\gaisen*2,\tate+\tate/2) -- (-\gaisen-\yoko-\gaisen-\yoko,\tate+\tate/2);
  \draw (-\gaisen-\gaisen-\yoko-\gaisen-\yoko-\gaisen*2,\tate/2) -- (-\gaisen-\yoko-\gaisen-\yoko,\tate/2);
  \draw (-\gaisen-\gaisen-\yoko-\gaisen-\yoko-\gaisen*2+3*\gaisen*0.25,\tate+\tate/2+\tate/2+\nobishiro) -- (-\gaisen-\gaisen-\yoko-\gaisen-\yoko-\gaisen*2+3*\gaisen*0.25,0-\nobishiro);
  \fill (-\gaisen-\gaisen-\yoko-\gaisen-\yoko-\gaisen*2+3*\gaisen*0.25,\tate+\tate/2+\tate/2+\nobishiro) circle (2mm); 
  \fill (-\gaisen-\gaisen-\yoko-\gaisen-\yoko-\gaisen*2+3*\gaisen*0.25,0-\nobishiro) circle (2mm);
  \draw (-\gaisen-\gaisen-\yoko-\gaisen-\yoko-\gaisen*2+3*\gaisen*0.75,\tate+\tate/2+\tate/2+\nobishiro) -- (-\gaisen-\gaisen-\yoko-\gaisen-\yoko-\gaisen*2+3*\gaisen*0.75,0-\nobishiro);
  \fill (-\gaisen-\gaisen-\yoko-\gaisen-\yoko-\gaisen*2+3*\gaisen*0.75,\tate+\tate/2+\tate/2+\nobishiro) circle (2mm); 
  \fill (-\gaisen-\gaisen-\yoko-\gaisen-\yoko-\gaisen*2+3*\gaisen*0.75,0-\nobishiro) circle (2mm);
  \fill (-\gaisen-\gaisen-\yoko-\gaisen-\yoko-\gaisen*2+3*\gaisen*0.375,\tate+\tate/2+\tate/2+\nobishiro/2) circle (1mm); 
  \fill (-\gaisen-\gaisen-\yoko-\gaisen-\yoko-\gaisen*2+3*\gaisen*0.5,\tate+\tate/2+\tate/2+\nobishiro/2) circle (1mm); 
  \fill (-\gaisen-\gaisen-\yoko-\gaisen-\yoko-\gaisen*2+3*\gaisen*0.625,\tate+\tate/2+\tate/2+\nobishiro/2) circle (1mm); 
  \fill (-\gaisen-\gaisen-\yoko-\gaisen-\yoko-\gaisen*2+3*\gaisen*0.375,-\nobishiro/2) circle (1mm); 
  \fill (-\gaisen-\gaisen-\yoko-\gaisen-\yoko-\gaisen*2+3*\gaisen*0.5,-\nobishiro/2) circle (1mm); 
  \fill (-\gaisen-\gaisen-\yoko-\gaisen-\yoko-\gaisen*2+3*\gaisen*0.625,-\nobishiro/2) circle (1mm); 
\end{tikzpicture}
\, .
\end{eqnarray}
By multiplying the left fixed point $\tilde{\Lambda}^{L(4)}_{\alpha\beta\gamma\delta}$ from the left,
\begin{eqnarray}
    \begin{tikzpicture}[line cap=round,line join=round,x=1.0cm,y=1.0cm, scale=0.3, baseline={([yshift=-.6ex]current bounding box.center)}, thick, shift={(0,0)}, scale=0.7]
  \def\tate{3} 
  \def\yoko{0.5} 
  \def\sen{7} 
  \def\gaisen{2} 
  \pgfmathsetmacro\hankei{\tate*1.7} 
  \def\hankeinosa{2} 
  \def\nobishiro{1}
  \draw (0,0) rectangle ++(\yoko,\tate+\tate/2);
  \draw (\yoko,\tate+\tate/2) -- (\yoko+\gaisen,\tate+\tate/2);
  \draw (\yoko,0) -- (\yoko+\gaisen+\gaisen+\yoko+\gaisen+\yoko+\gaisen*2,0);
  \draw (\yoko+\gaisen,\tate/2+\tate/2) rectangle ++(\yoko,\tate);
  \draw (\yoko+\gaisen+\yoko,\tate/2+\tate/2+\tate) -- (\yoko+\gaisen+\gaisen+\yoko+\gaisen+\yoko+\gaisen*2,\tate/2+\tate/2+\tate);
  \draw (\yoko+\gaisen+\yoko,\tate/2+\tate/2) -- (\yoko+\gaisen+\yoko+\gaisen,\tate/2+\tate/2);
  \draw (\yoko+\gaisen+\yoko+\gaisen,\tate/2) rectangle ++(\yoko,\tate);
  \draw (\yoko+\gaisen+\yoko+\gaisen+\yoko,\tate+\tate/2) -- (\yoko+\gaisen+\yoko+\gaisen+\yoko+\yoko+\yoko,\tate+\tate/2);
  \draw (\yoko+\gaisen+\yoko+\gaisen+\yoko,\tate/2) -- (\yoko+\gaisen+\gaisen+\yoko+\gaisen+\yoko+\gaisen*2,\tate/2);
  \draw (\yoko+\gaisen+\yoko+\gaisen+\yoko+3*\gaisen*0.25,\tate+\tate/2+\tate/2+\nobishiro) -- (\yoko+\gaisen+\yoko+\gaisen+\yoko+3*\gaisen*0.25,0-\nobishiro);
  \fill (\yoko+\gaisen+\yoko+\gaisen+\yoko+3*\gaisen*0.25,\tate+\tate/2+\tate/2+\nobishiro) circle (2mm); 
  \fill (\yoko+\gaisen+\yoko+\gaisen+\yoko+3*\gaisen*0.25,0-\nobishiro) circle (2mm);
  \draw (\yoko+\gaisen+\yoko+\gaisen+\yoko+3*\gaisen*0.75,\tate+\tate/2+\tate/2+\nobishiro) -- (\yoko+\gaisen+\yoko+\gaisen+\yoko+3*\gaisen*0.75,0-\nobishiro);
  \fill (\yoko+\gaisen+\yoko+\gaisen+\yoko+3*\gaisen*0.75,\tate+\tate/2+\tate/2+\nobishiro) circle (2mm); 
  \fill (\yoko+\gaisen+\yoko+\gaisen+\yoko+3*\gaisen*0.75,0-\nobishiro) circle (2mm);
  \fill (\yoko+\gaisen+\yoko+\gaisen+\yoko+3*\gaisen*0.375,\tate+\tate/2+\tate/2+\nobishiro/2) circle (1mm); 
  \fill (\yoko+\gaisen+\yoko+\gaisen+\yoko+3*\gaisen*0.5,\tate+\tate/2+\tate/2+\nobishiro/2) circle (1mm); 
  \fill (\yoko+\gaisen+\yoko+\gaisen+\yoko+3*\gaisen*0.625,\tate+\tate/2+\tate/2+\nobishiro/2) circle (1mm); 
  \fill (\yoko+\gaisen+\yoko+\gaisen+\yoko+3*\gaisen*0.375,-\nobishiro/2) circle (1mm); 
  \fill (\yoko+\gaisen+\yoko+\gaisen+\yoko+3*\gaisen*0.5,-\nobishiro/2) circle (1mm); 
  \fill (\yoko+\gaisen+\yoko+\gaisen+\yoko+3*\gaisen*0.625,-\nobishiro/2) circle (1mm); 
  \draw (\yoko+\gaisen+\gaisen+\yoko+\gaisen+\yoko+\gaisen*2+\gaisen+\yoko+\gaisen+\yoko,0) rectangle ++(\yoko,\tate);
  \draw (\yoko+\gaisen+\gaisen+\yoko+\gaisen+\yoko+\gaisen*2+\gaisen+\yoko+\yoko,\tate) -- (\yoko+\gaisen+\gaisen+\yoko+\gaisen+\yoko+\gaisen*2+\gaisen+\yoko+\gaisen+\yoko,\tate);
  \draw (\yoko+\gaisen+\gaisen+\yoko+\gaisen+\yoko+\gaisen*2,0) -- (\yoko+\gaisen+\gaisen+\yoko+\gaisen+\yoko+\gaisen*2+\gaisen+\yoko+\gaisen+\yoko,0);
  \draw (\yoko+\gaisen+\gaisen+\yoko+\gaisen+\yoko+\gaisen*2,\tate/2+\tate/2) rectangle ++(\yoko,\tate);
  \draw (\yoko+\gaisen+\gaisen+\yoko+\gaisen+\yoko+\gaisen*2-\yoko-\yoko,\tate/2+\tate/2) -- (\yoko+\gaisen+\gaisen+\yoko+\gaisen+\yoko+\gaisen*2,\tate/2+\tate/2);
  \draw (\yoko+\gaisen+\gaisen+\yoko+\gaisen+\yoko+\gaisen*2+\yoko+\gaisen,\tate/2) rectangle ++(\yoko,\tate);
  \draw (\yoko+\gaisen+\gaisen+\yoko+\gaisen+\yoko+\gaisen*2+\yoko,\tate+\tate/2) -- (\yoko+\gaisen+\gaisen+\yoko+\gaisen+\yoko+\gaisen*2+\yoko+\gaisen,\tate+\tate/2);
  \draw (\yoko+\gaisen+\gaisen+\yoko+\gaisen+\yoko+\gaisen*2,\tate/2) -- (\yoko+\gaisen+\gaisen+\yoko+\gaisen+\yoko+\gaisen*2+\yoko+\gaisen,\tate/2);
  \draw (\yoko+\gaisen+\yoko+\gaisen+\yoko+\yoko+\yoko,\tate+\tate/2) to [out=0,in=180] (\yoko+\gaisen+\gaisen+\yoko+\gaisen+\yoko+\gaisen*2-\yoko-\yoko,\tate/2+\tate/2);
\end{tikzpicture}
&=&c_{\alpha\beta\gamma\delta}\times
\begin{tikzpicture}[line cap=round,line join=round,x=1.0cm,y=1.0cm, scale=0.3, baseline={([yshift=-.6ex]current bounding box.center)}, thick, shift={(0,0)}, scale=0.7]
  \def\tate{3} 
  \def\yoko{0.5} 
  \def\sen{7} 
  \def\gaisen{2} 
  \pgfmathsetmacro\hankei{\tate*1.7} 
  \def\hankeinosa{2} 
  \def\nobishiro{1}
  \draw (\yoko+\gaisen+\gaisen+\yoko+\gaisen+\yoko+\gaisen*2+\yoko+\gaisen+\yoko+\gaisen,0) rectangle ++(\yoko,\tate+\tate/2);
  \draw (\yoko+\gaisen+\gaisen+\yoko+\gaisen+\yoko+\gaisen*2+\yoko+\gaisen+\yoko,\tate+\tate/2) -- (\yoko+\gaisen+\gaisen+\yoko+\gaisen+\yoko+\gaisen*2+\yoko+\gaisen+\yoko+\gaisen,\tate+\tate/2);
  \draw (\yoko+\gaisen+\gaisen+\yoko+\gaisen+\yoko+\gaisen*2,0) -- (\yoko+\gaisen+\gaisen+\yoko+\gaisen+\yoko+\gaisen*2+\yoko+\gaisen+\yoko+\gaisen,0);
  \draw (\yoko+\gaisen+\gaisen+\yoko+\gaisen+\yoko+\gaisen*2+\yoko+\gaisen,\tate/2+\tate/2) rectangle ++(\yoko,\tate);
  \draw (\yoko+\gaisen+\gaisen+\yoko+\gaisen+\yoko+\gaisen*2,\tate/2+\tate/2+\tate) -- (\yoko+\gaisen+\gaisen+\yoko+\gaisen+\yoko+\gaisen*2+\yoko+\gaisen,\tate/2+\tate/2+\tate);
  \draw (\yoko+\gaisen+\gaisen+\yoko+\gaisen+\yoko+\gaisen*2+\yoko,\tate/2+\tate/2) -- (\yoko+\gaisen+\gaisen+\yoko+\gaisen+\yoko+\gaisen*2+\yoko+\gaisen,\tate/2+\tate/2);
  \draw (\yoko+\gaisen+\gaisen+\yoko+\gaisen+\yoko+\gaisen*2,\tate/2) rectangle ++(\yoko,\tate);
  \draw (\yoko+\gaisen+\yoko+\gaisen+\yoko+3*\gaisen*0.25,\tate+\tate/2+\tate/2+\nobishiro) -- (\yoko+\gaisen+\yoko+\gaisen+\yoko+3*\gaisen*0.25,0-\nobishiro);
  \fill (\yoko+\gaisen+\yoko+\gaisen+\yoko+3*\gaisen*0.25,\tate+\tate/2+\tate/2+\nobishiro) circle (2mm); 
  \fill (\yoko+\gaisen+\yoko+\gaisen+\yoko+3*\gaisen*0.25,0-\nobishiro) circle (2mm);
  \draw (\yoko+\gaisen+\yoko+\gaisen+\yoko+3*\gaisen*0.75,\tate+\tate/2+\tate/2+\nobishiro) -- (\yoko+\gaisen+\yoko+\gaisen+\yoko+3*\gaisen*0.75,0-\nobishiro);
  \fill (\yoko+\gaisen+\yoko+\gaisen+\yoko+3*\gaisen*0.75,\tate+\tate/2+\tate/2+\nobishiro) circle (2mm); 
  \fill (\yoko+\gaisen+\yoko+\gaisen+\yoko+3*\gaisen*0.75,0-\nobishiro) circle (2mm);
  \fill (\yoko+\gaisen+\yoko+\gaisen+\yoko+3*\gaisen*0.375,\tate+\tate/2+\tate/2+\nobishiro/2) circle (1mm); 
  \fill (\yoko+\gaisen+\yoko+\gaisen+\yoko+3*\gaisen*0.5,\tate+\tate/2+\tate/2+\nobishiro/2) circle (1mm); 
  \fill (\yoko+\gaisen+\yoko+\gaisen+\yoko+3*\gaisen*0.625,\tate+\tate/2+\tate/2+\nobishiro/2) circle (1mm); 
  \fill (\yoko+\gaisen+\yoko+\gaisen+\yoko+3*\gaisen*0.375,-\nobishiro/2) circle (1mm); 
  \fill (\yoko+\gaisen+\yoko+\gaisen+\yoko+3*\gaisen*0.5,-\nobishiro/2) circle (1mm); 
  \fill (\yoko+\gaisen+\yoko+\gaisen+\yoko+3*\gaisen*0.625,-\nobishiro/2) circle (1mm); 
  \draw (0,0) rectangle ++(\yoko,\tate+\tate/2);
  \draw (\yoko,\tate+\tate/2) -- (\yoko+\gaisen,\tate+\tate/2);
  \draw (\yoko,0) -- (\yoko+\gaisen+\gaisen+\yoko+\gaisen+\yoko+\gaisen*2,0);
  \draw (\yoko+\gaisen,\tate/2+\tate/2) rectangle ++(\yoko,\tate);
  \draw (\yoko+\gaisen+\yoko,\tate/2+\tate/2+\tate) -- (\yoko+\gaisen+\gaisen+\yoko+\gaisen+\yoko+\gaisen*2,\tate/2+\tate/2+\tate);
  \draw (\yoko+\gaisen+\yoko,\tate/2+\tate/2) -- (\yoko+\gaisen+\yoko+\gaisen,\tate/2+\tate/2);
  \draw (\yoko+\gaisen+\yoko+\gaisen,\tate/2) rectangle ++(\yoko,\tate);
  \draw (\yoko+\gaisen+\yoko+\gaisen+\yoko,\tate+\tate/2) -- (\yoko+\gaisen+\gaisen+\yoko+\gaisen+\yoko+\gaisen*2,\tate+\tate/2);
  \draw (\yoko+\gaisen+\yoko+\gaisen+\yoko,\tate/2) -- (\yoko+\gaisen+\gaisen+\yoko+\gaisen+\yoko+\gaisen*2,\tate/2);
\end{tikzpicture}\, .
\end{eqnarray}
The diagram included on the right-hand side becomes $1$ due to the normalization condition. 
We thus obtain
Eq.\ \eqref{eq:quad_inner_product}.

\subsubsection{Proof of Props.\ \ref{prop3}
and \ref{thm:quad_inner}}

We are now ready to proceed to the proof of Props.\ \ref{prop3}
and \ref{thm:quad_inner}. First, we show Prop.\ \ref{prop3}. As we pointed out at the end of Sec.\ \ref{sec:fp_2_3_transfer}, $(\Lambda_{\alpha\beta}^{L(2)},\Lambda_{\alpha\beta}^{R(2)})$ and $(\Lambda_{\alpha\beta\gamma}^{L(3)},\Lambda_{\alpha\beta\gamma}^{R(3)})$ are reduction tensors. 
Thus, their inner product,  defined by
\begin{eqnarray}
    (\Lambda_{\alpha\beta}^{L(2)},\Lambda_{\alpha\beta}^{R(2)})=

\end{eqnarray}
are nonzero complex numbers.

Next, we show Prop.\ \ref{thm:quad_inner}. Except for statement (iii), the others are obvious. We can easily check that 
\begin{eqnarray}
    %
\end{eqnarray}
are right reduction tensors. By using Thm.\ \ref{thm:uniquenessofreduction}, multiplying a sufficient number of MPO tensors, they are proportional to each other. By multiplying by $\Lambda_{\alpha\delta}^{R}$ and appropriately contracting the legs, 
we obtain Eq.\ \eqref{eq:prop_constant}. 
Therefore this proportionality constant coincides with $c_{\alpha\beta\gamma\delta}$:
\begin{eqnarray}
    %
\, .
\end{eqnarray}
Let's consider the diagram
\begin{eqnarray}
    %
\end{eqnarray}
and follow a standard discussion to derive the pentagon equation of the F-symbol:
\begin{eqnarray}
    %
\, .
\end{eqnarray}
On the other hand, 
\begin{eqnarray}
    %
\, .
\end{eqnarray}
By comparing the two equations, we obtain the cocycle condition 
Eq.\ \eqref{eq:cocycle}.

\section{The CZX Model Protected by $\mathbb{Z}/2\mathbb{Z}$ Symmetry}
\label{Relation to the Group Cohomological Classification}


Having gone through the generalities of the quadruple inner product,  
we will now discuss some applications and examples
outlines at the end of Sec.\ \ref{preliminaries}.
Here, we will use (2+1)d SPT phases
to illustrate the use of the quadruple inner product.

%




As a particular example of SPT phases, let us consider 
the CZX model \cite{CLW11}
-- this is an example of (2+1)d SPT protected by 
$\mathbb{Z}/2\mathbb{Z}$ on-site unitary symmetry.
Let us consider the 
square lattice 
with 
$2^4$-dimensional 
local Hilbert space
(i.e., four copies of spin 1/2)
for each site. Let $p$ denote the plaquette label, and $i$ represent the site label. The local Hamiltonian labeled by $p$ is defined as an operator that acts on the plaquette $p$ and its four neighboring edges. To define this, we label the top, bottom, left, and right edges adjacent to plaquette $p$ as $u(p)$, $d(p)$, $l(p)$, and $r(p)$ respectively. The state of the $2$ sites belonging to an edge is written as $\{ \ket{i,j}_{k(p)} \left| \right. i,j=\uparrow,\downarrow, k=u,d,l,r \}$, where the order inside the ket is arranged in a clockwise manner. For instance, the state where the upper site is $\uparrow$ and the lower site is $\downarrow$ on the edge adjacent to the right of plaquette $p$ is expressed as $\ket{\uparrow\downarrow}_{r(p)}$. The sites belonging to site $i$ will be labeled in a counterclockwise manner starting from the top-left site as $i(1)$, $i(2)$, $i(3)$, and $i(4)$.

The local Hamiltonian 
for a plaquette $p$ is defined by
\begin{align}
    h_{p}&:=-X_{p}^{(4)}\otimes \bigotimes_{k=u,d,l,r}P_{k(p)}^{(2)},
     \\
\mbox{where}\qquad 
&
X_{p}^{(4)}:=
\ket{\uparrow\uparrow\uparrow\uparrow}\bra{\downarrow\downarrow\downarrow\downarrow}_{p}+\ket{\downarrow\downarrow\downarrow\downarrow}\bra{\uparrow\uparrow\uparrow\uparrow}_{p},
\nonumber \\ 
&
P_{k(p)}^{(2)}:=\ket{\uparrow\uparrow}\bra{\uparrow\uparrow}_{k(p)}+\ket{\downarrow\downarrow}\bra{\downarrow\downarrow}_{k(p)} \quad (k=u,d,l,r).
\end{align}
The total Hamiltonian is defined by 
the sum of the local Hamiltonians, 
$H_{\rm CZX}:=\sum_{p}h_{p}$.
This Hamiltonian commutes with the CZX symmetry
$U_{\rm CZX}:=\bigotimes_{i} u_{\mathrm{CZX},i}$,
where
\begin{eqnarray}
    u_{\mathrm{CZX},i}:=\left(\bigotimes_{k=1}^{4}X_{i(k)}\right)\otimes CZ_{i(1)i(2)} \otimes CZ_{i(2)i(3)} \otimes CZ_{i(3)i(4)} \otimes CZ_{i(4)i(1)}.
\end{eqnarray}
Here, $CZ_{ij}:=\sum_{ij}(-1)^{ab}\ket{a_i b_j}\bra{a_ib_j}$ is the CZ-gate operator acting on sites $i,j$. $U_{\rm CZX}$ generates onsite $\zmod{2}$ symmetry.
It is known that the ground state on the closed system is unique and given by 
\begin{eqnarray}
    \ket{G.S.}:=\bigotimes_{p}\frac{1}{\sqrt{2}}(\ket{\uparrow\uparrow\uparrow\uparrow}_{p}+\ket{\downarrow\downarrow\downarrow\downarrow}_{p}).
\end{eqnarray}
This is an example of $(2+1)$d SPT phases 
protected by $\zmod{2}$ symmetry
\cite{CLW11}.
The ground state of the CZX model 
can be expressed in terms of a semi-injective PEPS \cite{MGSC18}. 

As mentioned at the end of Sec.\ \ref{preliminaries},
what is important for our purposes is
the boundary symmetry operators. 
When making a boundary,
there appear free degrees of freedom
on the plaquettes that are cut open by the boundary. 
The boundary symmetry operators
are obtained by 
projecting the (bulk) symmetry operators to the space of the free degrees of freedom on the boundary. 
Explicitly, they are given by 
\begin{align}
  \mathcal{O}_{I}= \mathbf{1},
  \quad 
  \mathcal{O}_Z = \otimes_i {{\it CZ}}_{i, i+1} 
  \otimes_i X_i,
\end{align} 
where $i$ labels boundary sites.
The boundary symmetry operators can be represented as an MPU 
with 
\begin{eqnarray}
  B_{I}^{01} =B_{I}^{10} =0,\quad 
  B_{I}^{00}= B_{I}^{11}=1,
\end{eqnarray}
and
\begin{eqnarray}
  B_{Z}^{01} = \left(\begin{array}{cc}
    1 & 1 \\
    0 & 0
  \end{array}\right),
  \quad 
  B_{Z}^{10} = \left(\begin{array}{cc}
    0 & 0 \\
    1 & -1 
  \end{array}\right),
  \quad 
  B_{Z}^{00}= B_{Z}^{11}=0.
\end{eqnarray}
Remark that these tensors are in the right canonical form. 

Now, let's consider the three-leg tensors $\Lambda_{ZZI}^{R(2,1)}$, $\Lambda_{ZIZ}^{R(2,1)}$ and 
$\Lambda_{IZZ}^{R(2,1)}$. 
As $\Lambda_{ZIZ}^{R(2,1)}$ and $\Lambda_{IZZ}^{R(2,1)}$ intertwine $\mathcal{O}_{Z}$ and $\mathcal{O}_{Z}$, they are the identity matrix. 
$\Lambda_{ZZI}^{R(2,1)}$ is the only nontrivial one and is represented as $4\times1$ tensor. 
We need to compute the composed tensor $B_{ZZ}^{ij}:=\sum_{k}B_{Z}^{ij}\otimes B_{Z}^{kj}$:
\begin{eqnarray}
    B_{ZZ}^{00}&=&\sum_{k}B_{Z}^{0k}\otimes B_{Z}^{k0}=\left(\begin{array}{cccc}
        0&0&0&0\\
        1&-1&1&-1\\
        0&0&0&0\\
        0&0&0&0
    \end{array}\right),\\
    B_{ZZ}^{11}&=&\sum_{k}B_{Z}^{1k}\otimes B_{Z}^{k1}=\left(\begin{array}{cccc}
        0&0&0&0\\
        0&0&0&0\\
        1&1&-1&-1\\
        0&0&0&0
    \end{array}\right),\\
    B_{ZZ}^{01}&=&B_{ZZ}^{10}=0.
\end{eqnarray}
Here, given $2$-dimensional vector spaces $V$ and $W$ with orthonormal bases $v_0, v_1$ and $w_0, w_1$ respectively, the basis for the tensor product space is fixed in the order $v_0 \otimes w_0$, $v_0 \otimes w_1$, $v_1 \otimes w_0$, and $v_1 \otimes w_1$.
As illustrated in \cite{HV17} and App.\ A.4 of \cite{CP-GSV21}, 
the reduction pair ($V_{ZZI},W_{ZZI}$)
from
$B_{ZZ}^{ij}$ to $B_{I}^{ij}$ is given by
\begin{eqnarray}
    V_{ZZI}&=&\frac{1}{\sqrt{2}}\begin{pmatrix}
        0,1,-1,0
    \end{pmatrix}
    \Longleftrightarrow
    (V_{ZZI})_{0}{}^{ab}{}=\frac{(-1)^{b}}{\sqrt{2}}(1-\delta_{ab}),\\
    W_{ZZI}&=&\frac{1}{\sqrt{2}}\begin{pmatrix}
        0\\
        1\\
        -1\\
        0
    \end{pmatrix}
    \Longleftrightarrow
    (W_{ZZI})_{ab}{}^{0}=\frac{(-1)^{b}}{\sqrt{2}}(1-\delta_{ab}),
\end{eqnarray}
where $a,b=0,1$ and $\delta_{ab}$ is the Kronecker delta. 
Therefore, 
the reduced MPU matrices $\tilde{B}_{I}^{ij}
:=(V_{ZZI})^{\dagger}B_{ZZ}^{ij}
W_{ZZI}$ are given by
\begin{eqnarray}
\tilde{B}_{I}^{00}=\tilde{B}_{I}^{00}=-1,\;\tilde{B}_{I}^{01}=\tilde{B}_{I}^{10}=0.
\end{eqnarray}
Thus, the relative phase between $B_{I}^{ij}$ and $\tilde{B}_{I}^{ij}$ 
is $c_{ZZI}=-1$. 
This implies that the phase $c$ in 
Prop.\ \ref{thm:uniquenessofreduction} is $-1$. Therefore, the normalized transfer matrix
\begin{eqnarray}
    T^{R(2,1)}_{ZZI}=-\frac{1}{2}\sum_{ij}B_{ZZ}^{ij}\otimes B_{I}^{ij\dagger}
\end{eqnarray}
has a unique fixed point tensor.
One can check 
that the right reduction tensor itself satisfies the higher eigenequation. 
In fact, when calculating the left-hand side of the higher eigenequation, we find 
\begin{eqnarray}\label{eq:higher_eigen_equation}
    -\frac{1}{2}\sum_{ij}B_{ZZ}^{ij}
    W_{ZZI}B_{I}^{ij}=c_{ZZ}{}^{I}W_{ZZI}. 
\end{eqnarray}
Therefore, we can take $\Lambda^{R(2,1)}_{ZZI}=W_{ZZI}$. Similarly, the only nontrivial left fixed point is $\Lambda^{L(2,1)}_{ZZI}$. One can check that the left reduction $V_{ZZI}$ satisfies the fixed point equation of the left transfer matrix $T^{L(2,1)}_{ZZI}$. Therefore, we can take $\Lambda^{L(2,1)}_{ZZI}=V_{ZZI}$. Note that the normalization condition Eq.(\ref{eq:normalization_3leg}) holds. In this model, the nilpotency length $N_{\alpha\beta\gamma\delta}$ in Prop.\ \ref{thm:4_fp_as_reduction} is zero. Thus, we do not need to insert the vertical lines 
in Eq.\ \eqref{eq:quad_inner_product}.



By using these tensors 
$\Lambda_{ZZI}$, 
$\Lambda_{ZIZ}$ and 
$\Lambda_{IZZ}$, 
let us now compute the quadruple inner product. 
%
As a demonstration, 
the quadruple inner product  
$c_{IZIZ}$ can be computed as 
\begin{eqnarray}
    c_{IZIZ}=\sum_{ab}(\Lambda_{ZZI})_{ab}{}^{0}(\Lambda_{ZZI}^{\dagger})_{ba}{}^{0}=-\frac{1}{\sqrt{2}}\times\frac{1}{\sqrt{2}}-\frac{1}{\sqrt{2}}\times\frac{1}{\sqrt{2}}=-1.
\end{eqnarray}
After a similar calculation, we can readily check the following result:
\begin{eqnarray}
    c_{g_1g_2g_3g_4}=\begin{cases}
        -1&(g_1,g_2,g_3,g_4)=(I,Z,I,Z),(Z,I,Z,I)\\
        1&\text{otherwise}.
    \end{cases}
\end{eqnarray}
This is a nontrivial cocycle in $\cohoU{3}{\zmod{2}}$.

\section{The Model Parameterized Over $\rp{4}$}
\label{sec:Model parameterized over rp4}

In this section, we will discuss a concrete model
(a family of parameterized (2+1)d Hamiltonians 
or (2+1)d invertible states)
characterized by the non-trivial higher Berry phase.
The model is parametrized by 
$X=\rp{4}$.
We will show that the model is 
nontrivial in $\coho{4}{\rp{4}}{\zmod{2}}\simeq\zmod{2}$. 

\subsection{The $\rp{4}$-Parametrized Model}

The CZX model has the $\zmod{2}$ symmetry. 
By interpolating this symmetry, 
we can construct a family of $2$-d systems. 
To this end, we embed the local Hilbert space 
$\mathbb{C}^{2}$ to $\mathbb{C}^{4}$ 
and take the following orthonormal basis:
\begin{eqnarray}\label{eq:embeded basis}
    \ket{+}=\begin{pmatrix}
        1\\
        0\\
        0\\
        0
    \end{pmatrix},
    \quad
    \ket{-}=\begin{pmatrix}
        0\\
        1\\
        0\\
        0
    \end{pmatrix},
    \quad
    \ket{\perp_1}=\begin{pmatrix}
        0\\
        0\\
        1\\
        0
    \end{pmatrix},
    \quad
    \ket{\perp_2}=\begin{pmatrix}
        0\\
        0\\
        0\\
        1
    \end{pmatrix}.
\end{eqnarray}
Let $\{\vec{z}=(z_1,z_2,z_3)\left.\right
|z_{1,2,3}\in \mathbb{C}, \abs{\vec{z}}^2=1\}$ 
be a coordinate of $S^5$ 
and we define $\ket{+(\vec{z})}$ and $\ket{-(\vec{z})}$ by
\begin{eqnarray}
    \ket{+(\vec{z})}=\begin{pmatrix}
        1\\
        0\\
        0\\
        0
    \end{pmatrix},
    \quad
    \ket{-(\vec{z})}=\begin{pmatrix}
        0\\
        z_1\\
        z_2\\
        z_3
    \end{pmatrix}.
\end{eqnarray}
The parameterized basis is obtained 
from the basis \eqref{eq:embeded basis}
by a unitary transformation, 
\begin{align}\label{eq:unitary}
    &
    \big(\ket{+(\vec{z})},\ket{-(\vec{z})},
    \ket{\perp_1(\vec{z})},\ket{\perp_2(\vec{z})}
    \big)
    =
    V(\vec{z})
    \big(\ket{+},\ket{-},\ket{\perp_1},\ket{\perp_2}
    \big),
    \nonumber \\
    &
    \quad 
    \mbox{where}\quad
    V(\vec{z})
    =
    \begin{pmatrix}
        1&&&\\
        &z_1&a&b\\
        &z_2&c&d\\
        &z_3&e&f
    \end{pmatrix}
\end{align}
for suitable $a,b,c,d\in\mathbb{C}$. 
Since $\mathrm{SU}(3)$ is the nontrivial $S^3$ bundle over 
$S^5$, 
taking the global families of 
unitary transformations over $S^5$ continuously 
should not be possible.
However, 
since
we will only use 
the $2$-dimensional subspace of $\mathbb{C}^4$ 
spanned by $\ket{+(\vec{z})},\ket{-(\vec{z})}$ 
to define a Hamiltonian, 
our model will be parametrized by $S^5$ globally. 

Let $\sigma^{x}(\vec{z})$ and $\sigma^{z}(\vec{z})$ be 
parametrized Pauli matrices defined by 
\begin{align}
     \sigma^{x}(\vec{z})
     &=
     \ket{+(\vec{z})}\bra{+(\vec{z})}-\ket{-(\vec{z})}\bra{-(\vec{z})},\\
     \sigma^{z}(\vec{z})
     &=
     \ket{+(\vec{z})}\bra{-(\vec{z})}+\ket{-(\vec{z})}\bra{+(\vec{z})}.
\end{align}
We introduce the polar coordinate:
\begin{align}
 z_1&=\cos{(\theta_1)}+i\sin{(\theta_1)}\cos{(\theta_{2})},\\
 z_2&=\sin{(\theta_{1})}\sin{(\theta_{2})}\cos{(\theta_{3})}+i\sin{(\theta_{1})}\sin{(\theta_{2})}\sin{(\theta_{3})}\cos{(\theta_{4})},\\
 z_3&=\sin{(\theta_{1})}\sin{(\theta_{2})}\sin{(\theta_{3})}\sin{(\theta_{4})}\sin{(\theta_{5})}+i\sin{(\theta_{1})}\sin{(\theta_{2})}\sin{(\theta_{3})}\sin{(\theta_{4})}\cos{(\theta_{5})},
\end{align}
where $0\leq\theta_{1},\theta_{2},\theta_{3},\theta_{4}\leq\pi$ and $0\leq\theta_{5}<2\pi$.
Recall that the CZ operator is defined by ${\rm CZ}_{i,j}:=e^{i\frac{\pi}{4}(1-\sigma^{z}_{i})(1-\sigma^{z}_{j})}$. 
Let us introduce 
\begin{align}
     {\rm CZ}_{i,j}(\vec{z})
     =e^{i\frac{\theta_5}{4}(1-\sigma^{z}_{i}(\vec{z}))(1-\sigma^{z}_{j}(\vec{z}))}
     \quad
     \mbox{and}
     \quad
     U_{\rm CZ}(\vec{z})=\prod_{i,j}{\rm CZ}_{i,j}(\vec{z}).
\end{align}
Remark that $U_{\rm CZ}(\vec{z})$ is a symmetry of the ground state. 
As we saw in Sec.\ 
\ref{Relation to the Group Cohomological Classification},
the CZX model is defined by using only $\sigma^{x}$ and $\sigma^{z}$. 
Thus, we define $H_{\rm CZX}(\vec{z})$ by 
\begin{eqnarray}
    H_{\rm CZX}(\vec{z}):= U_{\rm CZ}(\vec{z})
    \Big(H_{\rm CZX}\bigl.\bigr|_{\sigma^{i}\mapsto\sigma^{i}(\vec{z})}\Big)
    U_{\rm CZ}(\vec{z})^{\dagger}.
\end{eqnarray}
We can check that 
$H_{\rm CZX}(-\vec{z})= 
U_{\rm CZX}(-\vec{z})H_{\rm CZX}(\vec{z})U_{\rm CZX}(-\vec{z})^{\dagger}$
on the closed lattice.
It then follows 
$\ket{\mathrm{G.S.}(-\vec{z})}=U_{\rm CZX}(-\vec{z})\ket{\mathrm{G.S.}(\vec{z})}$. 
$U_{\rm CZX}(-\vec{z})$ is not a symmetry of $H_{\rm CZX}(\vec{z})$, 
but the ground state of $H_{\rm CZX}(\vec{z})$ is invariant: By definition, $U_{\rm CZX}(-\vec{z})=U_{\rm X}(-\vec{z})U_{\rm CZ}(-\vec{z})=U_{\rm X}(\vec{z})U_{\rm CZ}(-\vec{z})$ and it is known that $\ket{\mathrm{G.S.}(\vec{z})}$ is invariant under the action of $U_{\rm X}(\vec{z})$. Thus, all we need to show is that $\ket{\mathrm{G.S.}(\vec{z})}$ is invariant under the action of $U_{\rm CZ}(-\vec{z})$. $U_{\rm CZ}(-\vec{z})$ consists of four CZ operators, and each operator checks whether the spins are aligned. 
More precisely, ${\rm CZ}_{i,j}(\vec{z})$ assigns $-1$ if the two spins are pointing down under their quantization axis, and ${\rm CZ}_{i,j}(-\vec{z})$ assigns $-1$ if the two spins are pointing up. 
In the ground state of the CZX model, 
the spins belonging to the same plaquette, align in the same direction, and neighboring plaquettes share exactly two sites. 
Thus, the actions of ${\rm CZ}_{i,j}(\vec{z})$ cancel pairwise. 
By the same mechanism,
$U_{\rm CZ}(-\vec{z})$ acts identically on the ground state $\ket{\mathrm{G.S.}(\vec{z})}$. 
Therefore, 
up to an overall phase, 
$\ket{\mathrm{G.S.}(-\vec{z})}=\ket{\mathrm{G.S.}(\vec{z})}$.

To summarize, 
$\ket{\mathrm{G.S.}(\vec{z})}$ is parametrized by $S^{5}/\zmod{2}\simeq\rp{5}$. 
Hereafter, we restrict $z_1$ to real values and regard $\ket{\mathrm{G.S.}(\vec{z})}$ as a model parametrized by $\rp{4}$.

\subsection{Transition MPU}\label{sec:transitionMPU}
We now proceed to the calculation of the higher Berry phase.
The first step is to identify the transition MPO.
The transition MPO (MPU) from $\vec{z}$ to $-\vec{z}$ is given by 
\begin{eqnarray}
{\cal O}(\vec{z})
=\mathcal{O}[
\{
\sum_{j,k}V(\vec{z})_{ij}B^{jk}V(\vec{z})^{\dagger}_{kl}
\}],
\end{eqnarray}
where $B^{ij}$ are MPS matrices defined by
\begin{align}
    B^{00}&=B^{11}=0,
    \quad
    B^{01}=\begin{pmatrix}
        1&1\\
        0&0
    \end{pmatrix},\;
    \quad
    B^{10}=\begin{pmatrix}
        0&0\\
        1&-1
    \end{pmatrix}.
\end{align}
Here, $V(\vec{z})$ and $V({\vec{z}})^{\dagger}$ act on the transition MPU as
\begin{eqnarray}
\begin{tikzpicture}[line cap=round,line join=round,x=1.0cm,y=1.0cm, scale=0.3, baseline={([yshift=-.6ex]current bounding box.center)}, thick, shift={(0,0)}, scale=0.7]
  \def\saitosu{4} 
  \pgfmathsetmacro\saitosutasuichi{\saitosu+1}
  \newcounter{n}
  \setcounter{n}{\saitosu} 
  \def\yoko{8*2} 
  \def\tate{7} 
  \draw (0,\tate/2) -- (\yoko,\tate/2);
  \draw (\yoko/\saitosutasuichi,\tate*3/4) node [left] {$V(\vec{z})$};
  \draw (\yoko/\saitosutasuichi,\tate/4) node [left] {$V(\vec{z})^{\dagger}$};
  \foreach \x in {1,...,\then}
      \draw (\x*\yoko/\saitosutasuichi,0) -- (\x*\yoko/\saitosutasuichi,\tate); 
  \foreach \x in {1,...,\then}
      \draw[fill=black] (\x*\yoko/\saitosutasuichi,\tate/2) circle (3mm);
  \foreach \x in {1,...,\then}
      \draw[fill=white] (\x*\yoko/\saitosutasuichi,\tate/4) circle (3mm);
  \foreach \x in {1,...,\then}
      \draw[fill=gray] (\x*\yoko/\saitosutasuichi,\tate*3/4) circle (3mm);
\end{tikzpicture}\,
,
\end{eqnarray}
where black dots represent the MPU tensor $\{B^{ij}\}$ and gray and white dots represent 
$V({\vec{z}})$ and 
$V({\vec{z}})^{\dagger}$,
respectively. Therefore, they are canceled out in the higher eigen equation 
Eq.\ (\ref{eq:higher_eigen_equation}).
Consequently, 
while the MPU tensor depends on $\vec{z}$, 
the 3-leg tensor does not.

\subsection{Integration of the Higher Berry Phase}

We have constructed the $2$-dimensional model parametrized by $\rp{4}$ 
and it is expected that this model is 
nontrivial in $\coho{4}{\rp{4}}{\zmod{2}}\simeq\zmod{2}$. 
In order to measure this nontriviality, 
we integrate the higher Berry phase over $\rp{3}$. 
In general, we expect that 
a family of semi-injective PEPS 
parameterized by 
$X$
gives an element of 
the $4$-th Deligne cohomology $\hat{\mathrm{H}}^{4}\left(X;\mathbb{Z}\right)$ \cite{Gawedzki88}. 
A representative element of $\hat{\mathrm{H}}^{4}\left({\rp{4}};\mathbb{Z}\right)$ is given by 
$\mathbb{C}^{\times}$-valued function $w^{(0)}_{\alpha\beta\gamma\delta}$ on $U_{\alpha\beta\gamma\delta}$, 1-form $w^{(1)}_{\alpha\beta\gamma}$ on $U_{\alpha\beta\gamma}$, 2-form $w^{(2)}_{\alpha\beta}$ on $U_{\alpha\beta}$, and 3-form $w^{(3)}_{\alpha}$ on $U_{\alpha}$ with suitable consistency conditions. We refer to $w^{(1)}$, $w^{(2)}$ and $w^{(3)}$ as higher connections. 
In our context, $w^{(0)}_{\alpha\beta\gamma\delta}$ is the higher Berry phase.
As mentioned in Sec.\ \ref{sec:transitionMPU}, 
the $3$-leg tensor does not depend on $\vec{z}$. 
This means that $c_{\alpha\beta\gamma\delta}$ admits the trivial higher connections, i.e., 
$(w^{(3)}_{\alpha}=0,w^{(2)}_{\alpha\beta}=0,w^{(1)}_{\alpha\beta\gamma}=0,w^{(0)}_{\alpha\beta\gamma\delta}=c_{\alpha\beta\gamma\delta})$ 
defines an element of 
the $4$-th Deligne cohomology $\hat{\mathrm{H}}^{4}\left({\rp{4}};\mathbb{Z}\right)$. 
In particular, the higher Berry phase is quantized to $\zmod{2}$.


Following the integration theory of the Deligne cohomology \cite{GT00,GT01,CJM04,GT09}, we evaluate the cocycle $(w^{(3)}_{\alpha}=0,w^{(2)}_{\alpha\beta}=0,w^{(1)}_{\alpha\beta\gamma}=0,c^{(0)}_{\alpha\beta\gamma\delta})$\footnote{See \cite{OTS23} as an example of the explicit computation of the integration.}. 
First, we need to take a nontrivial homology class of ${\rm H}_{3}(\rp{4};\mathbb{Z})\simeq\zmod{2}$ as an integration surface. This is generated by $\rp{3}\subset\rp{4}$.
Let's 
consider the open covering of $\rp{3}$. 
We will illustrate 
$\rp{3}$ by considering a cube and identifying antipodal points of the surface. 
For this cube, we take balls centered on the face-centered cubic lattice and the body center. 
Due to the identification of antipodal points, the number of such balls is eight. 
Furthermore, we consider balls centered at the midpoint of each edge. 
Again, due to the identification of antipodal points, the number of these balls is six. We adjust the radii of these balls appropriately to cover the entire cube.

To simplify the illustration, we flatten these intersections of two patches into planes and illustrate them as surfaces Fig.\ \ref{fig:open_covering}. 
Fortunately, as already mentioned, the transition MPU in the example 
is constant on the intersections. Therefore, we can flatten them without loss of information.

\begin{figure}[h]
  \centering
  \includegraphics[scale=0.2]{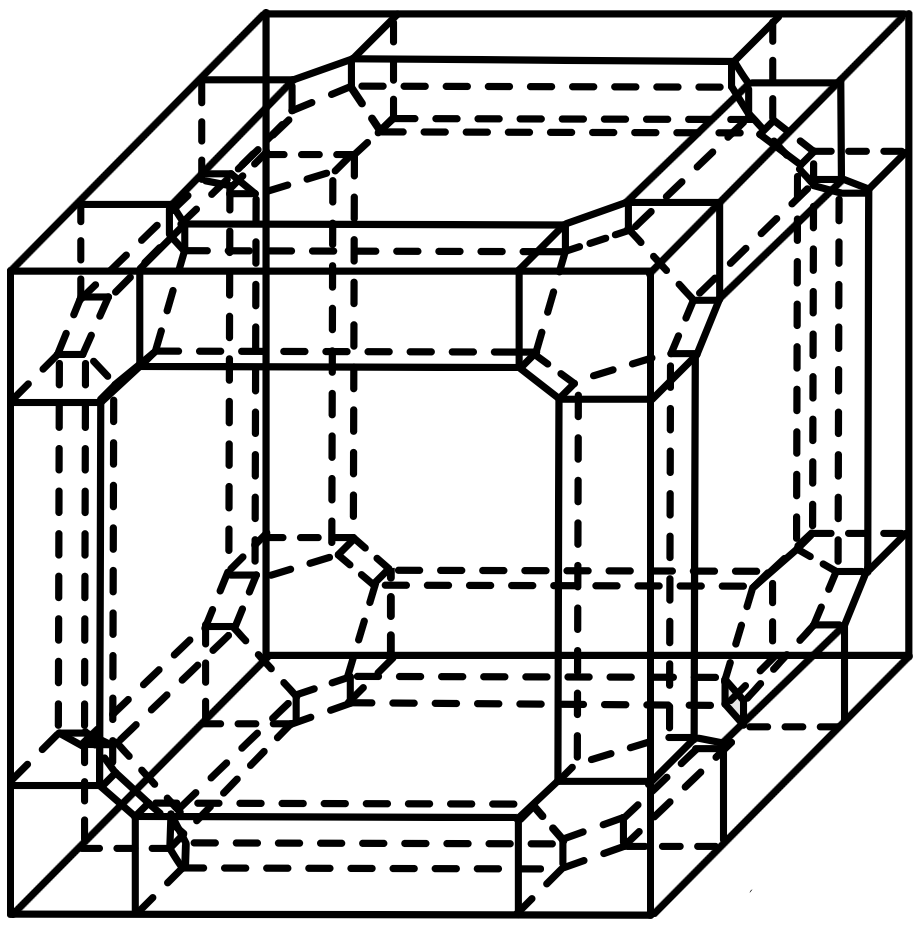}
  \hspace{0.5cm}
  \caption{An open covering of $\rp{3}$.}
  \label{fig:open_covering}
\end{figure}

All edges correspond to the parts where exactly three patches overlap. 
Therefore, a dual lattice provides a triangulation of $\rp{3}$. 
We introduce a branching structure, namely, an orientation that does not include closed loops as follows:
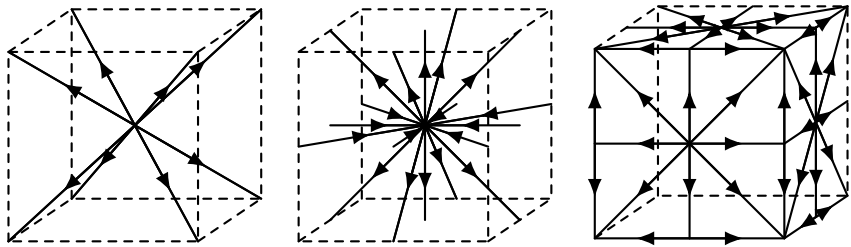
\begin{figure}[h]
\centering
\begin{tikzpicture}[line cap=round,line join=round,x=1.0cm,y=1.0cm, scale=0.3, baseline={([yshift=-.6ex]current bounding box.center)}, thick, shift={(0,0)}, scale=0.7]
  \def\ippen{12} 
  \def\okuyukix{4} 
  \def\okuyukiy{2.7} 
  \draw[dashed] (0,0) rectangle ++(\ippen,\ippen);
  \draw[dashed] (\okuyukix,\okuyukiy) rectangle ++(\ippen,\ippen);
  \draw[dashed] (0,0) -- (\okuyukix,\okuyukiy);
  \draw[dashed] (\ippen,0) -- (\ippen+\okuyukix,\okuyukiy);
  \draw[dashed] (0,\ippen) -- (\okuyukix,\ippen+\okuyukiy);
  \draw[dashed] (\ippen,\ippen) -- (\ippen+\okuyukix,\ippen+\okuyukiy);
  \draw (0,0) -- (\ippen+\okuyukix,\ippen+\okuyukiy);
  \draw (\ippen,0) -- (\okuyukix,\ippen+\okuyukiy);
  \draw (0,\ippen) -- (\ippen+\okuyukix,\okuyukiy);
  \draw (\okuyukix,\okuyukiy) -- (\ippen,\ippen);



  \draw[-{Latex[length=2.5mm]}] (0,0) -- (\ippen*11/14+\okuyukix*11/14,\ippen*11/14+\okuyukiy*11/14);
  \draw[-{Latex[length=2.5mm]}] (\ippen+\okuyukix,\ippen+\okuyukiy) -- (\ippen+\okuyukix-\ippen*11/14-\okuyukix*11/14,\ippen+\okuyukiy-\ippen*11/14-\okuyukiy*11/14);

  \draw[-{Latex[length=2.5mm]}] (\ippen,0) -- (\ippen+\okuyukix*11/14-\ippen*11/14,\ippen*11/14+\okuyukiy*11/14);
  \draw[-{Latex[length=2.5mm]}] (\okuyukix,\ippen+\okuyukiy)  --  (\okuyukix+\ippen*11/14-\okuyukix*11/14,\ippen+\okuyukiy-\ippen*11/14-\okuyukiy*11/14);

  \draw[-{Latex[length=2.5mm]}] (0,\ippen) -- (\ippen*11/14+\okuyukix*11/14,\ippen+\okuyukiy*11/14-\ippen*11/14);
  \draw[-{Latex[length=2.5mm]}] (\ippen+\okuyukix,\okuyukiy) -- (\ippen+\okuyukix-\ippen*11/14-\okuyukix*11/14,\okuyukiy+\ippen*11/14-\okuyukiy*11/14);
  
  \draw[-{Latex[length=2.5mm]}] (\okuyukix,\okuyukiy) -- (\okuyukix+\ippen*11/14-\okuyukix*11/14,\okuyukiy+\ippen*11/14-\okuyukiy*11/14);
  \draw[-{Latex[length=2.5mm]}] (\ippen,\ippen) -- (\ippen+\okuyukix*11/14-\ippen*11/14,\ippen+\okuyukiy*11/14-\ippen*11/14);

\end{tikzpicture}\;\;\;\;
\begin{tikzpicture}[line cap=round,line join=round,x=1.0cm,y=1.0cm, scale=0.3, baseline={([yshift=-.6ex]current bounding box.center)}, thick, shift={(0,0)}, scale=0.7]
  \def\ippen{12} 
  \def\okuyukix{4} 
  \def\okuyukiy{2.7} 
  \draw[dashed] (0,0) rectangle ++(\ippen,\ippen);
  \draw[dashed] (\okuyukix,\okuyukiy) rectangle ++(\ippen,\ippen);
  \draw[dashed] (0,0) -- (\okuyukix,\okuyukiy);
  \draw[dashed] (\ippen,0) -- (\ippen+\okuyukix,\okuyukiy);
  \draw[dashed] (0,\ippen) -- (\okuyukix,\ippen+\okuyukiy);
  \draw[dashed] (\ippen,\ippen) -- (\ippen+\okuyukix,\ippen+\okuyukiy);
  
  \draw (\ippen/2,\ippen/2) -- (\ippen/2+\okuyukix,\ippen/2+\okuyukiy);
  \draw (\ippen/2+\okuyukix/2,\okuyukiy/2) -- (\ippen/2+\okuyukix/2,\ippen+\okuyukiy/2);
  \draw (\okuyukix/2,\ippen/2+\okuyukiy/2) -- (\ippen+\okuyukix/2,\ippen/2+\okuyukiy/2);

  \draw (\ippen/2,0) -- (\ippen/2+\okuyukix,\ippen+\okuyukiy);
  \draw (\ippen/2,0) -- (\ippen/2+\okuyukix,\ippen+\okuyukiy);
  \draw (0,\ippen/2) -- (\ippen+\okuyukix,\ippen/2+\okuyukiy);

  \draw (\ippen/2,\ippen) -- (\ippen/2+\okuyukix,\okuyukiy);
  \draw (\ippen/2+\okuyukix,\okuyukiy) -- (\ippen/2,\ippen);
  \draw (\okuyukix,\ippen/2+\okuyukiy) -- (\ippen,\ippen/2);

  \draw (\ippen+\okuyukix/2,\okuyukiy/2) -- (\okuyukix/2,\ippen+\okuyukiy/2);
  \draw (\okuyukix/2,\ippen+\okuyukiy/2) -- (\ippen+\okuyukix/2,\okuyukiy/2);

  \draw (\okuyukix/2,\okuyukiy/2) -- (\ippen+\okuyukix/2,\ippen+\okuyukiy/2);

  \draw[-{Latex[reversed,length=2.5mm]}] (\ippen/2+\okuyukix/2,\ippen/2+\okuyukiy/2) -- (\ippen/2+\okuyukix/2+\okuyukix/3.5,\ippen/2+\okuyukiy/2+\okuyukiy/3.5);
  \draw[-{Latex[reversed,length=2.5mm]}] (\ippen/2+\okuyukix/2,\ippen/2+\okuyukiy/2) -- (\ippen/2+\okuyukix/2+\okuyukix/3.5+\ippen/3.5,\ippen/2+\okuyukiy/2+\okuyukiy/3.5);
  \draw[-{Latex[reversed,length=2.5mm]}] (\ippen/2+\okuyukix/2,\ippen/2+\okuyukiy/2) -- (\ippen/2+\okuyukix/2+\ippen/3.5,\ippen/2+\okuyukiy/2);
  \draw[-{Latex[reversed,length=2.5mm]}] (\ippen/2+\okuyukix/2,\ippen/2+\okuyukiy/2) -- (\ippen/2+\okuyukix/2-\okuyukix/3.5+\ippen/3.5,\ippen/2+\okuyukiy/2-\okuyukiy/3.5);
  \draw[-{Latex[reversed,length=2.5mm]}] (\ippen/2+\okuyukix/2,\ippen/2+\okuyukiy/2) -- (\ippen/2+\okuyukix/2-\okuyukix/3.5,\ippen/2+\okuyukiy/2-\okuyukiy/3.5);
  \draw[-{Latex[reversed,length=2.5mm]}] (\ippen/2+\okuyukix/2,\ippen/2+\okuyukiy/2) -- (\ippen/2+\okuyukix/2-\okuyukix/3.5-\ippen/3.5,\ippen/2+\okuyukiy/2-\okuyukiy/3.5);
  \draw[-{Latex[reversed,length=2.5mm]}] (\ippen/2+\okuyukix/2,\ippen/2+\okuyukiy/2) -- (\ippen/2+\okuyukix/2-\ippen/3.5,\ippen/2+\okuyukiy/2);
  \draw[-{Latex[reversed,length=2.5mm]}] (\ippen/2+\okuyukix/2,\ippen/2+\okuyukiy/2) -- (\ippen/2+\okuyukix/2-\ippen/3.5+\okuyukix/3.5,\ippen/2+\okuyukiy/2+\okuyukiy/3.5);

  \draw[-{Latex[length=2.5mm]}] (\ippen/2+\okuyukix/2,\ippen/2+\okuyukiy/2) -- (\ippen/2+\okuyukix/2+\okuyukix/3.5,\ippen/2+\okuyukiy/2+\okuyukiy/3.5+\ippen/3.5);
  \draw[-{Latex[length=2.5mm]}] (\ippen/2+\okuyukix/2,\ippen/2+\okuyukiy/2) -- (\ippen/2+\okuyukix/2,\ippen/2+\okuyukiy/2+\okuyukiy/3.5+\ippen/3.5);
  \draw[-{Latex[length=2.5mm]}] (\ippen/2+\okuyukix/2,\ippen/2+\okuyukiy/2) -- (\ippen/2+\okuyukix/2+\okuyukix/3.5,\ippen/2+\okuyukiy/2+\okuyukiy/3.5-\ippen/3.5);
  \draw[-{Latex[length=2.5mm]}] (\ippen/2+\okuyukix/2,\ippen/2+\okuyukiy/2) -- (\ippen/2+\okuyukix/2,\ippen/2+\okuyukiy/2-\okuyukiy/3.5-\ippen/3.5);
  \draw[-{Latex[length=2.5mm]}] (\ippen/2+\okuyukix/2,\ippen/2+\okuyukiy/2) -- (\ippen/2+\okuyukix/2-\okuyukix/3.5,\ippen/2+\okuyukiy/2-\okuyukiy/3.5-\ippen/3.5);
  \draw[-{Latex[length=2.5mm]}] (\ippen/2+\okuyukix/2,\ippen/2+\okuyukiy/2) -- (\ippen/2+\okuyukix/2-\okuyukix/3.5,\ippen/2+\okuyukiy/2-\okuyukiy/3.5+\ippen/3.5);

  \draw[-{Latex[length=2.5mm]}] (\ippen/2+\okuyukix/2,\ippen/2+\okuyukiy/2) -- (\ippen/2+\ippen/3.5+\okuyukix/2,\ippen/2+\ippen/3.5+\okuyukiy/2);
  \draw[-{Latex[length=2.5mm]}] (\ippen/2+\okuyukix/2,\ippen/2+\okuyukiy/2) -- (\ippen/2+\ippen/3.5+\okuyukix/2,\ippen/2-\ippen/3.5+\okuyukiy/2);
  \draw[-{Latex[length=2.5mm]}] (\ippen/2+\okuyukix/2,\ippen/2+\okuyukiy/2) -- (\ippen/2-\ippen/3.5+\okuyukix/2,\ippen/2-\ippen/3.5+\okuyukiy/2);
  \draw[-{Latex[length=2.5mm]}] (\ippen/2+\okuyukix/2,\ippen/2+\okuyukiy/2) -- (\ippen/2-\ippen/3.5+\okuyukix/2,\ippen/2+\ippen/3.5+\okuyukiy/2);
  
\end{tikzpicture}\;\;\;\;
\begin{tikzpicture}[line cap=round,line join=round,x=1.0cm,y=1.0cm, scale=0.3, baseline={([yshift=-.6ex]current bounding box.center)}, thick, shift={(0,0)}, scale=0.7]
  \def\ippen{12} 
  \def\okuyukix{4} 
  \def\okuyukiy{2.7} 
  \draw (0,0) rectangle ++(\ippen,\ippen);
  \draw[dashed] (\okuyukix,\okuyukiy) rectangle ++(\ippen,\ippen);
  \draw[dashed] (0,0) -- (\okuyukix,\okuyukiy);
  \draw (\ippen,0) -- (\ippen+\okuyukix,\okuyukiy);
  \draw[dashed] (0,\ippen) -- (\okuyukix,\ippen+\okuyukiy);
  \draw (\ippen,\ippen) -- (\ippen+\okuyukix,\ippen+\okuyukiy);

  \draw (\ippen/2,0) -- (\ippen/2,\ippen);
  \draw (0,\ippen/2) -- (\ippen,\ippen/2);
  \draw (0,0) -- (\ippen,\ippen);
  \draw (\ippen,0) -- (0,\ippen);

  \draw[-{Latex[length=2.5mm]}] (\ippen/2,\ippen/2) -- (\ippen/2+\ippen/3.5,\ippen/2+\ippen/3.5);
  \draw[-{Latex[length=2.5mm]}] (\ippen/2,\ippen/2) -- (\ippen/2+\ippen/3.5,\ippen/2);
  \draw[-{Latex[length=2.5mm]}] (\ippen/2,\ippen/2) -- (\ippen/2+\ippen/3.5,\ippen/2-\ippen/3.5);
  \draw[-{Latex[length=2.5mm]}] (\ippen/2,\ippen/2) -- (\ippen/2,\ippen/2-\ippen/3.5);
  \draw[-{Latex[length=2.5mm]}] (\ippen/2,\ippen/2) -- (\ippen/2-\ippen/3.5,\ippen/2-\ippen/3.5);
  \draw[-{Latex[length=2.5mm]}] (\ippen/2,\ippen/2) -- (\ippen/2-\ippen/3.5,\ippen/2);
  \draw[-{Latex[length=2.5mm]}] (\ippen/2,\ippen/2) -- (\ippen/2-\ippen/3.5,\ippen/2+\ippen/3.5);
  \draw[-{Latex[length=2.5mm]}] (\ippen/2,\ippen/2) -- (\ippen/2,\ippen/2+\ippen/3.5);

  \draw (\ippen+\okuyukix/2,0+\okuyukiy/2) -- (\ippen+\okuyukix/2,\ippen+\okuyukiy/2);
  \draw (\ippen,\ippen/2) -- (\ippen+\okuyukix,\ippen/2+\okuyukiy);
  \draw (\ippen,0) -- (\ippen+\okuyukix,\ippen+\okuyukiy);
  \draw (\ippen,\ippen) -- (\ippen+\okuyukix,\okuyukiy);

  \draw[-{Latex[length=2.5mm]}] (\ippen+\okuyukix/2,\ippen/2+\okuyukiy/2) -- (\ippen+\okuyukix/2+\okuyukix/3.5,\ippen/2+\okuyukiy/2+\okuyukiy/3.5+\ippen/3.5);
  \draw[-{Latex[length=2.5mm]}] (\ippen+\okuyukix/2,\ippen/2+\okuyukiy/2) -- (\ippen+\okuyukix/2+\okuyukix/3.5,\ippen/2+\okuyukiy/2+\okuyukiy/3.5);
  \draw[-{Latex[length=2.5mm]}] (\ippen+\okuyukix/2,\ippen/2+\okuyukiy/2) -- (\ippen+\okuyukix/2,\ippen/2+\okuyukiy/2+\okuyukiy/3.5+\ippen/3.5);
  \draw[-{Latex[length=2.5mm]}] (\ippen+\okuyukix/2,\ippen/2+\okuyukiy/2) -- (\ippen+\okuyukix/2+\okuyukix/3.5,\ippen/2+\okuyukiy/2+\okuyukiy/3.5-\ippen/3.5);
  \draw[-{Latex[length=2.5mm]}] (\ippen+\okuyukix/2,\ippen/2+\okuyukiy/2) -- (\ippen+\okuyukix/2,\ippen/2+\okuyukiy/2-\okuyukiy/3.5-\ippen/3.5);
  \draw[-{Latex[length=2.5mm]}] (\ippen+\okuyukix/2,\ippen/2+\okuyukiy/2) -- (\ippen+\okuyukix/2-\okuyukix/3.5,\ippen/2+\okuyukiy/2-\okuyukiy/3.5-\ippen/3.5);
  \draw[-{Latex[length=2.5mm]}] (\ippen+\okuyukix/2,\ippen/2+\okuyukiy/2) -- (\ippen+\okuyukix/2-\okuyukix/3.5,\ippen/2+\okuyukiy/2-\okuyukiy/3.5);
  \draw[-{Latex[length=2.5mm]}] (\ippen+\okuyukix/2,\ippen/2+\okuyukiy/2) -- (\ippen+\okuyukix/2-\okuyukix/3.5,\ippen/2+\okuyukiy/2-\okuyukiy/3.5+\ippen/3.5);

  \draw (\ippen/2,\ippen) -- (\ippen/2+\okuyukix,\ippen+\okuyukiy);
  \draw (\okuyukix/2,\ippen+\okuyukiy/2) -- (\okuyukix/2+\ippen,\ippen+\okuyukiy/2);
  \draw (0,\ippen) -- (\ippen+\okuyukix,\ippen+\okuyukiy);
  \draw (0+\okuyukix,\ippen+\okuyukiy) -- (\ippen,\ippen);

  \draw[-{Latex[length=2.5mm]}] (\ippen/2+\okuyukix/2,\ippen+\okuyukiy/2) -- (\ippen/2+\okuyukix/2+\okuyukix/3.5,\ippen+\okuyukiy/2+\okuyukiy/3.5);
  \draw[-{Latex[length=2.5mm]}] (\ippen/2+\okuyukix/2,\ippen+\okuyukiy/2) -- (\ippen/2+\okuyukix/2+\okuyukix/3.5+\ippen/3.5,\ippen+\okuyukiy/2+\okuyukiy/3.5);
  \draw[-{Latex[length=2.5mm]}] (\ippen/2+\okuyukix/2,\ippen+\okuyukiy/2) -- (\ippen/2+\okuyukix/2+\ippen/3.5,\ippen+\okuyukiy/2);
  \draw[-{Latex[length=2.5mm]}] (\ippen/2+\okuyukix/2,\ippen+\okuyukiy/2) -- (\ippen/2+\okuyukix/2-\okuyukix/3.5+\ippen/3.5,\ippen+\okuyukiy/2-\okuyukiy/3.5);
  \draw[-{Latex[length=2.5mm]}] (\ippen/2+\okuyukix/2,\ippen+\okuyukiy/2) -- (\ippen/2+\okuyukix/2-\okuyukix/3.5,\ippen+\okuyukiy/2-\okuyukiy/3.5);
  \draw[-{Latex[length=2.5mm]}] (\ippen/2+\okuyukix/2,\ippen+\okuyukiy/2) -- (\ippen/2+\okuyukix/2-\okuyukix/3.5-\ippen/3.5,\ippen+\okuyukiy/2-\okuyukiy/3.5);
  \draw[-{Latex[length=2.5mm]}] (\ippen/2+\okuyukix/2,\ippen+\okuyukiy/2) -- (\ippen/2+\okuyukix/2-\ippen/3.5,\ippen+\okuyukiy/2);
  \draw[-{Latex[length=2.5mm]}] (\ippen/2+\okuyukix/2,\ippen+\okuyukiy/2) -- (\ippen/2+\okuyukix/2-\ippen/3.5+\okuyukix/3.5,\ippen+\okuyukiy/2+\okuyukiy/3.5);

  \draw[-{Latex[length=2.5mm]}] (\ippen/2,\ippen) -- (\ippen/2+\ippen/3.5,\ippen);
  \draw[-{Latex[length=2.5mm]}] (\ippen/2,\ippen) -- (\ippen/2-\ippen/3.5,\ippen);
  \draw[-{Latex[length=2.5mm]}] (\ippen/2,0) -- (\ippen/2+\ippen/3.5,0);
  \draw[-{Latex[length=2.5mm]}] (\ippen/2,0) -- (\ippen/2-\ippen/3.5,0);
  \draw[-{Latex[length=2.5mm]}] (\ippen,\ippen/2) -- (\ippen,\ippen/2-\ippen/3.5);
  \draw[-{Latex[length=2.5mm]}] (\ippen,\ippen/2) -- (\ippen,\ippen/2+\ippen/3.5);
  \draw[-{Latex[length=2.5mm]}] (0,\ippen/2) -- (0,\ippen/2-\ippen/3.5);
  \draw[-{Latex[length=2.5mm]}] (0,\ippen/2) -- (0,\ippen/2+\ippen/3.5);

  \draw[-{Latex[length=2.5mm]}] (\ippen+\okuyukix/2,\ippen+\okuyukiy/2) -- (\ippen+\okuyukix/2+\okuyukix/3.5,\ippen+\okuyukiy/2+\okuyukiy/3.5);
  \draw[-{Latex[length=2.5mm]}] (\ippen+\okuyukix/2,\ippen+\okuyukiy/2) -- (\ippen+\okuyukix/2-\okuyukix/3.5,\ippen+\okuyukiy/2-\okuyukiy/3.5);
  \draw[-{Latex[length=2.5mm]}] (\ippen+\okuyukix/2,0+\okuyukiy/2) -- (\ippen+\okuyukix/2+\okuyukix/3.5,0+\okuyukiy/2+\okuyukiy/3.5);
  \draw[-{Latex[length=2.5mm]}] (\ippen+\okuyukix/2,0+\okuyukiy/2) -- (\ippen+\okuyukix/2-\okuyukix/3.5,0+\okuyukiy/2-\okuyukiy/3.5);

\end{tikzpicture}
\caption{
Branching structure of the triangulation\label{fig7}}
\end{figure}

\noindent
Here, to avoid complexity by incorporating all edges into a single figure, the edges were divided into three groups for illustration.

To compute the higher Berry phase, 
we need to assign the transfer MPU to each face, assign the fixed point tensors to each edge, and evaluate the quadruple inner product on each vertex. 
In the following, we explain the systematic assignment procedure.

{\bf Step 1: Transfer MPUs on faces} 
    To assign the transition MPU to each face systematically, we introduce the base points of each patch \cite{OTS23}, which are illustrated by orange dots in Fig.\ \ref{fig:base_point}. 
    By using the base points, we define transition functions on the intersection $U_{12}=U_{1}\cap U_{2}$ under the following rules:
\begin{enumerate}
    \item Take a path starting from the base point of the patch $U_{1}$ and passing through $U_{12}$ and terminating at the base point of $U_{2}$.
    \item If the path is through the surface of the cube, we give $\mathcal{O}[\{B^{ij}_Z\}]$ on the surface.
\end{enumerate}
    Under the rules, the configuration of the transition MPU of the model is illustrated in Fig.\ \ref{fig:config_of_MPU}.

\begin{figure}[h]
\begin{minipage}[t]{0.45\linewidth}
  \centering
  \includegraphics[scale=0.35]{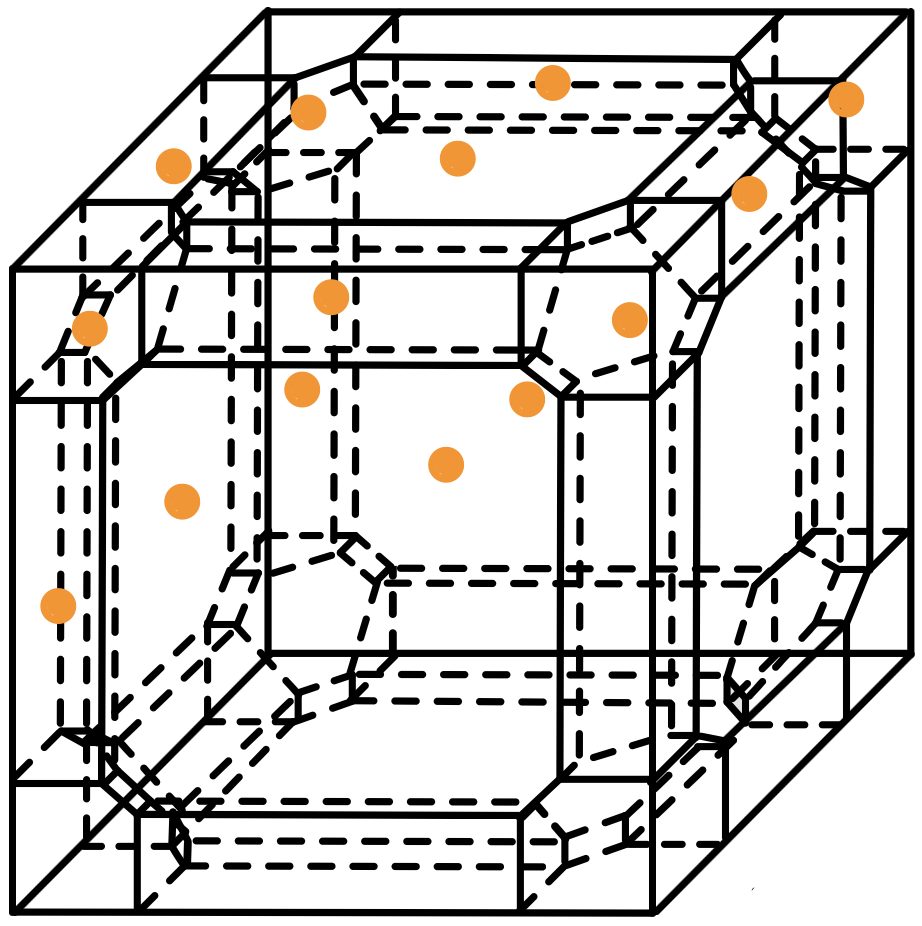}
  \hspace{0.5cm}
  \caption{The configuration of basepoints of $\rp{3}$.}
  \label{fig:base_point}
\end{minipage}
\begin{minipage}[t]{0.45\linewidth}
  \centering
  \includegraphics[scale=0.2]{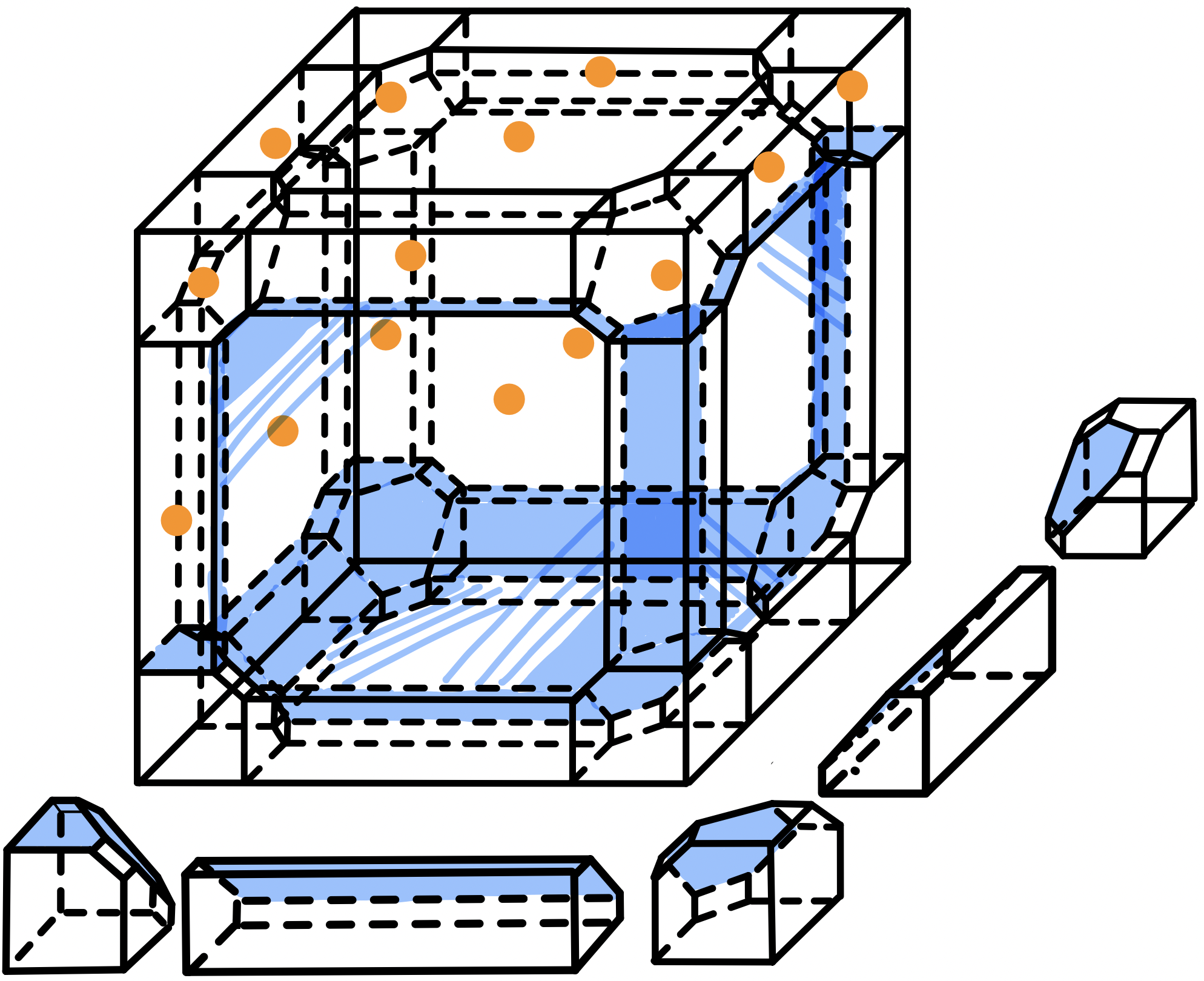}
  \hspace{0.5cm}
  \vspace{-0.4cm}
  \caption{The configuration of transition MPU.
  Blue faces represent 
  intersections for which 
  a non-trivial transition MPU is assigned.
  }
  \label{fig:config_of_MPU}
\end{minipage}
\end{figure}

{\bf Step 2: Fixed-point tensors on edges}
    A triangulation is provided by taking the dual lattice. Recall that each triangle was a triple overlap $U_{\alpha\beta\gamma}$ in the original open-covering picture. Therefore, a $3$-leg tensor, which is a fixed point of the transfer matrix (e.g. Fig.\ \ref{fig:highereigeneq}), is assigned to each triangle of the triangulation. Note that the $3$-leg tensor has two incoming legs and one outgoing leg (or the opposite assignment). To distinguish these legs, we utilize a branching structure. As shown in Fig.\ \ref{fig7}, draw the dual lattice of the triangle and rotate the direction of the edges determined by the branching structure 
    by
    $90$ degrees clockwise. 
    Depending on whether the resulting $3$-leg tensor has two incoming legs or two outgoing legs, assign $\Lambda^{R(2,1)}_{ZZI}$ or $\Lambda_{ZZI}^{L(2,1)}$, respectively. 

    In this model, the $3$-leg tensors are non-trivial only when both incoming legs are $\{B^{ij}_{Z}\}$ or both outgoing legs are $\{B^{ij}_{Z}\}$; otherwise, the $3$-leg tensors are trivial, i.e.,
    the Kronecker delta.  
    We show the edges to which non-trivial 3-leg tensors are assigned in Fig.\ \ref{fig:config_of_3_leg}.
\begin{figure}[h]
  \centering
  \includegraphics[scale=0.25]{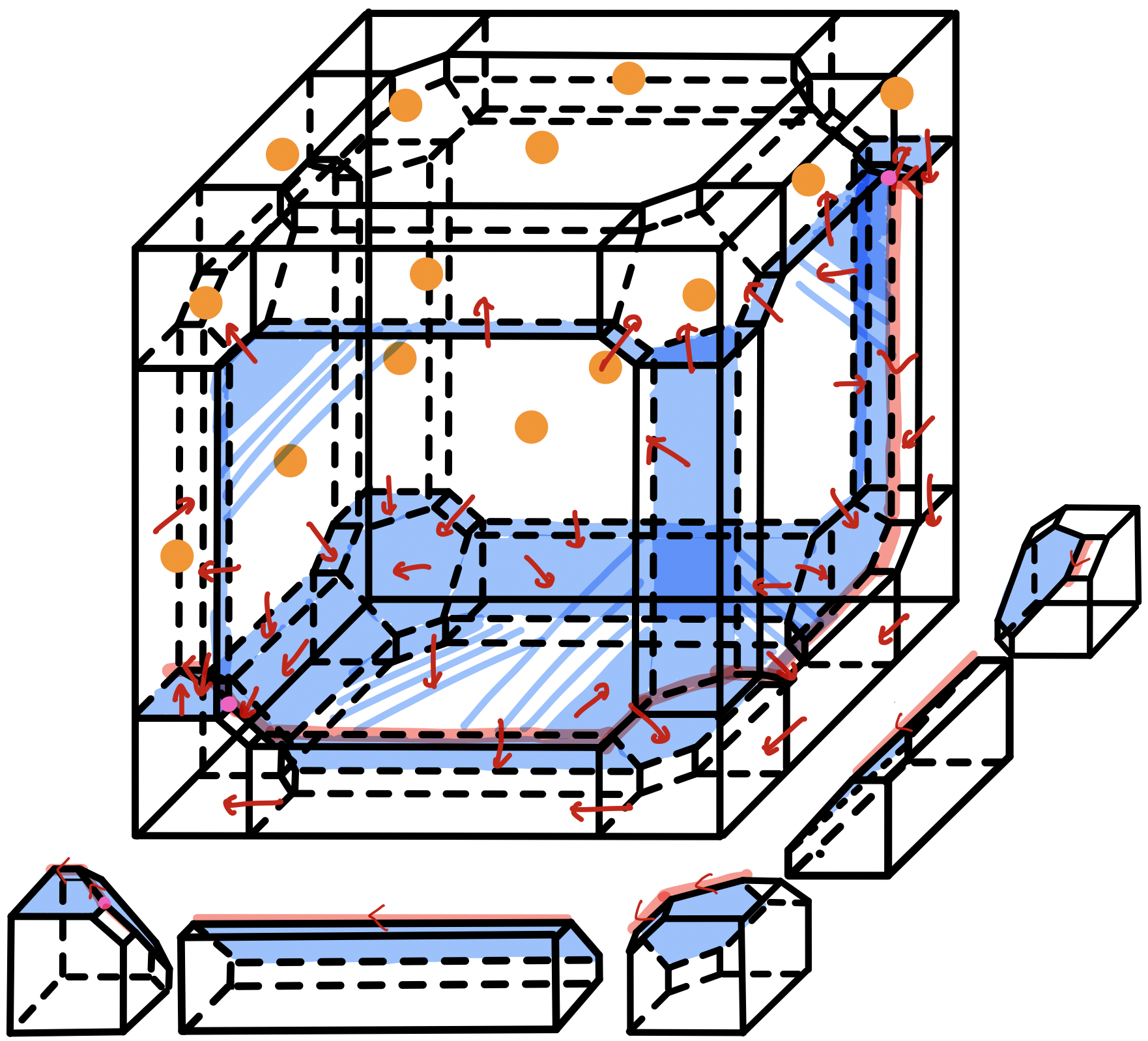}
  \hspace{0.5cm}
  \caption{The configuration of nontrivial $3$-leg tensors.
  Red edges represent
  3-intersections for which a non-trivial 3-leg tensor is assigned.
  }
  \label{fig:config_of_3_leg}
\end{figure}\\
    
    {\bf Step 3: Quadruple inner product on vertices} 
    Recall that each triangle 
    is
    a triple overlap $U_{\alpha\beta\gamma\delta}$ in the original open-covering picture. Therefore, the quadruple inner product is assigned to each tetrahedron of the triangulation. Using the procedure introduced in Step 2, a $3$-leg tensor is defined for each face of the tetrahedron, thus providing four $3$-leg tensors in total. By gluing them together naturally, we can evaluate the quadruple inner product.

As an example, we evaluate the quadruple inner product of three vertices in Fig.\ \ref{fig7}. In the triangulation, the tetrahedra, which correspond to these vertices, are
\begin{eqnarray}\label{eq:triang_tetrahedra}
\begin{tikzpicture}[line cap=round,line join=round,x=1.0cm,y=1.0cm, scale=0.3, baseline={([yshift=-.6ex]current bounding box.center)}, thick, shift={(0,0)}, scale=0.6]
  \def\ippen{12/2} 
  \def\okuyukix{6/2} 
  \def\okuyukiy{2.7/2} 
  \draw node at (\ippen/2,-\okuyukiy*2) {(a)};
  \draw[cyan!70!blue] (0,0) -- (0,\ippen);
  \draw[-{Latex[length=2.5mm]},dashed,cyan!70!blue] (0,0) -- (0,\ippen/1.75);
  
  \draw[cyan!70!blue] (0,0) -- (\okuyukix+\ippen,\okuyukiy+\ippen);
  \draw[-{Latex[length=2.5mm]},dashed,cyan!70!blue] (0,0) -- (\okuyukix/1.75+\ippen/1.75,\ippen/1.75+\okuyukiy/1.75);
  
  
  \draw[cyan!70!blue] (0+\ippen,\ippen) -- (\okuyukix+\ippen,\okuyukiy+\ippen);
  \draw[-{Latex[length=2.5mm]},dashed,cyan!70!blue] (0+\ippen,\ippen) -- (\ippen+\okuyukix/1.75+\ippen/1.75-\ippen/1.75,\ippen+\okuyukiy/1.75+\ippen/1.75-\ippen/1.75);
  
  \draw (0+\okuyukix+\ippen,\okuyukiy+\ippen) -- (0,\ippen);
  \draw[-{Latex[length=2.5mm]},dashed] (\okuyukix+\ippen,\okuyukiy+\ippen) -- (\okuyukix+\ippen-\okuyukix/1.75-\ippen/1.75,\okuyukiy+\ippen+\ippen/1.75-\okuyukiy/1.75-\ippen/1.75);

  \draw (0,0) -- (0+\ippen,\ippen);
  \draw[-{Latex[reversed,length=2.5mm]},dashed] (0,0) -- (0+\ippen/1.75,\ippen/1.75);

  \draw[cyan!70!blue] (0,\ippen) -- (\ippen,\ippen);
  \draw[-{Latex[length=2.5mm]},dashed,cyan!70!blue] (0,\ippen) -- (\ippen/1.75,\ippen+\ippen/1.75-\ippen/1.75);
\end{tikzpicture}
\, ,
\quad
\begin{tikzpicture}[line cap=round,line join=round,x=1.0cm,y=1.0cm, scale=0.3, baseline={([yshift=-.6ex]current bounding box.center)}, thick, shift={(0,0)}, scale=0.6]
  \def\ippen{12/2} 
  \def\okuyukix{6/2} 
  \def\okuyukiy{2.7/2} 
  \draw node at (\ippen/2,-\okuyukiy*2) {(b)};
  \draw[dashed,cyan!70!blue] (0,0) -- (\okuyukix+\ippen,\okuyukiy+\ippen);
  \draw[-{Latex[reversed,length=2.5mm]},dashed,cyan!70!blue] (0,0) -- (\okuyukix/1.75+\ippen/1.75,\ippen/1.75+\okuyukiy/1.75);
  
  \draw (\ippen,0) -- (0+\ippen,\ippen);
  \draw[-{Latex[length=2.5mm]}] (\ippen,0) -- (\ippen+\ippen/1.75-\ippen/1.75,\ippen/1.75);
  
  \draw[cyan!70!blue] (0+\ippen,\ippen) -- (\okuyukix+\ippen,\okuyukiy+\ippen);
  \draw[-{Latex[length=2.5mm]},dashed,cyan!70!blue] (0+\ippen,\ippen) -- (\ippen+\okuyukix/1.75+\ippen/1.75-\ippen/1.75,\ippen+\okuyukiy/1.75+\ippen/1.75-\ippen/1.75);
  
  \draw[cyan!70!blue] (0+\okuyukix+\ippen,\okuyukiy+\ippen) -- (\ippen,0);
  \draw[-{Latex[length=2.5mm]},dashed,cyan!70!blue] (\okuyukix+\ippen,\okuyukiy+\ippen) -- (\okuyukix+\ippen+\ippen/1.75-\okuyukix/1.75-\ippen/1.75,\okuyukiy+\ippen-\okuyukiy/1.75-\ippen/1.75);

  \draw (0,0) -- (0+\ippen,\ippen);
  \draw[-{Latex[length=2.5mm]},dashed] (0,0) -- (0+\ippen/1.75,\ippen/1.75);

  \draw (0,0) -- (\ippen,0);
  \draw[-{Latex[reversed,length=2.5mm]},dashed] (0,0) -- (\ippen/1.75,0);
\end{tikzpicture}
\, ,
\;
\quad
\begin{tikzpicture}[line cap=round,line join=round,x=1.0cm,y=1.0cm, scale=0.3, baseline={([yshift=-.6ex]current bounding box.center)}, thick, shift={(0,0)}, scale=0.6]
  \def\ippen{12/2} 
  \def\okuyukix{6/2} 
  \def\okuyukiy{2.7/2} 
  \draw node at (\ippen/2,-\okuyukiy*2) {(c)};
  \draw[dashed,cyan!70!blue] (0,0) -- (\okuyukix,\okuyukiy+\ippen);
  \draw[-{Latex[reversed,length=2.5mm]},dashed,cyan!70!blue] (0,0) -- (\okuyukix/1.75,\ippen/1.75+\okuyukiy/1.75);
  
  \draw (\ippen,0) -- (0,\ippen);
  \draw[-{Latex[length=2.5mm]}] (\ippen,0) -- (\ippen-\ippen/1.75,\ippen/1.75);
  
  \draw[cyan!70!blue] (0,\ippen) -- (\okuyukix,\okuyukiy+\ippen);
  \draw[-{Latex[length=2.5mm]},dashed,cyan!70!blue] (0,\ippen) -- (\okuyukix/1.75,\ippen+\okuyukiy/1.75+\ippen/1.75-\ippen/1.75);
  
  \draw[cyan!70!blue] (0+\okuyukix,\okuyukiy+\ippen) -- (\ippen,0);
  \draw[-{Latex[length=2.5mm]},dashed,cyan!70!blue] (\okuyukix,\okuyukiy+\ippen) -- (\okuyukix+\ippen/1.75-\okuyukix/1.75,\okuyukiy+\ippen-\okuyukiy/1.75-\ippen/1.75);

  \draw (0,0) -- (0,\ippen);
  \draw[-{Latex[reversed,length=2.5mm]},dashed] (0,0) -- (0,\ippen/1.75);

  \draw (0,0) -- (\ippen,0);
  \draw[-{Latex[length=2.5mm]},dashed] (0,0) -- (\ippen/1.75,0);
\end{tikzpicture},
\quad
\begin{tikzpicture}[line cap=round,line join=round,x=1.0cm,y=1.0cm, scale=0.3, baseline={([yshift=-.6ex]current bounding box.center)}, thick, shift={(0,0)}, scale=0.6]
  \def\ippen{12/2} 
  \def\okuyukix{6/2} 
  \def\okuyukiy{2.7/2} 
  \draw node at (\ippen/2,-\okuyukiy*2) {(d)};
  \draw[cyan!70!blue] (0,0) -- (\ippen,0);
  \draw[-{Latex[reversed,length=2.5mm]},cyan!70!blue] (0,0) -- (\ippen/1.75,0);
  
  \draw[cyan!70!blue] (\ippen,0) -- (\ippen+\okuyukix,\ippen+\okuyukiy);
  \draw[-{Latex[length=2.5mm]},cyan!70!blue] (\ippen,0) -- (\ippen+\ippen/1.75+\okuyukix/1.75-\ippen/1.75,\ippen/1.75+\okuyukiy/1.75);

  \draw (0,0) -- (\ippen+\okuyukix,\ippen+\okuyukiy);
  \draw[-{Latex[length=2.5mm]}] (0,0) -- (\ippen/1.75+\okuyukix/1.75,\ippen/1.75+\okuyukiy/1.75);

  \draw[dashed,cyan!70!blue] (0,0) -- (\ippen+\okuyukix,\okuyukiy);
  \draw[-{Latex[reversed,length=2.5mm]},dashed,cyan!70!blue] (0,0) -- (\ippen/1.75+\okuyukix/1.75,\okuyukiy/1.75);

  \draw (\ippen,0) -- (\ippen+\okuyukix,\okuyukiy);
  \draw[-{Latex[length=2.5mm]}] (\ippen,0) -- (\ippen+\ippen/1.75+\okuyukix/1.75-\ippen/1.75,\okuyukiy/1.75);

  \draw[cyan!70!blue] (\ippen+\okuyukix,\okuyukiy) -- (\ippen+\okuyukix,\okuyukiy+\ippen);
  \draw[-{Latex[length=2.5mm]},cyan!70!blue] (\ippen+\okuyukix,\okuyukiy) -- (\ippen+\okuyukix+\ippen/1.75+\okuyukix/1.75-\ippen/1.75-\okuyukix/1.75,\okuyukiy+\okuyukiy/1.75+\ippen/1.75-\okuyukiy/1.75);
  
\end{tikzpicture}
\, ,
\end{eqnarray}
where the edges with blue color imply that the transition MPU assigned to the edges are nontrivial. 
Let's compute the tetrahedra (b) in Eq.\ \eqref{eq:triang_tetrahedra}. For instance, focusing on the right-side triangle, according to the procedure explained in Step 2, the $3$-leg tensor with 2-in 1-out configuration is assigned to this face, where the incoming legs are labeled with nontrivial elements of $\zmod{2}$. Similarly, focusing on the left-side triangle, the $3$-leg tensor with 1-in 2-out configuration is assigned to this face, where the outgoing legs are labeled with nontrivial elements of $\zmod{2}$. We can easily check that the trivial $3$-leg tensors are assigned to the other edges. 
When these $3$-leg tensors are glued along the edges, the diagram for the quadruple inner product is 
\begin{eqnarray}
    \begin{tikzpicture}[line cap=round,line join=round,x=1.0cm,y=1.0cm, scale=0.3, baseline={([yshift=-.6ex]current bounding box.center)}, thick, shift={(0,0)}, scale=0.7]
  \def\ippen{14/2} 
  \def\okuyukix{8/2} 
  \def\okuyukiy{-4/2} 
  \def\takasa{8/2}
  \draw[dashed,cyan!70!blue] (0,0) -- (\ippen,0);
  \draw[-{Latex[length=2.5mm]},dashed,cyan!70!blue] (0,0) -- (\ippen/1.75,0);
  
  \draw (0,0) -- (\okuyukix,\okuyukiy);
  \draw[-{Latex[reversed,length=2.5mm]}] (0,0) -- (\okuyukix/1.75,\okuyukiy/1.75);
  
  \draw (\okuyukix,\okuyukiy) -- (\ippen,0);
  \draw[-{Latex[reversed,length=2.5mm]}] (\okuyukix,\okuyukiy) -- (\okuyukix+\ippen/1.75-\okuyukix/1.75,\okuyukiy-\okuyukiy/1.75);

  \draw[cyan!70!blue] (0,0) -- (\okuyukix,\takasa);
  \draw[-{Latex[length=2.5mm]},cyan!70!blue] (0,0) -- (\okuyukix/1.75,\takasa/1.75);
  
  \draw (\okuyukix,\okuyukiy) -- (\okuyukix,\takasa);
  \draw[-{Latex[length=2.5mm]}] (\okuyukix,\okuyukiy) -- (\okuyukix+\okuyukix/1.75-\okuyukix/1.75,\okuyukiy+\takasa/1.75-\okuyukiy/1.75);
  
  \draw[cyan!70!blue] (\okuyukix,\takasa) -- (\ippen,0);
  \draw[-{Latex[length=2.5mm]},cyan!70!blue] (\okuyukix,\takasa) -- (\okuyukix+\ippen/1.75-\okuyukix/1.75,\takasa-\takasa/1.75);
\end{tikzpicture}=\sum_{ab}(\Lambda_{ZZI}^{R(2,1)})_{ab}{}^{0}(\Lambda_{ZZI}^{L(2,1)})_{0}{}^{ab}=1.
\end{eqnarray}
Generally, as with the tetrahedron (b) and (c) 
in Eq.\ \eqref{eq:triang_tetrahedra}, 
the quadruple inner product turns out to be
trivial specifically in cases where exactly three edges have non-trivial transition functions. It can be confirmed that among the tetrahedra included in the triangulation, only the ones labeled (a) and (d) have four nontrivial edges. Therefore, it is sufficient to evaluate the quadruple inner products corresponding to these two tetrahedra. For 
the tetrahedron
(d) in Eq.\ \eqref{eq:triang_tetrahedra}, 
the quadruple inner product is 
\begin{eqnarray}
    \begin{tikzpicture}[line cap=round,line join=round,x=1.0cm,y=1.0cm, scale=0.3, baseline={([yshift=-.6ex]current bounding box.center)}, thick, shift={(0,0)}, scale=0.7]
  \def\ippen{14/2} 
  \def\okuyukix{8/2} 
  \def\okuyukiy{-4/2} 
  \def\takasa{8/2}
  \draw[dashed,cyan!70!blue] (0,0) -- (\ippen,0);
  \draw[-{Latex[length=2.5mm]},dashed,cyan!70!blue] (0,0) -- (\ippen/1.75,0);
  
  \draw (0,0) -- (\okuyukix,\okuyukiy);
  \draw[-{Latex[reversed,length=2.5mm]}] (0,0) -- (\okuyukix/1.75,\okuyukiy/1.75);
  
  \draw[cyan!70!blue] (\okuyukix,\okuyukiy) -- (\ippen,0);
  \draw[-{Latex[reversed,length=2.5mm]},cyan!70!blue] (\okuyukix,\okuyukiy) -- (\okuyukix+\ippen/1.75-\okuyukix/1.75,\okuyukiy-\okuyukiy/1.75);

  \draw[cyan!70!blue] (0,0) -- (\okuyukix,\takasa);
  \draw[-{Latex[length=2.5mm]},cyan!70!blue] (0,0) -- (\okuyukix/1.75,\takasa/1.75);
  
  \draw[cyan!70!blue] (\okuyukix,\okuyukiy) -- (\okuyukix,\takasa);
  \draw[-{Latex[reversed,length=2.5mm]},cyan!70!blue] (\okuyukix,\okuyukiy) -- (\okuyukix+\okuyukix/1.75-\okuyukix/1.75,\okuyukiy+\takasa/1.75-\okuyukiy/1.75);
  
  \draw (\okuyukix,\takasa) -- (\ippen,0);
  \draw[-{Latex[length=2.5mm]}] (\okuyukix,\takasa) -- (\okuyukix+\ippen/1.75-\okuyukix/1.75,\takasa-\takasa/1.75);
\end{tikzpicture}=\sum_{ab}(\Lambda_{ZZI}^{R(2,1)})_{ab}{}^{0}(\Lambda_{ZZI}^{L(2,1)})_0{}^{ab}=1.
\end{eqnarray}
On the other hand, for the tetrahedron
(a) in Eq.\ \eqref{eq:triang_tetrahedra}, 
the quadruple inner product is 
\begin{eqnarray}
    \begin{tikzpicture}[line cap=round,line join=round,x=1.0cm,y=1.0cm, scale=0.3, baseline={([yshift=-.6ex]current bounding box.center)}, thick, shift={(0,0)}, scale=0.7]
  \def\ippen{14/2} 
  \def\okuyukix{8/2} 
  \def\okuyukiy{-4/2} 
  \def\takasa{8/2}
  \draw[dashed,cyan!70!blue] (0,0) -- (\ippen,0);
  \draw[-{Latex[length=2.5mm]},dashed,cyan!70!blue] (0,0) -- (\ippen/1.75,0);
  
  \draw (0,0) -- (\okuyukix,\okuyukiy);
  \draw[-{Latex[reversed,length=2.5mm]}] (0,0) -- (\okuyukix/1.75,\okuyukiy/1.75);
  
  \draw[cyan!70!blue] (\okuyukix,\okuyukiy) -- (\ippen,0);
  \draw[-{Latex[length=2.5mm]},cyan!70!blue] (\okuyukix,\okuyukiy) -- (\okuyukix+\ippen/1.75-\okuyukix/1.75,\okuyukiy-\okuyukiy/1.75);

  \draw[cyan!70!blue] (0,0) -- (\okuyukix,\takasa);
  \draw[-{Latex[length=2.5mm]},cyan!70!blue] (0,0) -- (\okuyukix/1.75,\takasa/1.75);
  
  \draw[cyan!70!blue] (\okuyukix,\okuyukiy) -- (\okuyukix,\takasa);
  \draw[-{Latex[reversed,length=2.5mm]},cyan!70!blue] (\okuyukix,\okuyukiy) -- (\okuyukix+\okuyukix/1.75-\okuyukix/1.75,\okuyukiy+\takasa/1.75-\okuyukiy/1.75);
  
  \draw (\okuyukix,\takasa) -- (\ippen,0);
  \draw[-{Latex[length=2.5mm]}] (\okuyukix,\takasa) -- (\okuyukix+\ippen/1.75-\okuyukix/1.75,\takasa-\takasa/1.75);
\end{tikzpicture}=\sum_{ab}(\Lambda_{ZZI}^{R(2,1)})_{ab}{}^{0}(\Lambda_{ZZI}^{L(2,1)})_{0}{}^{ba}=-1. 
\end{eqnarray}
Therefore, the total higher Berry phase is $-1$ in $\cohoZ{4}{\rp{4}}\simeq \zmod{2}$, and this implies the nontriviality as a family.

\section{Geometric Interpretation: $2$-Gerbe}\label{sec:2gerbe}

In this section, we provide a further discussion
on our construction of the higher Berry phase and the topological invariant,
from the perspective of a gerbe. 
In particular,
we provide a geometric interpretation, called a $2$-gerbe \cite{Breen94}, of the family of semi-injective PEPS. 
We recall that the underlying gerbe structure of the higher Berry phase for (1+1)d many-body
states (MPS) are discussed in Refs.\ \cite{OR23} and \cite{qi2023charting}. 
To address the nontriviality of a family of PEPSs,
we need to consider the interface of the interface and put the zipper tensor on
it.
The zipper tensor can be regarded as the redundancy of the transition MPU, and the transition MPU is the redundancy of the PEPS tensor. In this sense, the zipper tensor is the redundancy of the redundancy. Now, $2$-gerbe is a sophisticated tool to handle the redundancy of the redundancy\footnote{In general, $n$-category has information about the $n$th-redundancy of objects.}.
We note that Ref.\ \cite{qi2023charting} 
points out a 2-gerbe structure for (2+1)-dimensional PEPS using the bundle gerbe.

A $1$-gerbe, or simply a gerbe, is a categorification of a line bundle, and a $2$-gerbe is a categorification of a $1$-gerbe. Therefore, we will first recall a line bundle and explain the relationship between line bundles and gerbes.
Consider a $(0+1)$-dimensional unique gapped Hamiltonian parametrized by $X$. We fix an open covering $\{U_\alpha\}$ of $X$. We can take a family of normalized ground states $\ket{\psi_{\alpha}(x)}$ over each open set $U_\alpha$. Note that the phase of the ground state can freely be changed. In other words, to take a family of the ground states $\ket{\psi_{\alpha}(x)}$ can be regarded as a gauge fixing of the phase redundancy of the states. At each point of the intersection $U_{\alpha\beta}$, we have two identical states in different gauges. Thus, in particular, it is possible to determine the phase difference between the two states, $\braket{\psi_{\alpha}(x)|\psi_{\beta}(x)}$. This is a physical 
realization
of a line bundle.

Next, consider a $(1+1)$-dimensional unique gapped Hamiltonian parametrized by $X$. We assume that we can take a normal MPS representation $\{A^i(x)\}$ of the ground state at each point of $X$. As in the $(0+1)$ dimensional case, this involves gauge fixing with respect to the redundancy of the MPS representation. Using the fundamental theorem of MPS, at each point of the two intersection $U_{\alpha\beta}$, we can take the unique largest eigenvector (or reduction tensor) $\Lambda_{\alpha\beta}$ of the mixed transfer matrix, and at the three intersection $U_{\alpha\beta\gamma}$, a triple product can be taken to obtain a ${\rm U}(1)$ function. That is, a line bundle spanned by $\Lambda_{\alpha\beta}$ is defined on the $2$-intersection, and we can extract 
a $\mathrm{U}(1)$-valued function on the three intersections. In the case of $(0+1)$ dimensions, a family of one-dimensional vector spaces is defined on each open set, and a ${\rm U}(1)$-valued function is provided at the two intersections. In contrast, in the case of $(1+1)$ dimensions, a family of one-dimensional vector spaces is established at the two intersections, and a ${\rm U}(1)$-valued function is provided at the three intersections. In this sense, the degree of intersection has increased by one. Such objects are mathematically referred to as ${\rm U}(1)$-gerbes. The type of gerbe obtained by considering a family of normal MPS is called an MPS gerbe \cite{OR23, QSWSPBH23}.

Let's now consider a $(2+1)$-dimensional unique gapped Hamiltonian parametrized by $X$. We assume that we can take a semi-injective PEPS representation $\{A_{\alpha}^i(x)\}$ of the ground state at each point of $U_\alpha$. Using the fundamental theorem of semi-injective, at each point of the two intersection $U_{\alpha\beta}$, we obtain an injective MPO $\mathcal{O}_{\alpha\beta}$ between $\{A_{\alpha}^i(x)\}$ and $\{A_{\beta}^i(x)\}$. If we assume the unitarity of the MPO, on the three intersection $U_{\alpha\beta\gamma}$, we have a unique fixed point $\Lambda_{\alpha\beta\gamma}$ of transfer matrix $T_{\alpha\beta\gamma}$, and we can extract a $\mathrm{U}(1)$-valued function on the four intersection. That is, a line bundle spanned by $\Lambda_{\alpha\beta\gamma}$ is defined on the three intersection, and we can extract a $\mathrm{U}(1)$-valued function on the four intersection. Such objects are mathematically referred to as $\mathrm{Tor}(\mathrm{U}(1))$-gerbes \cite{Breen94}. However, to extract the $\mathrm{U}(1)$-valued function, we need to use the information of the transition MPU in general, and in such a case, this is not a $2$-gerbe. 
In the case of the example in Sec.\ \ref{sec:Model parameterized over rp4}, 
we need not to insert the transition MPU. 
Therefore, in this case, $2$-gerbe over $\rp{4}$ is constructed from a family of semi-injective PEPS.

If the zipper condition holds, we do not need to insert the transition MPU. Thus, the zipper condition is a sufficient condition for the $2$-gerbe description. 
%
More abstractly, a $2$-gerbe over $X$ is a $2$-stack over $X$ such that every $1$-morphism is invertible up to $2$-morphism and every $2$-morphism is exactly invertible \cite{Breen94}. 
In the PEPS description, $1$-morphisms correspond to the transition MPU $\mathcal{O}_{\alpha\beta}$ and $2$-morphisms 
correspond to the fixed point tensors $\Lambda_{\alpha\beta\gamma}$. 
Since $\mathcal{O}_{\alpha\beta}$ is injective, we can take the inverse $\mathcal{O}^{-1}_{\alpha\beta}$ and reduction pairs from $\mathcal{O}_{\alpha\beta}\cdot\mathcal{O}^{-1}_{\alpha\beta}$ to the identity and from $\mathcal{O}^{-1}_{\alpha\beta}\cdot\mathcal{O}_{\alpha\beta}$ to the identity. 
This is the invertibility of $1$-morphisms up to $2$-morphisms. 
However, the invertibility of the $2$-morphism does not hold in general. 
For example, let's consider the right reduction $\Lambda^{R}$ from $\{A^i\}$ to $\{B^i\}$. 
Here, $\{A^i\}$ is an MPO tensor that is not necessarily normal and $\{B^i\}$ is a normal MPO tensor. 
The invertibility of $\Lambda^{R}$ implies that there is a tensor $X$ such that $X\cdot\Lambda^{R}=1_{B}$ and $\Lambda^{R}\cdot X=1_{A}$. 
Here, $1_A$ and $1_B$ are the identity operators acting on the virtual legs of $\{A^i\}$ and $\{B^i\}$. The left reduction $\Lambda^L$ is a good candidate for the inverse tensor. Indeed, $\Lambda^L$ satisfies the first condition $\Lambda^L \cdot\Lambda^R=1_{B}$. However, $\Lambda^{R}\cdot\Lambda^R$ is not the identity. Diagrammatically, when we consider the interface of the $\{B^i\}$, $\{A^i\}$ and $\{B^i\}$, we can zip the reduction tensors as follows:
\begin{eqnarray}
    \begin{tikzpicture}[line cap=round,line join=round,x=1.0cm,y=1.0cm, scale=0.3, baseline={([yshift=-.6ex]current bounding box.center)}, thick, shift={(0,0)}, scale=0.7]
  \def\saitosu{4} 
  \pgfmathsetmacro\saitosutasuichi{\saitosu+1}
  \setcounter{n}{0}
  \setcounter{n}{\saitosu} 
  \def\yoko{8} 
  \def\tate{3} 
  \draw (0,\tate/2) -- (\yoko,\tate/2);
  \fill[color=black] (\yoko,\tate/2) circle (3mm); 
  \draw (\yoko,\tate/2) node [below] {$\Lambda^{L}$};
  \draw (\yoko/2,0) node [below] {$B^i$};
  \draw (\yoko/2,\tate) node [above] {$\textcolor{white}{B^i}$};
  \foreach \x in {1,...,\then}
      \draw (\x*\yoko/\saitosutasuichi,0) -- (\x*\yoko/\saitosutasuichi,\tate); 
  \foreach \x in {1,...,\then}
      \draw[fill=pink!70!red] (\x*\yoko/\saitosutasuichi,\tate/2) circle (3mm);
  \draw (0+\yoko,\tate/2) -- (\yoko+\yoko,\tate/2);
  \fill[color=black] (\yoko+\yoko,\tate/2) circle (3mm); 
  \draw (\yoko+\yoko,\tate/2) node [below] {$\Lambda^{R}$};
  \draw (\yoko/2+\yoko,0) node [below] {$A^i$};
  \foreach \x in {1,...,\then}
      \draw (\x*\yoko/\saitosutasuichi+\yoko,0) -- (\x*\yoko/\saitosutasuichi+\yoko,\tate); 
  \foreach \x in {1,...,\then}
      \draw[fill=cyan!70!blue] (\x*\yoko/\saitosutasuichi+\yoko,\tate/2) circle (3mm);
  \draw (0+\yoko*2,\tate/2) -- (\yoko+\yoko*2,\tate/2);
  \draw (\yoko/2+\yoko*2,0) node [below] {$B^i$};
  \foreach \x in {1,...,\then}
      \draw (\x*\yoko/\saitosutasuichi+\yoko*2,0) -- (\x*\yoko/\saitosutasuichi+\yoko*2,\tate); 
  \foreach \x in {1,...,\then}
      \draw[fill=pink!70!red] (\x*\yoko/\saitosutasuichi+\yoko*2,\tate/2) circle (3mm);
\end{tikzpicture}
=
\begin{tikzpicture}[line cap=round,line join=round,x=1.0cm,y=1.0cm, scale=0.3, baseline={([yshift=-.6ex]current bounding box.center)}, thick, shift={(0,0)}, scale=0.7]
  \def\saitosu{12} 
  \pgfmathsetmacro\saitosutasuichi{\saitosu+1}
  \setcounter{n}{0}
  \setcounter{n}{\saitosu} 
  \def\yoko{8*3} 
  \def\tate{3} 
  \draw (0,\tate/2) -- (\yoko,\tate/2);
  \draw (\yoko/2,0) node [below] {$B^i$};
  \draw (\yoko/2,\tate) node [above] {$\textcolor{white}{B^i}$};
  \foreach \x in {1,...,\then}
      \draw (\x*\yoko/\saitosutasuichi,0) -- (\x*\yoko/\saitosutasuichi,\tate); 
  \foreach \x in {1,...,\then}
      \draw[fill=pink!70!red] (\x*\yoko/\saitosutasuichi,\tate/2) circle (3mm);
\end{tikzpicture}.
\end{eqnarray}
On the other hand, when we consider the interface of $\{A^i\}$, $\{B^i\}$ and $\{A^i\}$, we can unzip the reduction tensors but there remains a tensor $\Lambda^R\Lambda^L$:
\begin{eqnarray}
\begin{tikzpicture}[line cap=round,line join=round,x=1.0cm,y=1.0cm, scale=0.3, baseline={([yshift=-.6ex]current bounding box.center)}, thick, shift={(0,0)}, scale=0.7]
  \def\saitosu{4} 
  \pgfmathsetmacro\saitosutasuichi{\saitosu+1}
  \setcounter{n}{0}
  \setcounter{n}{\saitosu} 
  \def\yoko{8} 
  \def\tate{3} 
  \draw (0,\tate/2) -- (\yoko,\tate/2);
  \fill[color=black] (\yoko,\tate/2) circle (3mm); 
  \draw (\yoko,\tate/2) node [below] {$\Lambda^{R}$};
  \draw (\yoko/2,0) node [below] {$A^i$};
  \draw (\yoko/2,\tate) node [above] {$\textcolor{white}{B^i}$};
  \foreach \x in {1,...,\then}
      \draw (\x*\yoko/\saitosutasuichi,0) -- (\x*\yoko/\saitosutasuichi,\tate); 
  \foreach \x in {1,...,\then}
      \draw[fill=cyan!70!blue] (\x*\yoko/\saitosutasuichi,\tate/2) circle (3mm);
  \draw (0+\yoko,\tate/2) -- (\yoko+\yoko,\tate/2);
  \fill[color=black] (\yoko+\yoko,\tate/2) circle (3mm); 
  \draw (\yoko+\yoko,\tate/2) node [below] {$\Lambda^{L}$};
  \draw (\yoko/2+\yoko,0) node [below] {$B^i$};
  \foreach \x in {1,...,\then}
      \draw (\x*\yoko/\saitosutasuichi+\yoko,0) -- (\x*\yoko/\saitosutasuichi+\yoko,\tate); 
  \foreach \x in {1,...,\then}
      \draw[fill=pink!70!red] (\x*\yoko/\saitosutasuichi+\yoko,\tate/2) circle (3mm);
  \draw (0+\yoko*2,\tate/2) -- (\yoko+\yoko*2,\tate/2);
  \draw (\yoko/2+\yoko*2,0) node [below] {$A^i$};
  \foreach \x in {1,...,\then}
      \draw (\x*\yoko/\saitosutasuichi+\yoko*2,0) -- (\x*\yoko/\saitosutasuichi+\yoko*2,\tate); 
  \foreach \x in {1,...,\then}
      \draw[fill=cyan!70!blue] (\x*\yoko/\saitosutasuichi+\yoko*2,\tate/2) circle (3mm);
\end{tikzpicture}
=
\begin{tikzpicture}[line cap=round,line join=round,x=1.0cm,y=1.0cm, scale=0.3, baseline={([yshift=-.6ex]current bounding box.center)}, thick, shift={(0,0)}, scale=0.7]
  \def\saitosu{6} 
  \pgfmathsetmacro\saitosutasuichi{\saitosu+1}
  \setcounter{n}{0}
  \setcounter{n}{\saitosu} 
  \def\yoko{12} 
  \def\tate{3} 
  \draw (0,\tate/2) -- (\yoko,\tate/2);
  \fill[color=black] (\yoko,\tate/2) circle (3mm); 
  \draw (\yoko,\tate/2-1.5
  ) node [below] {$\Lambda^{R}\Lambda^{L}$};
  \draw (\yoko/2,0) node [below] {$A^i$};
  \draw (\yoko/2,\tate) node [above] {$\textcolor{white}{B^i}$};
  \foreach \x in {1,...,\then}
      \draw (\x*\yoko/\saitosutasuichi,0) -- (\x*\yoko/\saitosutasuichi,\tate); 
  \foreach \x in {1,...,\then}
      \draw[fill=cyan!70!blue] (\x*\yoko/\saitosutasuichi,\tate/2) circle (3mm);
  \draw (0+\yoko,\tate/2) -- (\yoko+\yoko,\tate/2);
  \draw (\yoko/2+\yoko,0) node [below] {$A^i$};
  \foreach \x in {1,...,\then}
      \draw (\x*\yoko/\saitosutasuichi+\yoko,0) -- (\x*\yoko/\saitosutasuichi+\yoko,\tate); 
  \foreach \x in {1,...,\then}
      \draw[fill=cyan!70!blue] (\x*\yoko/\saitosutasuichi+\yoko,\tate/2) circle (3mm);
\end{tikzpicture}
\end{eqnarray}
Since $\Lambda^R\Lambda^L$ is a projection that is supported on the normal part of $\{A^i\}$, we can annihilate it when we take a periodic boundary condition \cite{G-RLM23}. If the zipper condition holds exactly 
we can annihilate the projection freely. In this case, a $2$-gerbe over the parameter space is constructed from a family of semi-injective PEPS. This is consistent with the above observation.

\section{Summary and Discussion} \label{sec:summary}

In this paper, we provide the construction of the higher-Berry phase
using the semi-injective PEPS representation of (2+1)-dimensional invertible many-body quantum states.
There are many open questions to be explored. We close by listing some of them.

First and foremost, 
it is interesting to find many more examples and apply our formula. 
To this end, it is desirable to construct 
many more simple toy models. 
In this paper, we constructed 
the (2+1)d model characterized 
by the $\mathbb{Z}/2\mathbb{Z}$ higher-Berry phase.
It is interesting to construct and study 
some models that are characterized by 
the $\mathbb{Z}$ topological invariant
(the free rather than torsion part of $\cohoZ{4}{X}$).
Furthermore, it is interesting to study more "realistic" models. 
To this end,
it is important to consider a wider class of tensor networks, beyond semi-injective PEPS. 
We need to see if our formalism and formula work for generic tensor networks. 
It is also important to implement 
our formalism numerically.

There are also other fundamental aspects to explore. For example, 
while we constructed 
a parameterized model 
realizing 
a 2-gerbe 
over $X=\mathbb{R}P^4$,
for generic cases,
our construction does not fit 
completely into a 2-gerbe structure. While this is not required 
for the discussion of topological classifications and the higher Berry phase, 
we speculate 
it is interesting to explore
other settings, e.g., using other 
tensor networks or 
other implementation of normality. 
Also, it would be nice to develop quantum field theory descriptions of the higher-Berry phase. 
It is also interesting to make a connection with
higher structures explored in other areas, 
higher-energy physics and mathematics,
in particular 
(see, for example, \cite{borsten2024higher}).



\begin{acknowledgments}
  We would like to thank Takamasa Ando, Atsushi Ueda,
  Andr\'{a}s Moln\'{a}r, Norbert Schuch, Ken Shiozaki, 
  Yuji Terashima, Alex Turzillo and Mayuko Yamashita for valuable
  discussions.
  We would also like to thank Yuya Kusuki and Bowei Liu for 
  collaboration in a related project.
  Special thanks go to Ophelia Evelyn Sommer,
  Xueda Wen, and Ashvin Vishwanath, who let us know of their related works, and also kindly agreed to coordinate submissions of our papers to arXiv.  
  This research was supported in part by Perimeter Institute for Theoretical Physics.
  Research at Perimeter Institute is supported by the Government of Canada through the Department of Innovation, Science and Economic Development and by the Province of Ontario through the Ministry of Colleges and Universities. 
We are grateful to the long-term workshop YITP-T-23-01 held at YITP, Kyoto University, where a part of this work was done.
S.O. is supported by JSPS KAKENHI Grant Number 23KJ1252 and 24K00522.
S.R.~is supported
by a Simons Investigator Grant from
the Simons Foundation (Award No.~566116).
This work is supported by
the Gordon and Betty Moore Foundation through Grant
GBMF8685 toward the Princeton theory program. 

\end{acknowledgments}



\bibliography{parentbib.bib}

\end{document}